\documentclass{article}
\usepackage{fullpage}

\usepackage[hidelinks, colorlinks=true]{hyperref}
\hypersetup{
    colorlinks=true,
    linkcolor=blue,
    filecolor=magenta,      
    urlcolor=cyan,
}

\title{Quantum algorithms for path and cycle containment problems}
\date{August 5, 2026}

\usepackage{authblk}
\author[1]{Arjan Cornelissen}
\affil[1]{Simons Institute, University of California Berkeley, California, United States of America}
\author[2]{Amin Shiraz Gilani}
\affil[2]{QuICS, University of Maryland, United States of America}
\author[3]{Subhasree Patro}
\affil[3]{Eindhoven University of Technology, the Netherlands}

\usepackage{amsmath}

\usepackage{amsthm}
\usepackage{thmtools}

\usepackage{cleveref}

\newtheorem{theorem}{Theorem}[section]
\newtheorem{lemma}[theorem]{Lemma}
\newtheorem{proposition}[theorem]{Proposition}
\newtheorem{corollary}[theorem]{Corollary}
\newtheorem{remark}[theorem]{Remark}
\newtheorem{claim}[theorem]{Claim}
\newtheorem{observation}[theorem]{Observation}
\newtheorem{main-result}[theorem]{Main result}
\theoremstyle{definition}
\newtheorem{definition}[theorem]{Definition}

\crefname{observation}{observation}{Observations}
\Crefname{observation}{Observation}{Observations}
\crefname{proposition}{proposition}{propositions}
\Crefname{proposition}{Proposition}{Propositions}
\crefname{claim}{claim}{claims}
\Crefname{claim}{Claim}{Claims}

\usepackage{mathtools}
\usepackage{tikz}
\usepackage{algorithm}
\usepackage{algpseudocode}

\usepackage{amssymb}

\newcommand{\N}{\ensuremath{\mathbb{N}}}
\renewcommand{\P}{\ensuremath{\mathbb{P}}}
\newcommand{\R}{\ensuremath{\mathbb{R}}}

\newcommand{\Q}{\ensuremath{\mathsf{Q}}}

\newcommand{\ket}[1]{\ensuremath{\left|#1\right\rangle}}

\newcommand{\floor}[1]{\lfloor #1 \rfloor}
\newcommand{\ceil}[1]{\lceil #1 \rceil}

\newcommand{\DirCycleEqk}{\ensuremath{\mathtt{DirCycle}^{=k}}}
\newcommand{\DirCycleLeqk}{\ensuremath{\mathtt{DirCycle}^{\leq k}}}
\newcommand{\PromDirCycleEqk}{\ensuremath{\mathtt{PromDirCycle}^{= k}}}
\newcommand{\PromDirCycleLeqk}{\ensuremath{\mathtt{PromDirCycle}^{\leq k}}}
\newcommand{\CycleEqk}{\ensuremath{\mathtt{Cycle}^{=k}}}
\newcommand{\CycleLeqk}{\ensuremath{\mathtt{Cycle}^{\leq k}}}
\newcommand{\PromCycleEqk}{\ensuremath{\mathtt{PromCycle}^{=k}}}
\newcommand{\PromCycleLeqk}{\ensuremath{\mathtt{PromCycle}^{\leq k}}}

\newcommand{\DirCycleEqks}{\ensuremath{\mathtt{DirCycle}^{=k}_{s}}}
\newcommand{\DirCycleLeqks}{\ensuremath{\mathtt{DirCycle}^{\leq k}_{s}}}
\newcommand{\PromDirCycleEqks}{\ensuremath{\mathtt{PromDirCycle}^{= k}_{s}}}
\newcommand{\PromDirCycleLeqks}{\ensuremath{\mathtt{PromDirCycle}^{\leq k}_{s}}}
\newcommand{\CycleEqks}{\ensuremath{\mathtt{Cycle}^{=k}_{s}}}
\newcommand{\CycleLeqks}{\ensuremath{\mathtt{Cycle}^{\leq k}_{s}}}
\newcommand{\PromCycleEqks}{\ensuremath{\mathtt{PromCycle}^{=k}_{s}}}
\newcommand{\PromCycleLeqks}{\ensuremath{\mathtt{PromCycle}^{\leq k}_{s}}}

\begin{document}

\maketitle

    \begin{abstract}
        The quantum query complexity of subgraph-containment problems, which ask whether a given subgraph $H$ is present in an input graph $G$, has been the subject of considerable study. This interest stems not only from the natural and well-motivated formulation of these problems, but also from a flurry of novel quantum algorithmic techniques that were developed specifically to solve them. Notably, even for relatively simple subgraphs, such as paths and cycles, a complete understanding of their query complexities remains elusive.
        
        % The quantum query complexity of constant-size subgraph containment problems has been a well-studied area of research. \Subha{I feel the first sentence and the transition to the second sentence sounds unexciting. How about something like 'The quantum query complexity of constant-size subgraph containment problems has been the subject of considerable study. This interest stems not only from the natural and well-motivated formulation of these problems, but also from a difficulty of establishing quantum lower bounds that surpass the certificate complexity barrier. Notably, even for relatively simple subgraphs, such as paths and cycles, a complete understanding remains elusive. In this work,}\Arjan{I agree, but I guess part of the motivation is also that trying to solve the subgraph-containment problem led the development of novel algorithmic techniques, so perhaps it's nice to stress that as well. ;)}
        
        In this work, we consider several variants of path- and cycle-containment problems in the adjacency matrix model, where we search for paths or cycles of constant length $k \in O(1)$. We compare the settings where the graphs are directed or undirected, where the goal is to detect or find the existence of a path/cycle, and where the path/cycle we're looking for has length exactly $k$, or at most $k$. We also consider several promise versions of these problems, where we know beforehand that the input graph has a certain structure. We characterize the relative difficulty of these variants of the path- and cycle-containment problems, by relating them to one another using randomized reductions, and grouping them into several equivalence classes.

        When we restrict our attention to path-containment problems, this implies a dichotomy result. Some of the path-containment problems can be solved using a linear number of queries, and all the others are equivalent to one another (and additionally to several cycle-containment problems) under randomized reductions and up to constant multiplicative overhead. For the latter equivalence class, we prove a novel quantum-walk-based algorithm that achieves query complexity $\widetilde{O}(n^{3/2-\alpha_k})$, where $\alpha_k \in \Theta(c^{-k})$ and $c = \sqrt{3+\sqrt{17}}/2 \approx 1.33$, beating the previous best upper bound $O(n^{3/2})$ on its query complexity. We also provide a conditional lower bound based on the graph-collision problem, which implies that this equivalence class does not admit linear-query quantum algorithms unless graph collision admits an $O(\sqrt{n})$ query algorithm.
    \end{abstract}

    \thispagestyle{empty}

    \newpage
    \tableofcontents
    \thispagestyle{empty}

    \newpage
    \setcounter{page}{1}
    \pagenumbering{arabic}

    \section{Introduction}

    The subgraph containment problem asks whether a given (undirected) subgraph $H$ is contained in an (undirected) input graph $G$. In this work, we consider subgraphs $H$ of constant size that are known ahead of time, and input graphs $G$ consisting of $n$ vertices that we can access through adjacency-matrix queries. Now, we wish to decide whether the subgraph $H$ is present in the input graph $G$, while making a minimal number of queries to the adjacency matrix of $G$. We denote this problem concisely by $\mathtt{Subgraph}^H$, and the minimal number of queries to the input required to solve it is known as its query complexity. We also consider a restricted version of this problem where we select several vertices $v_1, \dots, v_m$ from $H$ and additionally require that these vertices get mapped to predetermined positions $s_1, \dots, s_m$ in $G$. Under this restriction, we then ask whether $H$ is present in $G$, and we denote the resulting problem by $\mathtt{Subgraph}^H_{v_1, \dots, v_m}$.
    
    % additionally fix the positions of several vertices $v_1, \dots, v_m$ of $H$ in $G$ ahead of time, which we denote by $\mathtt{Subgraph}^H_{v_1, \dots, v_m}$. \Subha{it is not quite clear from this amount of information as to what fixing positions of several vertices mean here?}
    
    In the randomized setting, we decide which edges to query by a (classical) randomized algorithm, and we require the resulting algorithm to output the correct answer with probability at least $2/3$. The randomized query complexity of the unrestricted subgraph containment problem, denoted by $\mathsf{R}(\mathtt{Subgraph}^H)$, satisfies $\Theta(n^2)$ for any non-empty graph $H$.
    % i.e., where we can decide which edges to query by a (classical) randomized algorithm, the unrestricted subgraph containment problem for a non-empty graph $H$ has query complexity $\Theta(n^2)$.
    Moreover, it is not hard to show that the restricted problems are classified into three distinct classes, where the randomized query complexities $\mathsf{R}(\mathtt{Subgraph}^H_{v_1, \dots, v_m})$ are $\Theta(1)$, $\Theta(n)$ and $\Theta(n^2)$, respectively.
    
    In contrast, in the quantum setting, we can query the adjacency matrix coherently in superposition, and we aim to give a quantum algorithm whose final measurement outputs the right answer with probability at least $2/3$. The corresponding quantum query complexities, denoted by $\mathsf{Q}(\mathtt{Subgraph}^H)$ for the unrestricted problem, and $\mathsf{Q}(\mathtt{Subgraph}^H_{v_1, \dots, v_m})$ for the restricted problem, seem to have a much richer structure than their classical counterparts. This is the focus of this work.
    
    % Surprisingly, this problem family seems to have a much richer structure in the quantum setting, where we can query the adjacency matrix coherently in superposition, which is the setting we focus on in this work.
    
    The study of the constant-size subgraph-containment problem in the quantum setting was pioneered by Magniez, Szegedy and Santha~\cite{magniez2007quantum}, who built on a quantum walk algorithm by Ambainis~\cite{ambainis2007quantum} to show that the existence of any subgraph $H$ consisting of $k$ vertices can be detected using $\widetilde{O}(n^{2-2/k})$ queries. This result was subsequently independently improved by Zhu~\cite{zhu2011quantum}, and Lee, Magniez and Santha~\cite{lee2012learning}, resulting in the current state-of-the-art complexity for general subgraphs $H$ through the following learning-graph-based result.
    
    \begin{theorem}[{\cite[Theorem~4.4]{lee2012learning}}]
       \label{thm:general-subgraph-containment}
       Let $H$ be an undirected graph with $k \geq 3$ vertices, $m$ edges, and minimum vertex-degree $d \geq 1$, where $k,m,d \in O(1)$. Then, $\mathsf{Q}(\mathtt{Subgraph}^H) \in O\left(n^{2-2/k-t}\right)$, where
       \[t = \max\left\{\frac{k^2 - 2(m+1)}{k(k+1)(m+1)}, \frac{2k-d-3}{k(d+1)(m-d+2)}\right\}.\]
    \end{theorem}
    
    Even though this result gives a non-trivial upper bound on the query complexity for the subgraph containment problem for any subgraph $H$, it is not expected that this bound is optimal in general. Indeed, for several specific families of subgraphs, we know better bounds than those provided by Lee, Magniez and Santha. Examples of such subgraph families include the line graphs with $k$ edges ($L_k$) and cycle graphs with $k$ edges ($C_k$), for which we refer to the subgraph-containment problems as $\mathtt{Path}^{=k} := \mathtt{Subgraph}^{L_k}$ and $\mathtt{Cycle}^{=k} := \mathtt{Subgraph}^{C_k}$, respectively.

    %The path-finding problem in general graphs was first studied by Duur et al~\cite{durr2006graph}, who gave an algorithm for deciding if there is an arbitrary-length path between all pairs of vertices (even in directed graphs) in query complexity $O(n^{3/2})$.
    The path-containment problem was pioneered by Childs and Kothari~\cite{childs2012quantum}, who showed that it can be solved using $\widetilde{O}(n^{3/2-1/\lceil k/2\rceil})$ quantum queries. 
    %\Amin{The quantum query complexity of graph problems paper did consider the non-parametrized path finding problems first} 
    This was later improved by Belovs and Reichardt~\cite{belovs2012span} to $\mathsf{Q}(\mathtt{Path}^{=k}) \in \Theta(n)$, for all $k \in O(1)$.     Their algorithm involves reducing this problem to the promise problem, which we denote by $\mathtt{PromPath}^{\leq k}_{s,t}$, where two vertices $s$ and $t$ are fixed, and the task is to distinguish between the existence of a path between $s$ and $t$ of length at most $k$, and no path between $s$ and $t$ at all. They used a span-program approach to show that $\mathsf{Q}(\mathtt{PromPath}^{\leq k}_{s,t}) \in O(n\sqrt{k})$, even for non-constant $k$.
    %More generally \Amin{It might not be clear to a reader why this is more general}, Belovs and Reichardt also considered the promise problem where two vertices $s$ and $t$ are fixed, and where the task is to distinguish between the existence of a path between $s$ and $t$ of length at most $k$, and no path between $s$ and $t$ at all. We denote this problem by $\mathtt{PromPath}^{\leq k}_{s,t}$, and they showed using a span-program approach that $\mathsf{Q}(\mathtt{PromPath}^{\leq k}_{s,t}) \in O(n\sqrt{k})$, even for non-constant $k$. \Amin{Andrew told me that people tried to study the directed analogue of this problem and tried to give the same upper bound but their wasn't any progress. Our work explains why!}

    Despite Belovs's and Reichardt's remarkable progress on the path-containment problem, subsequent efforts to generalize their techniques were unsuccessful. Specifically, all attempts to solve the path-containment problem in the directed setting, or a version of their promise problem where one is required to not just detect the existence of a path, but actually output it, have not resulted in linear-query algorithms. In this work, we argue why these efforts are unlikely to work, with the following reasoning containing two parts. First, we show that all these directed and finding versions of the path-containment problem are equivalent up to randomized reductions (\Cref{sec:reductions}), and then we show a conditional quantum query lower bound for all these problems based on the graph-collision problem (\Cref{sec:lower-bounds}). Since it is considered unlikely that we can make progress on the graph-collision problem with the current techniques, we conclude it's unlikely that we can make progress on these generalizations of Belovs's and Reichardt's work with the current techniques as well.
    
    % These ingredients combined showcase why a linear-query algorithm for these problems is unlikely to exist.
    
    % \todo{Remove:} In this work we explain why, as all these problems turn out to be equivalent up to randomized reductions, and we prove a conditional lower bound based on the graph-collision problem, which suggests that a linear-query algorithm is unlikely to exist. \Subha{Nice motivation! Can we give a peep into what we show about these kind of results?}\Arjan{What would you suggest? ;)}
    
    % Inspired by these results, it is natural to ask how the quantum query complexity of these problems changes if we also ask an algorithm to output a path in case it exists or consider analogs of these problems for directed graphs or remove the promise in $\mathtt{PromPath}^{\leq k}_{s,t}$? (\Amin{May be give a reference for precise definitions?}). It is known from the results of Duur et al~\cite{durr2006graph} that all these problems admit an algorithm with complexity $O(n^{3/2})$ and are lower bounded by $\Omega(n)$. But can we show better bounds that potentially depend on $k$? Moreover, can we show (conditional) separations between finding and detecting versions of the same problem, or the same problem considered in directed or undirected graphs? In this work, we considered relations between these problems, improved the existing upper bounds and gave conditional lower bounds for path problems not known to have linear query algorithms. 
    On the topic of cycle-containment problems, the problem of detecting $3$-cycles (more commonly known as the triangle-detection problem) has played a central role in the literature on quantum query complexity. The study of this problem gave rise to many novel quantum algorithmic frameworks, and has been improved multiple times since it was first considered in \cite{buhrman01elementdistinctness}. Currently, the best-known algorithm is by Carette, Lauri\`ere and Magniez~\cite{carette2020extended}, who removed the log-factors from the algorithm by Le Gall~\cite{le2014improved}.

    \begin{theorem}[\cite{le2014improved,carette2020extended}]
        $\mathsf{Q}(\mathtt{Cycle}^{=3}) \in O(n^{5/4})$.
    \end{theorem}
    
    The more general constant-length cycle-finding problem was also considered by Childs and Kothari~\cite{childs2012quantum}, who showed the following improved upper bounds for even values of $k$.
    
    \begin{theorem}[{\cite[Theorem~4.11,4.12]{childs2012quantum}}]
        Let $4 \leq k \in \N$ be even. Then, $\mathsf{Q}(\mathtt{Cycle}^{=k}) \in \widetilde{O}\left(n^{\frac32 - \frac{k-2}{k(k+2)}}\right)$. Moreover, $\mathsf{Q}(\mathtt{Cycle}^{=4}) \in \widetilde{O}(n^{5/4})$.
    \end{theorem}
    
    %For the specific case where $k = 4$, they improved their construction to obtain the following result.
    
    % \begin{theorem}[{\cite[Theorem~4.12]{childs2012quantum}}]
    %     $\mathsf{Q}(\mathtt{Cycle}^{=4}) \in \widetilde{O}(n^{5/4})$.
    % \end{theorem}
    
    %Furthermore, the specific problem where $k = 3$ is known as the triangle-finding problem, which has obtained considerable attention in the literature. Subsequent to the general result by Lee, Magniez and Santha, an improved algorithm was given by \cite{le2014improved}, and the log-factors were later removed by \cite{carette2020extended}, which gives the following state-of-the-art result.
    
    % \begin{theorem}[\cite{le2014improved,carette2020extended}]
    %     $\mathsf{Q}(\mathtt{Cycle}^{=3}) \in O(n^{5/4})$.
    % \end{theorem}

    %\Amin{I don't see a point in mentioning this here; not sure it adds to the context.We can potentially mention it when we describe our contributions?}
    Finally, Cade, Montanaro and Belovs~\cite{cade2018time} considered a promise version of the cycle-containment problem, where they fix a vertex $s$ in the input graph, and then distinguish between graphs with a cycle of length at most $k$ passing through $s$ and graphs with no cycles at all. Much like the promise path problem, they present an algorithm making $O(n\sqrt{k})$ queries, even for non-constant values of $k$.

    As in the path-containment setting, all attempts to generalize the above results to the directed setting proved to be futile. Moreover, generalizing the promise version of the problem considered by Cade, Montanaro and Belovs to the finding setting, where one additionally has to output the cycle, proved elusive as well. In this work, we connect the cycle-containment problems rigorously to the path-containment problems through randomized reductions, explaining in a similar fashion why these generalization efforts are unlikely to succeed.
    
    From the lower-bound perspective, progress is hampered by the \textit{certificate barrier}~\cite{zhang2005power,vspalek2006all}, which states that lower bounds obtained through the non-negative-weighted adversary bound are upper bounded by the geometric average of the worst-case positive and negative certificate sizes. For all subgraph containment problems where the subgraph $H$ is of constant size, the positive certificate sizes are of constant size as well, and therefore the best possible lower bound on the quantum query complexity we can obtain through the non-negative-weighted adversary bound is, up to constants, equal to $\sqrt{n^2 \cdot 1} = n$.
    
    Even though there are other methods available for proving lower bounds on quantum query complexity, like the polynomial method \cite{beals98polynomial}, and the general adversary bound \cite{hoyer07adversary}, applying these techniques to the subgraph containment problem has so far not been fruitful. As such, proving any unconditional super-linear lower bound on the quantum query complexity for any subgraph containment problem is a major open problem in this area that has been open for about two decades.
    %since the inception of quantum algorithms. \Subha{any subgraph or constant-size subgraph? Also 'since inception of quantum algorithms' seems a bit too far stretched to me}
    %for about two decades.
    
    Despite (or perhaps because of) these barriers, alternative progress was made by Balodis and Iraids~\cite{balodis2016quantum} in showing a conditional lower bound for the triangle-detection problem, based on the hardness of the graph-collision problem. This problem is parametrized by a graph $G$ on $n$ vertices, which is completely known ahead of time. The input to the problem is a bit sting $x$ that we have query access to, where each bit of $x$ is associated with a distinct vertex of $G$. The graph collision problem on $G$ asks whether there exists an edge in $G$ for which both endpoints are labeled by a $1$ in $x$. We refer to this problem as $\mathtt{GC}_G$, and we write $\mathsf{Q}(\mathtt{GC}_n)$ for the maximum of $\mathsf{Q}(\mathtt{GC}_G)$, over all graphs $G$ on $n$ vertices.

    The best generic algorithm for this problem was given by Magniez, Szegedy and Santha~\cite[Theorem~3]{magniez2007quantum}, which is essentially based on the algorithm of Ambainis~\cite{ambainis2007quantum} and has query complexity $\mathsf{Q}(\mathtt{GC}_n) \in O(n^{2/3})$. This complexity is believed by some in the community to be optimal. However, the best-known lower bound is only $\mathsf{Q}(\mathtt{GC}_n) \in \Omega(\sqrt{n})$, where further progress is again impeded by the certificate barrier.
    
    Balodis and Iraids embedded the OR of $n$ instances of the graph-collision problem into a single instance of the triangle-detection problem on $3n$ vertices. In doing so, they prove a conditional lower bound on the family of cycle-containment problems.

    \begin{theorem}[{\cite[Theorem~1]{balodis2016quantum}}]
        $\mathsf{Q}(\mathtt{Cycle}^{=3}) \in \Omega(\sqrt{n} \cdot \mathsf{Q}(\mathtt{GC}_n))$.
    \end{theorem}

    This result implies a plausible path toward lower bounding the query complexity of the triangle-detection problem. Indeed, any lower bound $\omega(\sqrt{n})$ for the graph-collision problem would immediately imply a super-linear lower bound $\omega(n)$ for the triangle-detection problem.

    In this work, we provide similar graph-collision-based conditional lower bounds for the families of path- and cycle-containment problems. This solidifies the role of the graph-collision problem in the landscape of subgraph-containment problems, as a non-trivial lower bound for graph-collision would imply non-trivial lower bounds for both cycle- and path-containment problems.
    
    %Balodis's and Iraids's result now becomes the following statement.

    % \Amin{Should we introduce the various variants of the path finding and cycle finding problems that are natural to arise (and we consider)? Should we talk about the need for categorization of path finding and cycle finding problems and parametrized algorithms for other classes of path finding and cycle finding problems?}

    % \Amin{I think it's somehow better to naturally introduce some of the open questions in this section that we want to talk about in the next section. I commented on some of the places earlier where there is a possibility to do that}

    % \Amin{I also think that we should focus our story around that particular island of path finding and cycle finding problems that we can show have the same complexity and give the algorithm for. That is the most interesting aspect! In particular that all the non-known-to-be-linear path finding problems (and a few cycle finding problems) have the same complexity and is different from linear (assuming the hardness of GC) and have a parametrized algorithmic advantage for all constant k. (The rest of our work is an attempt for a characterization of islands which may not immediately impress the reviewers. Also, pretty much all our techniques are not very non-trivial so our main contribution is the above story that we get) Not sure what's the best way to do this though :(}

    \subsection{Contributions}

    In this work, we provide a rigorous study of the different variants of subgraph-containment problems, with a specific focus on path- and cycle-containment problems. We present a number of simpler \textit{observations} based on prior work that help paint the landscape of these problems, and we obtain three \textit{main results} that improve our understanding of their query complexities.
    
    For the general constant-size subgraph-containment problem, we consider relations between the undirected and the directed setting, as well as between their detection and finding versions. Even though it is not explicitly analyzed by Lee, Magniez and Santha, we make the simple observation that their learning-graph-based approach readily generalizes to the directed setting.
    
    \begin{observation}[Informal version of \Cref{prop:learning-graph-directed-subgraph-containment}]
        \Cref{thm:general-subgraph-containment} holds too whenever the subgraph $H$ and the input graph are directed.
    \end{observation}

    Next, we also observe that for non-promise subgraph-containment problems, detection and finding are equivalent up to randomized reductions and constant multiplicative overhead.

    \begin{observation}[Informal version of \Cref{lem:find-detect-subgraph}]
        \label{obs:detection-vs-finding}
        For any (directed/undirected) subgraph $H$, the quantum query complexities of the finding and detection versions of the non-promise (possibly restricted) subgraph-containment problem with subgraph $H$ are equal up to a multiplicative constant.
    \end{observation}

    We note that the above only holds for non-promise problems, i.e., where all graphs are valid inputs to the problem. That is, one can only have separations between the detection and finding versions of a subgraph-containment problem in a setting where the input satisfied some promise known a priori. Indeed, in such promise settings, there are several graph problems where the finding and detection versions seem to be very much distinct. For instance, there is a quantum algorithm that makes polylogarithmic queries to output the (first bit of the) \textsc{exit} vertex in the Welded tree graph \cite{childs03weldedtree}, while outputting an \textsc{entrance}-to-\textsc{exit} path in this graph takes polynomially many queries for a certain natural class of quantum algorithms \cite{childs23weldedtree}. Moreover, as we will see, the promise problem $\mathtt{PromFindPath}_{s,t}^{=k}$ would be polynomially harder than $\mathtt{PromPath}_{s,t}^{=k}$ assuming graph collision is harder than its trivial lower bound.

    We also show that there is a generic randomized reduction from the undirected version of a subgraph-containment problem to the directed one. 

    \begin{observation}[Informal version of \Cref{lem:undirected_to_directed}]
        Every directed subgraph-containment problem is at least as hard (in terms of randomized/quantum query complexity) as its undirected counterpart.
    \end{observation}

    Next, we focus specifically on the path- and cycle-containment problems. For these problems, we consider the restricted and unrestricted versions of the problem, where in the restricted versions, we fix the start and end vertices $s$ and $t$ for the path-containment problem, and a single vertex $s$ on the cycle for the cycle-finding problem. We also consider the directed versions of these problems, and the finding versions where we ask to output the path or cycle. Finally, we consider versions where we search for a path or cycle of length exactly $k$, or at most $k$. We denote the directed and finding versions with the keywords $\mathtt{Dir}$ and $\mathtt{Find}$, we attach a superscript $=k$ or $\leq k$ to indicate the length constraint, and we supply a subscript $s$ or $s,t$ for the restricted versions. Thus, for instance, the problem $\mathtt{FindDirPath}_{s,t}^{=k}$ is the finding version of the directed path-finding problem between $s$ and $t$, of length exactly $k$.

    Akin to Belovs and Reichardt~\cite{belovs2012span}, and Cade et al.~\cite{cade2018time}, we also consider promise versions of this problem. In the restricted path problem, i.e., where we fix the start and end vertex $s$ and $t$, the promise version satisfies that a path of the required length exists between $s$ and $t$, or $s$ and $t$ are not connected at all. Similarly, in the restricted cycle-finding problem through a vertex $s$, either a cycle through $s$ exists with the required length, or no cycle through $s$ exists at all. We note that this promise is weaker than the promise considered in Cade et al.~\cite{cade2018time}, and so our algorithmic results for this problem are stronger. Finally, for the unrestricted cycle-finding problem, we also consider a promise version, where either a cycle exists of the required length, or there exists no cycle at all (i.e., the graph is a forest). These problems are indicated with the $\mathtt{Prom}$ keyword, i.e., $\mathtt{PromDirCycle}_s^{\leq k}$ is the promise problem of detecting a directed cycle through $s$ of length at most $k$.\footnote{Note that we don't consider a promise version of the unrestricted path-containment problem, because we don't see a meaningful way of defining it.}
    
    Our first observation in this direction is that every $\mathtt{PromFind}$-version of these problems requires the same number of queries as the corresponding $\mathtt{Find}$-version.

    \begin{observation}[Informal version of \Cref{Thm:PromReducesNonProm,lem:PromFindForCycle}]
        \label{obs:find-vs-promfind}
        Let $\mathtt{P}$ be a (restricted/unrestricted, directed/undirected) path- or cycle-containment problem, with length constraint $=k$ or $\leq k$ with $k \in O(1)$. Then, the finding version of $\mathtt{P}$ is equivalent to the promise-finding version of $\mathtt{P}$, up to randomized reductions and multiplicative constants.
    \end{observation}

    \Cref{obs:detection-vs-finding,obs:find-vs-promfind} show that for all path- and cycle-containment problems $\mathtt{P}$, we have that $\mathtt{P}$, $\mathtt{FindP}$ and $\mathtt{PromFindP}$ are equivalent up to randomized reductions and multiplicative constants. As such, these three problems fall in the same equivalence class, and it suffices to restrict our attention to $\mathtt{P}$ and $\mathtt{PromP}$ in what follows.

    Now, we arrive at the first main result of this work. We classify all remaining path- and cycle-containment problems into several equivalence classes via randomized reductions between them. We aid the reader by providing a graphical overview of these reductions in \Cref{fig:islands}, and we obtain the following classification.

    % Taking these equivalences into account, the remaining number of distinct path and cycle problems that we consider is $26$. For all these problems, we provide randomized reductions between them, and we show that all of them are equivalent to one of $10$\todo{Update!} problems displayed in~\Cref{fig:islands}.

    \begin{main-result}
        \label{ref:classification}
        The unrestricted and undirected path-containment problem, and the undirected promise versions of the path- and cycle-containment problems admit linear-query quantum algorithms. All other path-containment problems are equivalent, and all other cycle-containment problems can be classified into at most 5 equivalence classes. See \Cref{fig:islands}.
    \end{main-result}

    \begin{figure}[!ht]
        \centering
        \begin{tikzpicture}
            
            % The path picture
            \begin{scope}[shift={(0,-6.5)}]
                \node[teal] (path) at (3,-1.5) {$\mathtt{Path}^{=(k-2)}$};
                \node[teal, left, align = right] at (path.west) {\small Lem~\ref{lem:subgraph-characterization}-3 \\[-.2em] $\Omega(n)\ni$};
                \node[red] (dirpath) at (5.5,0) {$\mathtt{DirPath}^{=(k-2)}$};
                \node[red, below=-.3em, align = center] at (dirpath.south) {\rotatebox{90}{$\ni$} \small Cor~\ref{cor:dir-path-detection} \\[-.2em] $\widetilde{O}(n^{\frac32 - \alpha_{k-2}})$};

                \node[red] (pathst) at (3,0) {$\mathtt{Path}^{=k}_{s,t}$};
                \node[red] (dirpathst) at (3.5,1.5) {$\mathtt{DirPath}^{=k}_{s,t}$};
                \node[red] (pathstle) at (0,0) {$\mathtt{Path}^{\leq k}_{s,t}$};
                \node[red] (dirpathstle) at (0,1.5) {$\mathtt{DirPath}^{\leq k}_{s,t}$};
                \node[teal] (prompathst) at (3,-3) {$\mathtt{PromPath}^{=k}_{s,t}$};
                \node[red] (promdirpathst) at (3.5,3) {$\mathtt{PromDirPath}^{=k}_{s,t}$};
                
                \node[teal] (prompathstle) at (0,-3) {$\mathtt{PromPath}^{\leq k}_{s,t}$};
                \node[teal, above=-.5em, align = center] at (prompathstle.north) {$\Theta(n)$ \\[-.2em] \rotatebox{90}{$\in$}~\small\cite{belovs2012span}};
                
                \node[red] (promdirpathstle) at (0,3) {$\mathtt{PromDirPath}^{\leq k}_{s,t}$};

                \draw[->] ([shift={(-.15,0)}]dirpathst.north east) arc (180:-90:.15) node[midway, above, align = center, font = \small] {Lem~\ref{lem:path_monotonicity} \\[-.2em] $k \mapsto k+1$};

                \draw[bend left = 5, ->] (dirpath) to node[below left=-.3em, align = right, font = \small] {Prop \\[-.2em] \ref{thm:dir_st_blank_paths}-\ref{item:dirpath-dirpathst}} (dirpathst);
                \draw[bend left = 5, ->] (dirpathst) to node[above right=-.3em, align = center, font = \small] {Prop \\[-.2em] \ref{thm:dir_st_blank_paths}-\ref{item:dirpathst-dirpath}} (dirpath);
                
                \draw[->] (pathst) to node[above, align = center, font = \small] {Prop~\ref{thm:dir_undir_st_paths}-\ref{item:pathst-pathstle}} (pathstle);
                \draw[->] (pathstle) to node[left, align = right, font = \small] {Prop \\[-.2em] \ref{thm:dir_undir_st_paths}-\ref{item:pathstle-dirpathstle}} (dirpathstle);
                \draw[->] (dirpathstle) to node[above, align = center, font = \small] {Prop~\ref{thm:dir_undir_st_paths}-\ref{item:dirpathstle-dirpathst}} (dirpathst);
                \draw[->] (dirpathst) to node[left, align = right, font = \small] {Prop \\[-.2em] \ref{thm:dir_undir_st_paths}-\ref{item:dirpathst-pathst}} (pathst);
            
                \draw[->] (promdirpathst) to node[above, align = center, font = \small] {Prop \\[-.2em] \ref{thm:prom_nonprom_dir_st_paths}-\ref{item:promdirpathst-promdirpathstle}} (promdirpathstle);
                \draw[->] (promdirpathstle) to node[left, align = right, font = \small] {Prop \\[-.2em] \ref{thm:prom_nonprom_dir_st_paths}-\ref{item:promdirpathstle-dirpathstle}} (dirpathstle);
                \draw[->] (dirpathst) to node[left, align = right, font = \small] {Prop \\[-.2em] \ref{thm:prom_nonprom_dir_st_paths}-\ref{item:dirpathst-promdirpathst}} (promdirpathst);

                \draw[bend left = 5, ->] (prompathst) to node[below, align = center, font = \small] {\cite{belovs2012span}} (prompathstle);
                \draw[bend left = 5, ->] (prompathstle) to node[above, align = center, font = \small] {\cite{belovs2012span}} (prompathst);
                \draw[->] (path) to node[right, align = left, font = \small] {\cite{belovs2012span}} (prompathst);
            \end{scope}

            % The unrestricted cycle subgraph picture
            \begin{scope}[shift={(3.5,0)}]
                \node[purple] (cycle) at (-.5,1.5) {$\mathtt{Cycle}^{=k}$};
                \node[purple, left, align = right] at (cycle.west) {$\widetilde{O}(n^{\frac32 - \frac{k-2}{k(k+2)}}) \ni$ \\[-.2em] \small $k$ even, \cite[Thm~4.11]{childs2012quantum}};
                
                \node[blue] (dircycle) at (-.5,0) {$\mathtt{DirCycle}^{=k}$};
                \node[blue, left, align = right] at (dircycle.west) {$O(n^{2 - \frac{3}{k+1} + \frac{1}{(k+1)^2}}) \ni$ \\[-.2em] \small\cite[Thm~9]{lee2012learning}};
                
                \node[violet] (cyclele) at (3,1.5) {$\mathtt{Cycle}^{\leq k}$};
                \node[violet, right, align = left] at (cyclele.east) {$\in \widetilde{O}(n^{\frac32 - \frac{k-3}{(k-1)(k+1)}})$ \\[-.2em] \small $k$ odd, Thm~\ref{thm:cycles<=k}};
                
                \node[blue] (dircyclele) at (3,0) {$\mathtt{DirCycle}^{\leq k}$};
                \node[teal] (promcycle) at (7.5,3) {$\mathtt{PromCycle}^{=k}$};
                \node[orange] (promdircycle) at (7.5,-1.5) {$\mathtt{PromDirCycle}^{=k}$};
                \node[teal] (promcyclele) at (3,3) {$\mathtt{PromCycle}^{\leq k}$};
                \node[orange] (promdircyclele) at (3,-1.5) {$\mathtt{PromDirCycle}^{\leq k}$};

                \draw[->] ([shift={(-.2,0)}]promdircycle.south) arc (-180:0:.2) node[midway, below, align = center, font = \small] {$k \mapsto k+1$ \\[-.2em] Lem \ref{thm:MonotonicityCycles}-\ref{item:PromMonotonicityDirCycle}};
                \draw[->] ([shift={(-.2,0)}]dircycle.south) arc (-180:0:.2) node[midway, below, align = center, font = \small] {$k \mapsto k+1$ \\[-.2em] Lem \ref{thm:MonotonicityCycles}-\ref{item:MonotonicityDirCycle}};
                \draw[dashed, ->] ([shift={(.15,0)}]cyclele.north west) arc (0:270:.15) node[midway, above, align = center, font = \small] {Lem \ref{thm:MonotonicityCycles}-\ref{item:MonotonicityCycleLeqK} \\[-.2em] $k-1 \mapsto k$};

                \node[teal, above=-.5em, align = center] at (promcycle.north) {$\Theta(n)$ \\[-.2em] \rotatebox{90}{$\in$}~\small\cite{belovs2012span}};

                \node[gray] (orgc) at (-.5,3) {$\mathtt{OR}_n \circ \mathtt{GC}_G$};
                \draw[->] (orgc) to node[left] {$k=3$} node[right, align = left, font = \small] {\cite{balodis2016quantum}} (cycle);
                \node[gray, left] at (orgc.west) {$\Theta(\sqrt{n} \cdot \mathsf{Q}(\mathtt{GC}_G)) \ni$};

                \draw[bend left = 5, ->] (dircycle) to node[above, align = center, font = \small] {Prop~\ref{Cycle-Equivalences-Part1}-\ref{item:DirCycleEqLeqk}} (dircyclele);
                \draw[bend left = 5, ->] (dircyclele) to node[below, align = center, font = \small] {Prop~\ref{Cycle-Equivalences-Part1}-\ref{item:DirCycleEqLeqk}} (dircycle);
                
                \draw[bend left = 4, ->] (promdircycle) to node[below, align = center, font = \small] {Prop~\ref{Cycle-Equivalences-Part1}-\ref{item:PromDirCycleEqLeqk}} (promdircyclele);
                \draw[bend left = 4, ->] (promdircyclele) to node[above, align = center, font = \small] {Prop~\ref{Cycle-Equivalences-Part1}-\ref{item:PromDirCycleEqLeqk}} (promdircycle);
                
                \draw[->] (cyclele) to node[below, align = center, font = \small] {Prop~\ref{Cycle-Equivalences-Part1}-\ref{item:CycleEqLeqk}} (cycle);
                
                \draw[bend left = 4, ->] (promcycle) to node[below, align = center, font = \small] {Prop~\ref{Cycle-Equivalences-Part1}-\ref{Item:PromCycleEqLeqk}} (promcyclele);
                \draw[bend left = 4, ->] (promcyclele) to node[above, align = center, font = \small] {Prop~\ref{Cycle-Equivalences-Part1}-\ref{Item:PromCycleEqLeqk}} (promcycle);
                
                \draw[bend right = 10, dashed, ->] (dircycle) to node[right, align = left, font = \small] {Prop \\[-.2em] \ref{Cycle-Equivalences-Part1}-\ref{Item:DirCycleLeq}} (cycle);
                \draw[bend right = 10, ->] (cycle) to node[left, align = right, font = \small] {Prop \\[-.2em] \ref{Cycle-Equivalences-Part1}-\ref{item:ClEqKdir}} (dircycle);
                
                \draw[->] (promcyclele) to node[right, align = left, font = \small] {Prop \\[-.2em] \ref{Cycle-Equivalences-Part1}-\ref{item:promCleqkNonprom}} (cyclele);
                \draw[->] (promdircyclele) to node[right, align = left, font = \small] {Prop \\[-.2em] \ref{Cycle-Equivalences-Part1}-\ref{item:promDirCleqkNonprom}} (dircyclele);
                \draw[->] (promcycle) to node[right, align = left, font = \small] {Prop \\[-.2em] \ref{Cycle-Equivalences-Part1}-\ref{item:PromCEqKPromDir}} (promdircycle);
                \draw[->] (cyclele) to node[right, align = left, font = \small] {Prop \\[-.2em] \ref{Cycle-Equivalences-Part1}-\ref{item:CEqKDir}} (dircyclele);
            \end{scope}

            % The restricted cycle subgraph picture
            \begin{scope}[shift={(8,-6.5)}]
                % s-problems
                \node[red] (cycles) at (0,0) {$\mathtt{Cycle}^{=k}_s$};
                \node[red] (dircycles) at (-.5,1.5) {$\mathtt{DirCycle}^{=k}_s$};
                \node[red] (cyclesle) at (3,0) {$\mathtt{Cycle}^{\leq k}_s$};
                \node[red] (dircyclesle) at (3,1.5) {$\mathtt{DirCycle}^{\leq k}_s$};
                \node[teal] (promcycles) at (-.5,-3) {$\mathtt{PromCycle}^{=k}_s$};
                \node[red] (promdircycles) at (-1,3) {$\mathtt{PromDirCycle}^{=k}_s$};
                \node[teal] (promcyclesle) at (3,-3) {$\mathtt{PromCycle}^{\leq k}_s$};
                \node[red] (promdircyclesle) at (3,3) {$\mathtt{PromDirCycle}^{\leq k}_s$};

                \node[gray] (gcor) at (0,-1.75) {$\mathtt{GC}_G \circ \mathtt{OR}_n$};
                \draw[->] (gcor) to node[left] {$k = 5$} node[right, align = left, font = \small] {Prop~\ref{prop:cycle-s-lb}} (cycles);
                \node[gray, left] at (gcor.west) {$\Theta(\sqrt{n} \cdot \mathsf{Q}(\mathtt{GC}_G)) \ni$};

                \draw[bend left = 5, ->] (dircycles) to node[above, align = center, font = \small] {Prop~\ref{Cycle-Equivalences-Part2}-\ref{item:DirCycleEqLeqks}} (dircyclesle);
                \draw[bend left = 5, ->] (dircyclesle) to node[below, align = center, font = \small] {Prop~\ref{Cycle-Equivalences-Part2}-\ref{item:DirCycleEqLeqks}} (dircycles);

                \draw[bend left = 4, ->] (promdircycles) to node[above, align = center, font = \small] {Prop~\ref{Cycle-Equivalences-Part2}-\ref{item:PromDirCycleEqLeqks}} (promdircyclesle);
                \draw[bend left = 4, ->] (promdircyclesle) to node[below, align = center, font = \small] {Prop~\ref{Cycle-Equivalences-Part2}-\ref{item:PromDirCycleEqLeqks}} (promdircycles);

                \draw[bend left = 5, ->] (cyclesle) to node[below, align = center, font = \small] {Prop~\ref{Cycle-Equivalences-Part2}-\ref{item:CycleEqLeqksPartA}} (cycles);
                \draw[bend left = 5, ->] (cycles) to node[above, align = center, font = \small] {Prop~\ref{Cycle-Equivalences-Part2}-\ref{item:CycleEqLeqksPartB}} (cyclesle);

                \draw[bend left = 5, ->] (promcycles) to node[above, align = center, font = \small] {Prop~\ref{Cycle-Equivalences-Part2}-\ref{Item:PromCycleEqLeqks}} (promcyclesle);
                \draw[bend left = 5, ->] (promcyclesle) to node[below, align = center, font = \small] {Prop~\ref{Cycle-Equivalences-Part2}-\ref{Item:PromCycleEqLeqks}} (promcycles);
                
                \draw[bend left = 10, ->] (promdircycles) to node[right, align = left, font = \small] {Prop \\[-.2em] \ref{Cycle-Equivalences-Part2}-\ref{item:PromDirCycleEqksDirCycleEqks}} (dircycles);
                \draw[bend left = 10, ->] (dircycles) to node[left, align = right, font = \small] {Prop \\[-.2em] \ref{Cycle-Equivalences-Part2}-\ref{item:PromDirCycleEqksDirCycleEqks}} (promdircycles);
                \draw[->] (promcyclesle) to node[right, align = left, font = \small] {Prop \\[-.2em] \ref{Cycle-Equivalences-Part2}-\ref{item:promCleqkNonproms}} (cyclesle);
                \draw[->] (promdircyclesle) to node[right, align = left, font = \small] {Prop \\[-.2em] \ref{Cycle-Equivalences-Part2}-\ref{item:promDirCleqkNonproms}} (dircyclesle);

                %\draw[->] (promcycles) to node[right] {\ref{Cycle-Equivalences-Part2}-\ref{item:PromCEqKPromDirs}} (promdircycles);

                \draw[->] (cyclesle) to node[right, align = left, font = \small] {Prop \\[-.2em] \ref{Cycle-Equivalences-Part2}-\ref{item:CEqKDirs}} (dircyclesle);
                
                \draw[bend left = 10, ->] (dircycles) to node[right, align = left, font = \small] {Prop \\[-.2em] \ref{Cycle-Equivalences-Part2}-\ref{item:dirCycleEqkCycleEqks}} (cycles);
                \draw[bend left = 10, ->] (cycles) to node[left, align = left, font = \small] {Prop \\[-.2em] \ref{Cycle-Equivalences-Part2}-\ref{item:ClEqKdirs}} (dircycles);
            \end{scope}

            % \draw[->] (dircycles) to node[left] {\ref{thm:CrossArrow-DirCycleEqksToDirCycleEqk}} (dircycle);
            \draw[->] (promdircycles) to node[near start, above left=-.2em, font = \small, align = center] {Prop~\ref{thm:CrossArrow-DirCycleEqksToDirCycleEqk} \\[-.2em] $k+1 \mapsto k$} (promdircycle);
            
            \draw[->] (dircycles) to node[above, align = center, font = \small] {Prop~\ref{thm:CrossArrow-DirCycleEqksToDirPathstEqk}} (dirpathst);
            \draw[->] (promdirpathst) to node[above, align = center, font = \small] {Prop \\[-.2em] \ref{thm:crossArrow-promDirPathEqkstToPromDirCycleEqks}} (promdircycles);
            \draw[bend left = 5, ->] (promcycles) to node[below, align = center, font = \small] {Prop~\ref{thm:crossArrow-promCycleEqksToPromPathEqkst}} (prompathst);
            \draw[bend left = 5, ->] (prompathst) to node[above, align = center, font = \small] {Prop \ref{thm:crossArrow-PromPathstToPromCycles}} (promcycles);
        \end{tikzpicture}
        \caption{Randomized reductions between problems for $k \geq 3$. The dashed connections hold only for odd values of $k$. For the graph-collision problems, we can take any graph $G$ with $n$ vertices. All problems' randomized query complexities are $\Theta(n^2)$. All complexity bounds stated in the figure are upper and lower bounds on the quantum query complexity. In the upper bound for $\mathsf{Q}(\mathtt{DirPath}^{=(k-2)})$, $\alpha_k \in \Theta(c^{-k})$, with $c \approx 1.33$. If two problems have the same color, they have the same quantum query complexity up to constants.}
        \label{fig:islands}
    \end{figure}
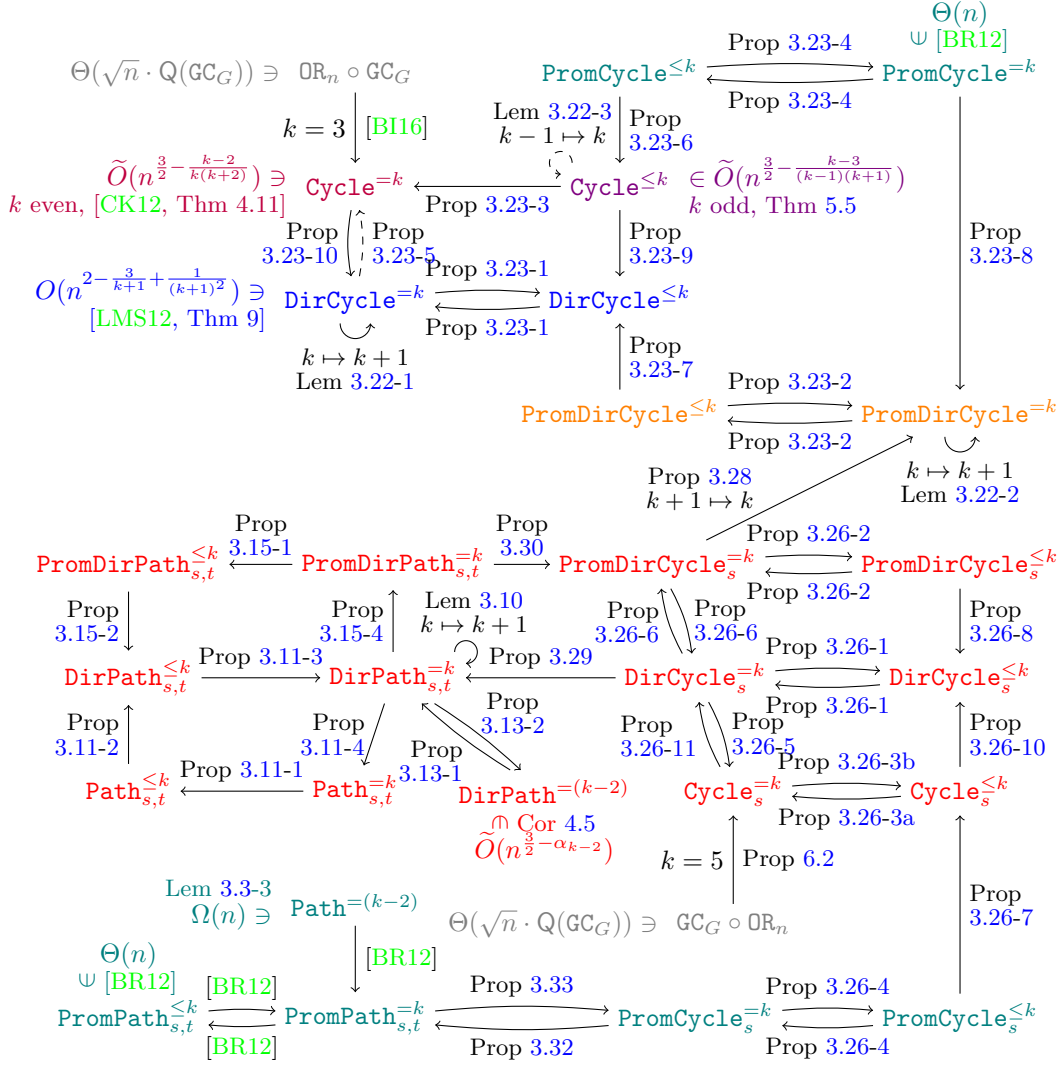

    % \todo{Remove this paragraph...}
    % The predominant technique used in these randomized reductions is the color-coding technique, introduced by Alon, Yuster and Zwick~\cite{alon1995color}. The idea is to assign distinct colors to the vertices of the path/cycle graph $L_k$/$C_k$, and then randomly assign these colors to the vertices of the input graph $G$. Next, we only keep an edge in $G$ if the colors of the two vertices it connects are associated to two vertices in the path/cycle graph that are connected by an edge too. It is easy to show that a quantum query to the resulting graph can be made using a single query to the old graph, and that the new graph contains the path/cycle with constant probability if and only if the original graph contained a path/cycle as well. The resulting color-coded graph often has more structure than the original input graph, which can subsequently be exploited to prove connections between different problems through randomized reductions.

    % \todo{... and replace with this one.}
    For some reductions we use the color-coding technique, introduced by Alon, Yuster and Zwick~\cite{alon1995color}. The idea is to assign distinct colors to the vertices of the path/cycle graph $L_k$/$C_k$, and then to randomly assign these colors to the vertices of the input graph $G$. Next, we only keep an edge in $G$ if the colors of the two vertices it connects are associated to two vertices in the path/cycle graph that are connected by an edge too. The resulting color-coded graph becomes a layered graph in the path-containment case, and a layered cycle graph in the cycle-containment case, as displayed in \Cref{fig:layers}. With at least constant probability, a path/cycle in the original graph gets mapped to a path/cycle in the layered graph that straddles across all the layers.

    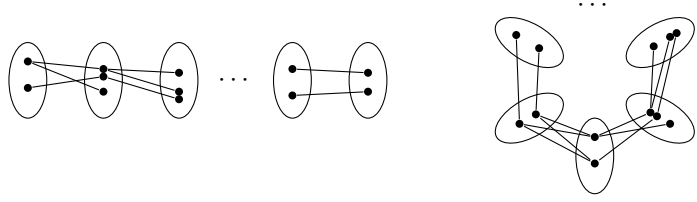
\begin{figure}[!ht]
        \centering
        \begin{tikzpicture}[vertex/.style={fill,rounded corners = .15em, inner sep = .15em}, scale = .5]
            \begin{scope}
                \draw (0,0) ellipse (.5 and 1);
                \draw (2,0) ellipse (.5 and 1);
                \draw (4,0) ellipse (.5 and 1);
                \draw (7,0) ellipse (.5 and 1);
                \draw (9,0) ellipse (.5 and 1);
    
                \node[vertex] (v11) at (0,-.2) {};
                \node[vertex] (v12) at (0,.5) {};
                
                \node[vertex] (v21) at (2,.1) {};
                \node[vertex] (v22) at (2,.3) {};
                \node[vertex] (v23) at (2,-.3) {};
    
                \draw (v11) to (v21);
                \draw (v12) to (v22);
                \draw (v12) to (v23);
    
                \node[vertex] (v31) at (4,-.5) {};
                \node[vertex] (v32) at (4,-.3) {};
                \node[vertex] (v33) at (4,.2) {};
    
                \draw (v21) to (v31);
                \draw (v22) to (v32);
                \draw (v22) to (v33);
    
                \node at (5.5,0) {$\cdots$};
    
                \node[vertex] (v41) at (7,.3) {};
                \node[vertex] (v42) at (7,-.4) {};
                \node[vertex] (v51) at (9,.2) {};
                \node[vertex] (v52) at (9,-.3) {};
    
                \draw (v41) to (v51);
                \draw (v42) to (v52);
            \end{scope}

            \begin{scope}[shift={(15,0)}]
                \foreach \i in {-2,...,2} {
                    \draw[rotate={-\i*60}] (0,-2) ellipse (.5 and 1);
                }
                \node at (0,2) {$\cdots$};
                
                % \draw (2,0) ellipse (.5 and 1);
                % \draw (4,0) ellipse (.5 and 1);
                % \draw (6,0) ellipse (.5 and 1);
                % \draw (8,0) ellipse (.5 and 1);
    
                \node[vertex] (v11) at (0,-2.2) {};
                \node[vertex] (v12) at (0,-1.5) {};
                
                \node[vertex] (v21) at ({-1.9*sin(-60)},{-1.9*cos(-60)}) {};
                \node[vertex] (v22) at ({-1.7*sin(-60)},{-1.7*cos(-60)}) {};
                \node[vertex] (v23) at ({-2.3*sin(-60)},{-2.3*cos(-60)}) {};
    
                \draw (v11) to (v21);
                \draw (v12) to (v22);
                \draw (v12) to (v23);
    
                \node[vertex] (v31) at ({-2.5*sin(-120)},{-2.5*cos(120)}) {};
                \node[vertex] (v32) at ({-2.3*sin(-120)},{-2.3*cos(120)}) {};
                \node[vertex] (v33) at ({-1.8*sin(-120)},{-1.8*cos(120)}) {};
    
                \draw (v21) to (v31);
                \draw (v22) to (v32);
                \draw (v22) to (v33);
    
                \node[vertex] (v41) at ({-1.7*sin(-240)},{-1.7*cos(-240)}) {};
                \node[vertex] (v42) at ({-2.4*sin(-240)},{-2.4*cos(-240)}) {};
                
                \node[vertex] (v51) at ({-1.8*sin(-300)},{-1.8*cos(-300)}) {};
                \node[vertex] (v52) at ({-2.3*sin(-300)},{-2.3*cos(-300)}) {};
    
                \draw (v41) to (v51);
                \draw (v42) to (v52);

                \draw (v51) to (v11);
                \draw (v52) to (v11);
                \draw (v51) to (v12);
                \draw (v52) to (v12);
            \end{scope}
        \end{tikzpicture}
        \caption{A layered graph (left) and a layered cycle graph (right). These can be obtained from general graphs using the color-coding technique, where a path of length $k$ (resp.\ cycle of length $k$) is preserved with constant probability.}
        \label{fig:layers}
    \end{figure}
    
    To show relations between the various subgraph-containment problems, we cleverly make subtle modifications to the input graphs, oftentimes combined with the aforementioned color-coding technique. For instance, we can insert layers into a layered graph to show that one can find short paths using an algorithm that finds longer paths. Similarly, we can merge the vertices $s$ and $t$ into a single vertex $s$ to convert an $st$-path containment problem into a cycle-containment problem through the vertex $s$, and similarly we can separate a vertex $s$ into two vertices $s$ and $t$ to prove a reduction in the reverse direction.
    
    All these conversions, however, are subject to subtle caveats. For instance, in a layered cycle graph with $k$ layers, a cycle of length $k$ is guaranteed to go through all layers whenever $k$ is odd, whereas when $k$ is even it can exist between two consecutive layers. Similarly, when we merge two vertices $s$ and $t$ into a single vertex $s$, a cycle through $s$ in the resulting graph does not necessarily imply that the original graph contains an $st$-path. We obtain our classification, i.e., \Cref{ref:classification}, by carefully checking which reductions can be made to work despite these caveats. Since we don't expect \Cref{fig:islands} to collapse to a single island, we don't expect these subtleties to be artifacts of the techniques used, but rather highlight the fundamental relations between these problems.

    The overview of reductions presented in \Cref{fig:islands} reveals some interesting relations. For instance, it shows that one can reduce the $\mathtt{DirCycle}_s^{\leq k}$-problem to the $\mathtt{Cycle}_s^{\leq k}$-problem, through a series of reductions that involve the $\mathtt{DirCycle}_s^{=k}$- and $\mathtt{Cycle}_s^{=k}$-problems. We remark here that making this connection without going through the intermediate steps is not immediate, since it's not clear how one can remove directions in a layered graph while preventing the introduction of new cycles. This highlights the benefit of this classification effort, as it uncovers relations between problems that would otherwise not be easily obtained.

    Next, the classification of path- and cycle-finding problems displayed in \Cref{fig:islands} motivates the search for new quantum algorithms and query lower bounds for problems in these equivalence classes specifically. For instance, the $\mathtt{DirPath}^{=k}$-problem is part of an equivalence class containing 12 other path- and cycle-containment problems, hence an improved algorithm for this problem immediately provides an algorithm for all the others as well. Moreover, this equivalence class includes all directed versions of the path problems, as well as the finding versions of the promise-path problems, i.e., exactly the problems for which generalizing Belovs's and Reichardt's span-program-based construction failed.
    
    Our main algorithmic result is a novel quantum-walk-based algorithm for this $\mathtt{DirPath}^{=k}$-problem, and hence for all problems in the equivalence class highlighted in red in \Cref{fig:islands}. The best previously-known algorithm for any of the problems in this equivalence class is an algorithm that finds the shortest path between $s$ and $t$, if it exists, and thereby solves the $\mathtt{Path}_{s,t}^{\leq k}$-problem. This algorithm was originally developed by D\"urr et al.~\cite{durr2006graph}, making $\widetilde{O}(n^{3/2})$ queries, and the log-factors were subsequently removed by a series of works culminating in a paper by Lin and Lin~\cite[Theorem~13]{lin2015upper}. We note that Beigi and Taghavi also developed $O(n^{3/2})$-query algorithms for the $\mathtt{Cycle}_s^{=k}$ and $\mathtt{Cycle}_s^{\leq k}$-problems~\cite[Proposition~9.iv-v]{beigi2020quantum}. Our result provides an improved algorithm for all these problems.
    
    %, and the query complexity is also shown in \Cref{fig:islands}.\Arjan{I don't see how D\"urr et al.'s results hold in the $=k$-setting? As far as I understand, they just give algorithms to decide connectivity in general, but then this doesn't imply anything immediately about the length of the path. Am I missing something here?} \Arjan{D\"urr et al.\ solves the $\mathtt{Path}^{\leq k}_{s,t}$-problem in $O(n^{3/2})$, so we can claim that the best-known algorithm for any of the problems in the red equivalence class was this one.}

    \begin{main-result}[\Cref{cor:dir-path-detection}]
        \label{res:quantum_algo}
        Let $5 \leq k \in \N$. Then, $\mathsf{Q}(\mathtt{DirPath}^{=k}) \in \widetilde{O}(n^{3/2 - \alpha_k})$, where $\alpha_k \in \Theta(c^{-k})$ as $k \to \infty$, with $c = \sqrt{3+\sqrt{17}}/2 \approx 1.33$.
    \end{main-result}

    The main algorithmic idea is to reduce the directed-path-detection problem of length $k$ to the same problem with length $k-2$. To that end, we employ the MNRS-framework~\cite{magniez2007search} to search for vertices that have at least one incoming edge, and those that have at least one outgoing edge. Then, it remains to find a pair of such vertices that has a path of length $k-2$ between them, which we then solve recursively. This results in a nested quantum-walk construction, where we optimize the parameters of the quantum-walk framework to minimize the resulting query complexity. The resulting minimization problem involves solving a recurrence relation, whose solution yields an exponent that satisfies $3/2-\Theta(c^{-k})$ for large values of $k$.
    
    We additionally provide an improved algorithm for the $\mathtt{Cycle}^{\leq k}$-problem, which is also displayed in \Cref{fig:islands}. A slight modification of Childs and Kothari~\cite{childs2012quantum} yields the following result.

    \begin{observation}[\Cref{thm:cycles<=k}]
        Let $5 \leq k \in \N$ be an odd integer. Then, $\mathsf{Q}(\mathtt{Cycle}^{\leq k}) \in \widetilde{O}(n^{\frac32 - \frac{k-5}{(k-1)(k+3)}})$.
    \end{observation}

    Finally, we focus on the lower bounds. Our main result in this direction is relating the cycle-finding problem through a particular vertex $s$ to the graph-collision problem. In particular, through the reductions highlighted in \Cref{fig:islands}, we emphasize that this result also provides a conditional lower bound on the $\mathtt{DirPath}^{=k}$-problem, and consequently on all path-containment problems that are in the same equivalence class.

    \begin{main-result}[\Cref{prop:cycle-s-lb}]
        Let $5 \leq k \in \N$ be an odd integer. Then, $\mathsf{Q}(\mathtt{Cycle}_s^{\leq k}) \in \Omega(\sqrt{n} \cdot \mathsf{Q}(\mathtt{GC}_n))$.
    \end{main-result}

    Even though the resulting lower bound is identical to the lower bound derived by Balodis and Iraids for the triangle-detection problem, the underlying construction is profoundly different. Whereas Balodis and Iraids embed $n$ instances of the graph-collision problem into a single instance of triangle-detection, we embed a single instance of the graph-collision problem composed with the OR-function into a single instance of the $5$-cycle-detection problem through a fixed vertex $s$. In other words, Balodis and Iraids prove $\mathsf{Q}(\mathtt{Cycle}^{=3}) \in \Omega(\mathsf{Q}(\mathtt{OR}_n \circ \mathtt{GC}_n))$, whereas we prove that $\mathsf{Q}(\mathtt{Cycle}^{=5}_s) \in \Omega(\mathsf{Q}(\mathtt{GC}_n \circ \mathtt{OR}_n))$, where $\mathtt{OR}_n$ is the boolean functions encoding the OR-function on $n$ bits.

    % Notice that in \Cref{fig:islands}, $\mathtt{Cycle}_s^{\leq k}$ reduces to $\mathtt{DirPath}^{=k-2}$, which means that the above conditional lower bound also holds for the $\mathtt{DirPath}^{=k-2}$ problem. In particular, this result and the formal version of \Cref{res:quantum_algo} (see \Cref{cor:dir-path-detection} and \Cref{tbl:dir-path-detection}) combined with the assumption that $\mathsf{Q}(\mathtt{GC}_n) \in \tilde{\Theta}(n^{2/3})$ implies that the quantum query complexity of $\mathtt{DirPath}^{=3}$ is $\tilde{\Theta}(n^{7/6})$.

    Finally, we remark that combining all our results provides conditional separations between the $\mathtt{Path}^{=k}$ and $\mathtt{DirPath}^{=k}$ problems, the $\mathtt{PromDirPath}_{s,t}^{=k}$ and $\mathtt{PromPath}_{s,t}^{=k}$ problems, and the $\mathtt{PromFindPath}_{s,t}^{=k}$ and $\mathtt{PromPath}_{s,t}^{=k}$ problems. Indeed, any non-trivial lower bound on the graph-collision problem implies a query separation between each pair of these problems. This in particular suggests that Belovs's and Reichardt's span-program-based approach can likely not be generalized to either the directed or finding settings, and, as such, that the undirected promise problems are indeed likely to be much easier than their directed and non-promise counterparts.

    \subsection{Organization}

    The document is structured as follows. We provide the preliminaries in \Cref{sec:preliminaries}, where we fix notation, and provide several results from prior work that we will use throughout the text. Then, we dedicate one section to each main result separately. That is, in \Cref{sec:reductions}, we provide all the randomized reductions between the path- and cycle-containment problems, in \Cref{sec:dir-path-algo}, we develop the novel quantum algorithm for the $\mathtt{DirPath}^{=k}$-problem, and \Cref{sec:cycle-leq-k-algo}, we develop an improved algorithm for the $\mathtt{Cycle}^{\leq k}$-problem. Finally, in \Cref{sec:lower-bounds}, we provide the novel lower-bound proofs based on the graph-collision problem.

    \section{Preliminaries}
    \label{sec:preliminaries}

    \subsection{Notation}

    Let $\N = \{1, 2, \ldots\}$ be the set of all positive integers, and $\N_0 = \N \cup \{0\}$. For any integer $n \in \N$, we define $[n] \coloneqq \{1,\ldots, n\}$ and $[n]_0 \coloneqq \{0,1,\ldots, n\}$.
    
    Let $g : \R \supseteq D_g \to \R_{\geq 0}$. We define $O(g)$ to be the set of functions $f : D_g \supseteq D_f \to \mathbb{C}$ such that there exist $C,M > 0$ such that for all $x \in D_f$, $x > M \Rightarrow |f(x)| \leq C \cdot g(x)$. Similarly, we define $\widetilde{O}(g)$ to be the set of functions $f : D_g \supseteq D_f \to \mathbb{C}$ such that there exist $C,M,k > 0$ such that for all $x \in D_f$, $x > M \Rightarrow |f(x)| \leq C \cdot g(x) \cdot \log^k(g(x))$. We define $\Omega(g)$ and $\widetilde{\Omega}(g)$ analogously, but with the conditions being $|f(x)| \geq C \cdot g(x)$ and $|f(x)| \geq C \cdot g(x) / \log^k(g(x))$, respectively. Finally, we write $\Theta(g) := O(g) \cap \Omega(g)$, and $\widetilde{\Theta}(g) := \widetilde{O}(g) \cap \widetilde{\Omega}(g)$.

    We denote undirected graphs by $G = (V,E)$, where every edge $\{v,w\} = e \in E$ is a set of two vertices, with $v,w \in V$. Similarly, we denote directed graphs by $G = (V,A)$, where every arc $(v,w) = a \in A$ is an ordered pair of vertices, with $v,w \in V$. If $V' \subseteq V$, then $G - V'$ denotes the graph $G$ with all vertices from $V'$ and their adjacent edges removed. All (undirected/directed) graphs considered in this work are simple, i.e., there are no duplicate edges and arcs. However, in the directed setting, we do allow for both arcs in opposite directions to exist simultaneously between two vertices.

    We consider the adjacency-matrix model, where any algorithm can make queries to the entries of the adjacency matrix of an input graph $G$. Equivalently, for an undirected graph, the algorithm can supply a set of two vertices to an oracle, that outputs whether an edge exists between the two, and similarly in the directed setting the algorithm can ask whether an arc exists between two vertices from an ordered pair. In the randomized setting, the algorithm can employ randomness to select which vertices to query. In the quantum setting, the vertex pairs can be queried coherently in superposition, e.g., for an undirected graph $G = (V,E)$, the oracle behaves as a unitary $O_G$ that acts as
    \[O_G : \ket{v,w}\ket{b} \mapsto \ket{v,w}\ket{b \oplus [\{v,w\} \in E]}, \qquad \forall v,w \in V, b \in \{0,1\}.\]
    We refer the reader to a quantum computing textbook for more details on the quantum computational model, e.g.,~\cite{nielsen2010quantum}. The randomized and quantum query complexities of a given problem are the minimal numbers of queries required to solve the problem in these respective computational models.

    Finally, we use regular expressions to denote sets of problems. That end, for two expressions $\mathtt{A}$ and $\mathtt{B}$, we write $\mathtt{(A|B)}$ to denote the set of expressions $\{\mathtt{A},\mathtt{B}\}$. We use this in combination with concatenation, i.e., $\mathtt{A(B|C)}$ denotes the set $\{\mathtt{AB},\mathtt{AC}\}$. In particular, some of the expressions can be empty, i.e., $\mathtt{(|A)B}$ denotes the set $\{\mathtt{B},\mathtt{AB}\}$.

    \subsection{Subroutines}

    We begin with a quantum subroutine that approximates the hamming weight of a Boolean string with a an error parameter $\epsilon$.

    \begin{theorem}[{Approximate quantum counting~\cite[Theorem~15]{brassard00}}] \label{thm:approxcounting}
        Let $f:\{0,1\}^n \to \{0,1\}$ and $\epsilon \in (0, 1]$. Let $t = |f^{-1}(1)|$. Then, there is a quantum algorithm that outputs $\tilde{t}$ such that $|t-\tilde{t}| \leq \epsilon t$ with probability at least $1 - \delta$ using $O\left(\frac{1}{\epsilon} \sqrt{n/t} \log \frac{1}{\delta}\right)$ queries to $f$.   
    \end{theorem}

    We will mainly need to use the approximate counting algorithm for approximating the number of edges in a graph and approximating the degree of a vertex. For that reason, we will need the following corollaries.

    \begin{corollary}[{\cite[Corollary~4.1]{childs2012quantum}}]
    \label{cor:approx_counting} 
        Let $A: \{0,1\}^{{n}\choose{2}} \to \{0,1\}$ be an adjacency matrix representation of an undirected graph $G$, $\epsilon \in (0, 1]$ be a constant and $\overline{m}$ be an integer. Then, there is a quantum algorithm \textsc{ApproxEdgeCount} that accepts $G$ if the number of edges $|E(G)| = |A|$ in $G$ are at least $(1+\epsilon) \overline{m}$ and rejects $G$ if $|E(G)|$ is at most $(1-\epsilon) \overline{m}$ with probability at least $1 - \delta$ using $O\left(\frac{1}{\epsilon} \sqrt{n^2/\overline{m}} \log \frac{1}{\delta}\right)$ queries to $A$.  
    \end{corollary}

    \begin{corollary}
    \label{cor:approx_degree_counting} 
        Let $A: \{0,1\}^{{n}\choose{2}} \to \{0,1\}$ be an adjacency matrix representation of an undirected graph $G$, $v$ be a vertex in $G$, $\epsilon \in (0, 1]$ be a constant and $\overline{d}$ be an integer. Then, there is a quantum algorithm \textsc{ApproxDegreeCount} that accepts $(G,v)$ if the degree $\deg(v)$ of $v$ in $G$ are at least $(1+\epsilon) \overline{d}$ and rejects $G$ if $\deg(v)$ is at most $(1-\epsilon) \overline{d}$ with probability at least $1 - \delta$ using $O\left(\frac{1}{\epsilon} \sqrt{n/\overline{d}} \log \frac{1}{\delta}\right)$ queries to $A$.  
    \end{corollary}

    We will also need a subroutine that retrieves all the marked elements of a function.

    \begin{lemma}[{\cite[Lemma~4.1]{childs2012quantum}}] \label{lem:find_all_marked_elements}
        Let $f:[n] \mapsto \{0,1\}$ and $t = |f^{-1}(1)|$. Then, there is a quantum algorithm that computes $f^{-1}(1)$ using $O(\sqrt{n(t+1)})$ queries to $f$.
    \end{lemma}

    The following lemma relates the query complexity of problems that can be related by a randomized reduction with one-sided error.

    \begin{lemma} \label{lem:1-sided_error}
        Let $\mathtt{P}$ and $\mathtt{P}'$ be arbitrary decision problems. Let $G$ be an instance of $\mathtt{P}$ and let $G'$ be obtained from $G$ via an efficient randomized reduction, i.e., each query to $G'$ can be made using $O(1)$ queries to $G$. Suppose that if $G$ is a \textsc{yes} instance of $\mathtt{P}$, then $G'$ is a \textsc{yes} instance of $\mathtt{P}'$ with probability at least $p \in \Omega(1)$, and if $G$ is a \textsc{no} instance of $\mathtt{P}$, then $G'$ is a \textsc{no} instance of $\mathtt{P}'$ with probability $1$. Then, for $\mathsf{M} \in \{\mathsf{R},\mathsf{Q}\}$, we have $\mathsf{M}(\mathtt{P}) \in O(\mathsf{M}(\mathtt{P}'))$.
    \end{lemma}

    \begin{proof}
        Let $\mathcal{A}$ be the optimal (randomized/quantum) algorithm for the problem $\mathtt{P}'$. We first construct a new algorithm $\mathcal{B}$ with majority voting, where we run $\mathcal{A}$ on the same input a total of $\Theta(\log(1/p))$
        times, so that it decides between \textsc{yes} and \textsc{no} inputs with probability at least $1-p/4$.
        
        Now, we run the algorithm $\mathcal{B}$ on an instance $G'$ generated through the randomized reduction. Then, if $G$ is a positive input, this procedure will succeed with probability at least $p(1-p/4) > 3p/4$, whereas if $G$ is a negative input, this procedure will succeed with probability at most $p/4$. Thus, by running this procedure a total of $\Theta(1/p)$ times, we can distinguish between these two cases with high probability. As such, the total number of calls to $\mathcal{A}$ is $\Theta((1/p)\log(1/p))$, which is constant as $p \in \Omega(1)$. Thus, $\mathsf{M}(\mathtt{P}) \in O(\mathsf{M}(\mathtt{P}'))$, for $\mathsf{M} \in \{\mathsf{Q},\mathsf{R}\}$.
    \end{proof}

    %\todo{Add a little remark that the oracle query must also be efficient for this result to work -- but in the graph setting this is usually the case, as long as the existence of an edge in input $G'$ only depends on constantly many edges in $G$. In our reductions, we will express the new edge set using the original edge set such that checking if there is a specific edge in the new edge set will need querying $O(1)$ queries to the original edge set}

    \begin{remark}
        Note that, to be able to argue sub-quadratic quantum query bounds for graph problems on $n$ vertices through randomized reductions, these reductions must be query-efficient. In the reductions we use throughout this work, we express the edge set of the new graph in terms of the edge set of the original graph, such that checking a specific edge in the new graph requires $O(1)$ queries to the edge set of the original graph, unless explicitly stated otherwise.
    \end{remark}

    \subsection{Quantum walks}

    Throughout this paper, we will be making use of the MNRS-framework for quantum walks, developed by Magniez, Nayak, Rolland and Santha~\cite{magniez2007search}. We recall their result here.

    \begin{theorem}[{\cite[Theorem~3]{magniez2007search}}]
        \label{thm:mnrs}
        Let $P$ be a reversible ergodic Markov chain over a state space $V$, with spectral gap $\delta$. Let $M \subseteq V$ be a marked subset, such that if $M \neq \varnothing$, then $\P_{v \sim \pi}[v \in M] \geq \varepsilon > 0$, where $\pi$ is the stationary distribution of $P$. Let $\mathsf{S}$ be such that the number of queries required to prepare a q-sample of $\ket{\pi}$ up to norm-error at most $\delta$ is $\widetilde{O}(\mathsf{S}\mathrm{polylog}(1/\delta))$. Similarly, let $\mathsf{U}$ be such that the number of queries required to prepare a q-sample over the neighbors of a given vertex $v \in V$ with norm-error at most $\delta$ is $\widetilde{O}(\mathsf{U}\mathrm{polylog}(1/\delta))$. Finally, let $\mathsf{C}$ be the number of queries required to check whether a given state $v \in V$ is marked, with high probability. Then, we can decide whether $M$ is empty or non-empty with a total number of queries that satisfies
        \[\widetilde{O}\left(\mathsf{S} + \frac{1}{\sqrt{\varepsilon}}\left(\frac{1}{\sqrt{\delta}}\mathsf{U} + \mathsf{C}\right)\right).\]
    \end{theorem}

    On several occasions, we will recursively use this construction, which is known as a nested quantum walk. This has been considered ample times in previous works, for instance~\cite{childs2012quantum,jeffery2013nested}. These constructions typically suffer from polylogarithmic overhead, but recent work has shown that this can be avoided in certain settings~\cite{jeffery2022quantum,cornelissen2025quantum}. We leave using these more recent techniques to remove the polylogarithmic overhead from our algorithms for future work.

    \section{Reductions between subgraph containment problems}
    \label{sec:reductions}

    We start by formally introducing the different versions of the subgraph-containment problems we consider in this work.

    \begin{definition}[Subgraph containment]
        We define several computational problems.
        \begin{enumerate}
            \item Let $H = (V_H,E_H)$ be a known undirected graph, and suppose we have access to an undirected graph $G = (V,E)$ through adjacency-matrix queries. Let $m \in \N_0$, $v_1, \dots, v_m \in V_H$ and $s_1, \dots, s_m \in V$ all distinct. The \textit{subgraph containment problem} asks whether there exists an injection $\varphi : V_H \to V$, such that $\varphi : v_j \mapsto s_j$, for all $j \in [m]$, and such that the induced map $\varphi' : \binom{V_H}{2} \to \binom{V}{2}$ acting as $\{v,w\} \mapsto \{\varphi(v),\varphi(w)\}$ satisfies $\varphi'(E_H) \subseteq E$. The resulting boolean function, that evaluates to $1$ if and only if such an injection exists, is denoted by
            \[\mathtt{Subgraph}^{H,v_1, \dots, v_m}_{V,s_1, \dots, s_m} : \{0,1\}^{\binom{V}{2}} \to \{0,1\},\]
            and we use the following shorthand notation for the sequence of boolean functions
            \[\mathtt{Subgraph}^H_{v_1, \dots, v_m} := (\mathtt{Subgraph}^{H,v_1, \dots, v_m}_{[n],1,\dots, m})_{n=1}^{\infty}.\]
            Whenever $m = 0$, we refer to the subgraph containment problem as \textit{unrestricted}, in which case the boolean function is simply denoted by $\mathtt{Subgraph}^H_V$, and the sequence of boolean functions by $\mathtt{Subgraph}^H$. On the other hand, when $m \geq 1$, we say that the problem is \textit{restricted}.
            % \item The \textit{(unrestricted) subgraph containment problem} asks whether there exists an injective map $\varphi : V_H \to V$ such that the induced map $\varphi' : \binom{V_H}{2} \to \binom{V}{2}$ satisfies $\varphi'(E_H) \subseteq E$. We write $\mathtt{Subgraph}_V^H : \{0,1\}^{\binom{V}{2}} \to \{0,1\}$ for the resulting boolean function, and we let $\mathtt{Subgraph}^H := (\mathtt{Subgraph}_{[n]}^H)_{n=1}^{\infty}$.
            % \item In the \textit{restricted subgraph containment problem}, we additionally require that $\varphi : v_j \mapsto s_j$, for all $j \in [m]$. We denote the resulting boolean function as $\mathtt{Subgraph}^{H,(v_1, \dots, v_m)}_{V,(s_1, \dots, s_m)} : \{0,1\}^{\binom{V}{2}} \to \{0,1\}$, and we let $\mathtt{Subgraph}^H_{v_1, \dots, v_m} = (\mathtt{Subgraph}^{H,(v_1, \dots, v_m)}_{[n],(1,\dots,m)})_{n=1}^{\infty}$.
            \item We define the restricted and unrestricted subgraph containment problems analogously for directed graphs, i.e., if $H = (V_H,A_H)$ is a directed graph and we have adjacency-matrix-query access to a directed graph $G = (V,A)$, we define the boolean function $\mathtt{Subgraph}^H_V : \{0,1\}^{\{(v,w) : v,w \in V : v \neq w\}} \to \{0,1\}$, and similarly in the restricted setting.
            \item Finally, we define finding versions of these problems, which ask to output $0$ if the subgraph $H$ is not contained in $G$, and the injective map $\varphi$ otherwise. Formally, in the undirected unrestricted setting, this is a relation $\mathtt{FindSubgraph}_V^H \subseteq \{0,1\}^{\binom{V}{2}} \times (\{0\} \cup \{\varphi : V_H \to V \text{ injective}\})$, and similarly for the restricted and directed settings.
        \end{enumerate}
    \end{definition}

    We can now phrase many graph problems formally as subgraph containment problems. For instance, if we let $H$ be an undirected cycle on three vertices, then the sequence of boolean functions $\mathtt{Subgraph}^H$ represents the triangle finding problem, where we have the bounds $\mathsf{Q}(\mathtt{Subgraph}^H) \in O(n^{5/4}) \cap \Omega(n)$~\cite{le2014improved,carette2020extended}.

    In this work, we restrict our attention to subgraph-containment problems where the subgraph is of constant size. In this setting, we can prove several elementary structural properties of these problems.

    \begin{lemma}
        \label{lem:subgraph-wlog-connected}
        Let $H$ be a graph of constant size, $m \in \N_0$ and $v_1, \dots, v_m \in V_H$. Suppose that $H_1, \dots, H_{\ell}$ are the connected components of $H$, and we write $i^{\ell'}_1, \dots, i^{\ell'}_{m_{\ell'}} \in [m]$ such that $v_{i^{\ell'}_1}, \dots, v_{i^{\ell'}_{m_{\ell'}}}$ are the vertices from $v_1, \dots, v_m$ that are contained in $H_{\ell'}$, for all $\ell' \in [\ell]$. Then, for $\mathsf{M} \in \{\mathsf{Q},\mathsf{R}\}$, we have
        \[\mathsf{M}\left(\mathtt{Subgraph}^H_{v_1, \dots, v_m}\right) \in \Theta\left(\max_{\ell' \in [\ell]}\left\{\mathsf{M}\left(\mathtt{Subgraph}^{H_{\ell'}}_{v_{i^{\ell'}_1}, \dots, v_{i^{\ell'}_{m_{\ell'}}}}\right)\right\}\right).\]
    \end{lemma}

    \begin{proof}
        For the upper bound, we randomly divide the free vertices $[n] \setminus \{1, \dots, m\}$ into $\ell$ groups $I_1, \dots, I_{\ell}$ of size $\Theta(n/\ell)$, and we define the subgraphs $G_{\ell'}$ as the input graph restricted to $I_{\ell'} \cup \{i^{\ell'}_1, \dots, i^{\ell'}_{m_{\ell'}}\}$, for all $\ell' \in [\ell]$. If $H$ was not present in $G$, then there must be an $\ell' \in [\ell]$ such that $H_{\ell'}$ is not present in $G_{\ell'}$. On the other hand, if $H$ was present in $G$, then with constant probability, $H_{\ell'}$ will be contained in $G_{\ell'}$ for all $\ell' \in [\ell]$. Thus, we now use \Cref{lem:1-sided_error} to argue that
        \[\mathsf{M}\left(\mathtt{Subgraph}^{H,v_1, \dots, v_m}_{[n],1,\dots,m}\right) \in O\left(\mathsf{M}\left(\bigwedge_{\ell'=1}^{\ell} \mathtt{Subgraph}^{H_{\ell'},v_{i_1^{\ell'}}, \dots, v_{i_{m_{\ell'}}^{\ell'}}}_{I_{\ell'} \cup \{i^{\ell'}_1, \dots, i^{\ell'}_{m_{\ell'}}\}, i_1^{\ell'}, \dots, i_{m_{\ell'}}^{\ell'}}\right)\right).\]
        Finally, we use that $\ell$ is constant, and $\mathsf{M}(\mathtt{P}_1 \land \mathtt{P}_2) \in \Theta(\max\{\mathsf{M}(\mathtt{P}_1), \mathsf{M}(\mathtt{P}_2)\})$, for total problems $\mathtt{P}_1$ and $\mathtt{P}_2$.

        For the lower bound, suppose that $\ell' \in [\ell]$ maximizes the expression on the right-hand side. Suppose we have an input on $n$ vertices for the $\mathtt{Subgraph}^{H_{\ell'}}_{[n],v_{i_1^{\ell'}}, \dots, v_{i_{m_{\ell'}}^{\ell'}}}$ problem. Then, we add a constant number of vertices $n_0 \in \Theta(1)$ containing the subgraphs $H_{\ell''}$ for all $\ell'' \in [\ell] \setminus \{\ell'\}$. Now, the new input contains $H$ if and only if the original input contains $H_{\ell'}$, and so $\mathsf{M}(\mathtt{Subgraph}^H_{[n+n_0],v_1,\dots,v_m}) \geq \mathsf{M}(\mathtt{Subgraph}_{[n],v_{i_1^{\ell'}}, \dots, v_{i_{m_{\ell'}}^{\ell'}}})$.
    \end{proof}

    The above lemma informs us that we can without loss of generality focus on the case where $H$ is connected. Next, we provide a full classification of these problems in the randomized setting, and a similar partial classification in the quantum setting. This lemma is folklore (see for instance \cite[Proposition~4]{belovs2012span} for the third item), but we provide a formal statement and proof here for completeness.

    \begin{lemma}
        \label{lem:subgraph-characterization}
        Let $H$ be a connected non-empty graph of constant size, and let $m \in \N_0$ with $v_1, \dots, v_m \in V_H$. Then, 
        \begin{enumerate}
            \item If $\{v_1, \dots, v_m\} = V_H$, then $\mathsf{Q}(\mathtt{Subgraph}^H_{v_1, \dots, v_m}) \in \Theta(1)$ and $\mathsf{R}(\mathtt{Subgraph}^H_{v_1, \dots, v_m}) \in \Theta(1)$.
            \item If $H - \{v_1, \dots, v_m\}$ contains just isolated vertices, then $\mathsf{Q}(\mathtt{Subgraph}^H_{v_1, \dots, v_m}) \in \Theta(\sqrt{n})$ and $\mathsf{R}(\mathtt{Subgraph}^H_{v_1, \dots, v_m}) \in \Theta(n)$.
            \item If $H - \{v_1, \dots, v_m\}$ contains an edge, then $\mathsf{Q}(\mathsf{Subgraph}^H_{v_1, \dots, v_m}) \in \Omega(n)$ and we have $\mathsf{R}(\mathtt{Subgraph}^H_{v_1, \dots, v_m}) \in \Theta(n^2)$.
        \end{enumerate}
    \end{lemma}

    \begin{proof}
        In the first case, there is only one possibility for the injective map, and so we simply (deterministically) check whether all the required edges exist in the input.

        In the second case, let $V_H \setminus \{v_1, \dots, v_m\} = \{w_1, \dots, w_{\ell}\}$. Similar to the proof of \Cref{lem:subgraph-wlog-connected}, the query complexity is upper bounded by searching for an image of each of the vertices $w_{\ell'}$ individually, for $\ell' \in [\ell]$. As for each $\ell' \in [\ell]$, checking whether a given vertex plays the role of $w_{\ell'}$ takes $O(1)$ (deterministic) queries, the total query complexity is $O(\sqrt{n})$ in the quantum setting and $O(n)$ in the randomized setting.

        In the third case, we have the trivial upper bound of $O(n^2)$ in the randomized setting, because we can simply learn the entire input with $\binom{n}{2} \in \Theta(n^2)$ queries.

        For the lower bounds, observe in the second case that we can consider an input that contains all edges except for the ones that are connected to $w_1$. Now, we can reduce a search problem of size $n$ to this subgraph containment problem, and so we find that $\mathsf{Q}(\mathtt{Subgraph}^H_{v_1, \dots, v_m}) \in \Omega(\sqrt{n})$ and $\mathsf{R}(\mathtt{Subgraph}^H_{v_1, \dots, v_m}) \in \Omega(n)$.

        Finally, in the third case, we can embed all of $H$ except for the connections between $v$ and $w$, with $\{v,w\} \in E_{H - \{v_1, \dots, v_m\}}$. Now, we add $n$ candidate vertices for $v$ and $n$ candidate vertices for $w$, and we connect them up to their respective neighbors in $H$. Now, we can embed a search problem on $\binom{n}{2}$ items into the bipartite graph between the candidate vertices for $v$ and $w$. Hence, we find that $\mathsf{Q}(\mathtt{Subgraph}^H_{v_1, \dots, v_m}) \in \Omega(\sqrt{\binom{n}{2}}) = \Omega(n)$ and $\mathsf{R}(\mathtt{Subgraph}^H_{v_1, \dots, v_m}) \in \Omega(\binom{n}{2}) = \Omega(n^2)$.
    \end{proof}

    Note that the above lemma completely characterizes the randomized query complexity of all detection versions of undirected subgraph containment problems. The simplicity of this characterization is in stark contrast with the quantum setting, where such a characterization seems elusive.

    Next, we provide some generic reductions between the detection and finding versions of the subgraph containment problems, and between the undirected and directed versions.
    
    \begin{lemma}
        \label{lem:find-detect-subgraph}
        Let $H$ be a connected non-empty graph of constant size, $m \in \N_0$, and $v_1, \dots, v_m \in V_H$. Then, for $\mathsf{M} \in \{\mathsf{Q},\mathsf{R}\}$, we have the characterization $\mathsf{M}(\mathtt{Subgraph}^H_{v_1, \dots, v_m}) \in \Theta(\mathsf{M}(\mathtt{FindSubgraph}^H_{v_1, \dots, v_m}))$.
    \end{lemma}

    \begin{proof}
        It is clear that $\mathsf{M}(\mathtt{Subgraph}^H_{v1,\dots,v_m}) \in O(\mathsf{M}(\mathtt{FindSubgraph}^H_{v_1, \dots, v_m}))$, since we can run an algorithm for the finding problem, and then output $1$ if its output was non-zero. Thus, it remains to prove the reverse direction. In the randomized setting, it is clear that the finding versions of the subgraph containment problems follow the same classification as in \Cref{lem:subgraph-characterization}, since in the second case, we can learn the entire neighborhoods for all the vertices $v_1, \dots, v_m$ with $O(n)$ queries, and that is sufficient to find the subgraph, if it exists. Thus it remains to focus on the quantum setting.
        
        To that end, observe that if we are in the first case of \Cref{lem:subgraph-characterization}, there is only one possible place where $H$ could be embedded in the input, and so the finding version of the subgraph-containment problem also has constant query complexity as well. thus without loss of generality we are in the second or third case in \Cref{lem:subgraph-characterization}. Now, we describe a randomized reduction that solves the finding problem, using the detection algorithm as a black box. To that end, suppose that for every $n \in \N$, the optimal algorithm makes $T_n$ quantum queries to the input to solve the $\mathsf{Subgraph}_{[n],1, \dots, m}^{H,v_1, \dots, v_m}$-problem. It is immediate that $(T_n)_{n=1}^{\infty}$ is increasing, since we can always trivially embed an instance of $n$ vertices into an instance with $n+1$ vertices by ignoring one of the vertices in the input.
        
        Then, by considering an integer $\ell > 1$, we can embed an OR of $\ell$ disjoint instances of the boolean function $\mathtt{Subgraph}^{H,v_1, \dots, v_m}_{[n],1, \dots, m}$ into $\mathrm{Subgraph}^{H,v_1, \dots, v_m}_{[m+\ell(n-m)],1,\dots,m}$, and so
        \begin{align*}
            \mathsf{Q}(\mathtt{Subgraph}^{H,v_1, \dots, v_m}_{[m + \ell (n-m)],1,\dots,m}) &\geq \mathsf{Q}(\mathtt{OR}_{\ell} \circ \mathtt{Subgraph}^{H,v_1, \dots, v_m}_{[n],1,\dots,m}) \\
            &\in \Theta(\sqrt{\ell} \cdot \mathsf{Q}(\mathtt{Subgraph}^{H,v_1, \dots, v_m}_{[n],1, \dots, m})),
        \end{align*}
        where in the last step we use that quantum query complexity is multiplicative under composition up to constants~\cite{reichardt2011reflections}. Thus, $T_{m + \ell(n-m)} \in \Omega(\sqrt{\ell} \cdot T_n)$ and so by fixing $\ell \in \N$ large enough, we obtain that $T_{\ell n} \geq T_{m + \ell(n-m)} \geq KT_n$ for large enough $n$, for some constant $K > 1$. We will without loss of generality also assume that $\ell > 9$.

        Now, we consider the finding problem of size $\ell n$. We can randomly choose $n$ of the $\ell n$ available vertices, and remove all the others. Then, if $H$ was present in the original input, it will still be present in the new input with probability at least $\Omega(\ell^{-|V_H|}) = \Omega(1)$, where we use that $n$ is large enough. On the other hand, if $H$ was not present in the original input, it will still not be present in the new input. Thus, using the generic amplification procedure, \Cref{lem:1-sided_error}, we can with $O(T_n)$ queries build a procedure that outputs $0$ if $H$ is not present in the graph, and that outputs a subset of size $n$ that contains $H$ if it is, with probability at least $5/6$. Then, in the latter case we can subsequently run the finding routine with size $n$ on this subset, with three runs and certificate checking to have a failure probability of at most $1/6$ in this step as well. The total cost of the finding routine, which we denote by $F_{\ell n}$, then satisfies for some constant $C > 0$,
        \[F_{\ell n} \leq C \cdot T_n + 3 F_n.\]
        Now since $T_n$ satisfies the regularity condition, as we argued earlier, and since $\log_\ell(3) < 1/2$ and $T_n \in \Omega(\sqrt{n})$ by \Cref{lem:subgraph-characterization}, we can use the master theorem \cite{bentley1980general} to conclude that $F_n \in O(T_n)$, and so $\mathsf{Q}(\mathtt{FindSubgraph}^H_{v_1, \dots, v_m}) \in O(\mathsf{Q}(\mathtt{Subgraph}^H_{v_1, \dots, v_m}))$.
    \end{proof}
    
    \begin{lemma}
        \label{lem:undirected_to_directed}
        Let $H$ be a connected non-empty undirected graph of constant size, $m \in \N_0$, and $v_1, \dots, v_m \in V_H$. Let $\vec{H}$ be the same graph as $H$, but then with arcs pointed in arbitrary directions. Then, for $\mathsf{M} \in \{\mathsf{Q},\mathsf{R}\}$, we have $\mathsf{M}(\mathtt{Subgraph}^H_{v_1, \dots, v_m}) \in O(\mathsf{M}(\mathtt{Subgraph}^{\vec{H}}_{v_1, \dots v_m}))$.
    \end{lemma}

    \begin{proof}
        We provide a randomized reduction. Indeed, suppose that we have an instance to the undirected version of the subgraph containment problem. We now generate an input to the directed version of the problem by randomly associating a direction to each edge. We observe that if $H$ was not present in the undirected input, then $\vec{H}$ will definitely also not be present in the new input. On the other hand, if $H$ was present in the original input, now $\vec{H}$ will be contained in the new input with probability at least $2^{-|E_H|} \in \Omega(1)$. Thus, we can now use the generic amplification procedure \Cref{lem:1-sided_error} to solve the undirected version of the problem, with constant overhead.
    \end{proof}
    
    We note that the above lemmas also completely characterize the randomized query complexity for all directed and finding subgraph containment problems. We spend the rest of the document investigating the quantum counterparts of these results in more detail.

    Finally, we recall that \Cref{thm:general-subgraph-containment}, as proved by Lee, Magniez and Santha, strictly speaking only holds for undirected subgraph-containment problems. However, we argue that their learning-graph-based algorithm also works in the directed setting, with the same complexity.

    \begin{proposition}
        \label{prop:learning-graph-directed-subgraph-containment}
        Let $H = (V,A)$ be a directed graph, and let $H' = (V,E')$ be its undirected counterpart, i.e., $\{v,w\} \in E'$ if and only if $(v,w) \in A$ or $(w,v) \in A$. Then, $\mathsf{Q}(\mathtt{Subgraph}^H)$ is upper bounded by the expression in \Cref{thm:general-subgraph-containment} applied to $H'$.
    \end{proposition}

    \begin{proof}
        As $\mathtt{Subgraph}^{H'}$ is a boolean function on $|V|(|V|-1)$ bits, we can also interpret it as a function on $\binom{V}{2}$ $2$-tuples, where the $2$-tuple labeled by $\{v,w\}$ contains the information whether the arcs $(v,w)$ and $(w,v)$ are present in the input. Now, we have
        \[\mathtt{Subgraph}^{H'} : \{0,1\}^{\binom{V}{2}} \to \{0,1\}, \qquad \text{and} \qquad \mathtt{Subgraph}^H : (\{0,1\}^2)^{\binom{V}{2}} \to \{0,1\}.\]

        In particular, we observe that the index labels into the input are now the same. We argue that on top of that, the certificate structures for both functions, as defined in \cite[Definition~1]{belovs2014power}, are also the same. To that end, observe that both graph properties are monotone, and so all the minimal positive certificates contain exactly a single copy of $H'$ and $H$, respectively. As such, we conclude that the (non-adaptive) learning-graph complexities for both problems are also equal, due to \cite[Theorem~2]{belovs2014power}.

        Finally, we observe that the two learning-graph algorithms presented in \cite{lee2012learning} are both non-adaptive. As such, the upper bound stated in \Cref{thm:general-subgraph-containment} is in fact an upper bound on the non-adaptive learning graph complexity of the undirected subgraph-containment problem with subgraph $H'$. This means it is also an upper bound on the non-adaptive learning graph complexity for the directed subgraph-containment problem with subgraph $H$, which in turn upper bounds $\mathsf{Q}(\mathtt{Subgraph}^H)$.
    \end{proof}

    \subsection{Path-containment reductions}
    \label{sec:PathContainmentReductions}

    We introduce shorthand notations for the subgraph-containment problems where the subgraph is a line graph. We refer to these as path-containment problems.

    \begin{definition}[Path-containment problems]
        Let $k \in \N$.
        \begin{enumerate}
            \item Let $L_k$ be the undirected line graph with $k$ edges, with the endpoint vertices labeled by $s$ and $t$. We define $\mathtt{Path}^{=k} := \mathtt{Subgraph}^{L_k}$, and $\mathtt{Path}^{=k}_{s,t} := \mathtt{Subgraph}^{L_k}_{s,t}$.
            \item Let $L_k'$ be the directed version of $L_k$, where the edges are all pointing from $s$ to $t$. We define $\mathtt{DirPath}^{=k} := \mathtt{Subgraph}^{L_k'}$, and $\mathtt{DirPath}^{=k}_{s,t} := \mathtt{Subgraph}^{L_k'}_{s,t}$.
        \end{enumerate}    
        Finally, for $\mathtt{P} \in \mathtt{(|Dir)Path}_{(|s,t)}$, we write
        \[\mathtt{P}^{\leq k} := \bigvee_{\ell=1}^k \mathtt{P}^{=\ell}.\]
        We also consider the finding versions of these problems, denoted by $\mathtt{FindP}$ for all $\mathtt{P} \in \mathtt{(|Dir)Path}_{(|s,t)}^{(=k|\leq k)}$. In the $(\leq k)$-problems, it suffices to output a path for any of the subproblems $\mathtt{P}^{=\ell}$, with $\ell \in \{1, \dots, k\}$.
    \end{definition}

    We also consider promise versions of the path-containment problems.

    \begin{definition}[Promise path-containment problems]
        Let $k \in \N$.
        \begin{enumerate}
            \item The $\mathtt{PromPath}^{(=k|\leq k)}_{s,t}$-problem is the $\mathtt{Path}^{(=k|\leq k)}_{s,t}$-problem, restricted to inputs that satisfy the promise that if $s$ and $t$ are connected, then there exists a path of length exactly, or at most, $k$ between $s$ and $t$, respectively.
            \item The $\mathtt{PromDirPath}^{(=k|\leq k)}_{s,t}$-problem is the $\mathtt{DirPath}^{(=k|\leq k)}_{s,t}$-problem, restricted to inputs that satisfy the promise that if $s$ and $t$ are connected by a directed path from $s$ to $t$, then there exists a directed path of length exactly, or at most, $k$ from $s$ to $t$, respectively.
        \end{enumerate}
        The finding versions of these problems, $\mathtt{PromFindPath}_{s,t}^{(=k|\leq k)}$ and $\mathtt{PromFindDirPath}^{(=k|\leq k)}_{s,t}$ respectively, are the finding problems under the same restriction.
    \end{definition}

    Formally speaking the promise path-containment problems are also boolean functions, but on a restricted domain. As such, they are partial booelan functions, rather than total ones.

    We start by showing some elementary reductions between the promise and non-promise versions of the path-containment problems.

    \begin{lemma}
        \label{Thm:PromReducesNonProm}
        Let $\mathtt{P} \in \mathtt{(|Dir)Path}_{s,t}^{(=k|\leq k)}$.
        \begin{enumerate}
            \item We have $\mathsf{Q}(\mathtt{PromP}) \in O(\mathsf{Q}(\mathtt{P}))$.
            \item We have $\mathsf{Q}(\mathtt{FindP}) \in \Theta(\mathsf{Q}(\mathtt{PromFindP}))$.
        \end{enumerate}
    \end{lemma}

    \begin{proof}
        It is clear that $\mathsf{Q}(\mathtt{PromP}) \in O(\mathsf{Q}(\mathtt{P}))$, because the function that is to be computed in the promise version of the problem is a restriction of the one in the non-promise problem. The same follows for $\mathsf{Q}(\mathtt{PromFindP}) \in O(\mathtt{FindP})$ by comparing the relations to one another. Thus, it remains to show the reverse direction in the finding case.

        To that end, suppose that we have an instance of the $\mathtt{FindP}$-problem. Any accepting instance will also be a valid instance for the promise version as well, and so in that case the algorithm for $\mathtt{PromFindP}$ should output a valid path with high probability. In the rejecting case, though, the promise might not be satisfied, so the promise problem could output anything. However, we can simply check whether the path it outputs actually exists in the graph, and reject if not. Since in the rejecting case no path exists, this procedure always rejects, and it adds only constant additive overhead to the query complexity.
    \end{proof}

    % \begin{claim}
    %     For all path and cycle-containment problems, $\mathtt{P} \in \mathtt{(|Dir)(Path_{s,t}|Cycle_{(|s)})^{(=k|\leq k)}}$, we have $\Q(\mathtt{PromP})) \in O(\Q(\mathtt{P}))$.
    % \end{claim}
    
    % \begin{proof}
    %     These relations are straightforward because we are reducing a promise version of the problem to the non-promise version of itself.
    % \end{proof}

    Through \Cref{Thm:PromReducesNonProm,lem:find-detect-subgraph}, we conclude that the $\mathtt{Find}$ and $\mathtt{PromFind}$ versions of the path-containment problems considered above all have the same query complexity up to constants as the non-promise detection version of the problem. Thus, to understand the relative hardness of the path-containment problems introduced above, we will ignore the $\mathtt{Find}$ and $\mathtt{PromFind}$ versions in what follows. We provide a graphical overview of the remaining problems in \Cref{fig:path-problems}, with arrows representing the reductions we prove in the remainder of this section.
    
    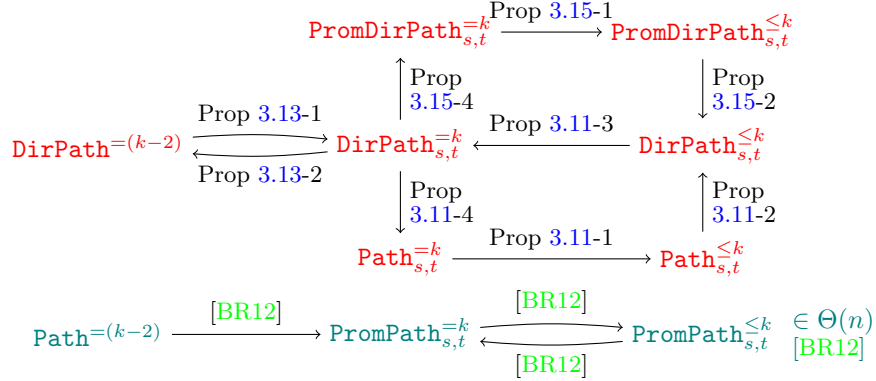
\begin{figure}[!ht]
        \centering
        \begin{tikzpicture}
            % The path picture
            \begin{scope}
                % Unrestricted problems
                \node[teal] (path) at (-4,-1.5) {$\mathtt{Path}^{=(k-2)}$};
                \node[red] (dirpath) at (-4,1) {$\mathtt{DirPath}^{=(k-2)}$};
            
                % st-problems
                \node[red] (pathst) at (0,-.5) {$\mathtt{Path}^{=k}_{s,t}$};
                \node[red] (dirpathst) at (0,1) {$\mathtt{DirPath}^{=k}_{s,t}$};
                \node[red] (pathstle) at (4,-.5) {$\mathtt{Path}^{\leq k}_{s,t}$};
                \node[red] (dirpathstle) at (4,1) {$\mathtt{DirPath}^{\leq k}_{s,t}$};
                \node[teal] (prompathst) at (0,-1.5) {$\mathtt{PromPath}^{=k}_{s,t}$};
                \node[red] (promdirpathst) at (0,2.5) {$\mathtt{PromDirPath}^{=k}_{s,t}$};
                \node[teal] (prompathstle) at (4,-1.5) {$\mathtt{PromPath}^{\leq k}_{s,t}$};
                \node[red] (promdirpathstle) at (4,2.5) {$\mathtt{PromDirPath}^{\leq k}_{s,t}$};

                \node[teal, right, align = left] at (prompathstle.east) {$\in \Theta(n)$ \\[-.2em] \small\cite{belovs2012span}};
    
                % Non-linear equivalence class
                \draw[->] (pathst) to node[above, align = center, font = \small] {Prop~\ref{thm:dir_undir_st_paths}-1} (pathstle);
                \draw[->] (pathstle) to node[right, align = left, font = \small] {Prop \\[-.2em] \ref{thm:dir_undir_st_paths}-2} (dirpathstle);
                \draw[->] (dirpathstle) to node[above, align = center, font = \small] {Prop~\ref{thm:dir_undir_st_paths}-3} (dirpathst);
                \draw[->] (dirpathst) to node[right, align = left, font = \small] {Prop \\[-.2em] \ref{thm:dir_undir_st_paths}-4} (pathst);
                \draw[bend left=5, ->] (dirpath) to node[above, align = center, font = \small] {Prop~\ref{thm:dir_st_blank_paths}-1} (dirpathst);
                \draw[bend left=5, ->] (dirpathst) to node[below, align = center, font = \small] {Prop~\ref{thm:dir_st_blank_paths}-2} (dirpath);
                \draw[->] (promdirpathst) to node[above, align = center, font = \small] {Prop~\ref{thm:prom_nonprom_dir_st_paths}-1} (promdirpathstle);
                \draw[->] (promdirpathstle) to node[right, align = left, font = \small] {Prop \\[-.2em] \ref{thm:prom_nonprom_dir_st_paths}-2} (dirpathstle);
                \draw[->] (dirpathst) to node[right, align = left, font = \small] {Prop \\[-.2em] \ref{thm:prom_nonprom_dir_st_paths}-4} (promdirpathst);

                \draw[->] (path) to node[above, align = center, font = \small] {\cite{belovs2012span}} (prompathst);
                \draw[bend left=5, ->] (prompathst) to node[above, align = center, font = \small] {\cite{belovs2012span}} (prompathstle);
                \draw[bend left=5, ->] (prompathstle) to node[below, align = center, font = \small] {\cite{belovs2012span}} (prompathst);
            \end{scope}
        \end{tikzpicture}
        \caption{Randomized reductions between path-containment problems. All the green problems can be solved in linear number of queries, and all the red problems are equivalent under randomized reductions and constant multiplicative overhead.}
        \label{fig:path-problems}
    \end{figure}

    If we set $k = 1$, or $k = 2$ then the problems $\mathtt{(|Dir)Path}_{s,t}^{(=k|\leq k)}$ are instances of the first and second classes derived in \Cref{lem:subgraph-characterization}, respectively. Thus, whenever we set $k = 1$ or $k = 2$, we already understand the exact query complexity of the problems in \Cref{fig:path-problems}, and so we assume without loss of generality that $k \geq 3$.

    For the $\mathtt{PromPath}_{s,t}^{\leq k}$-problem, Belovs and Reichardt~\cite{belovs2012span} proved that it can be solved using $O(n\sqrt{k})$ queries, which is simply $O(n)$ when $k$ is constant. They also realize that the $\mathtt{PromPath}^{=k}_{s,t}$-problem has a stronger promise, and so there is a generic reduction from this problem to the $\leq k$-version. Moreover, they provide a randomized reduction from the $\mathtt{Path}^{=(k-2)}$-problem to the $\mathtt{PromPath}^{=k}_{s,t}$-problem~\cite[Page 2]{belovs2012span}. This means that these three problems have linear query complexity, i.e., their query complexity is $\Theta(n)$ for $k \geq 3$.

    We now prove that all the other path-containment problems form an equivalence class, i.e., they have the same query complexity up to a constant. We then prove a new quantum algorithm that solves all the problems from this equivalence class in \Cref{sec:dir-path-algo}.

    We begin with showing monotonicity of $\mathsf{Q}(\mathtt{DirPath}^{\leq k}_{s,t})$ in terms of $k$. From our results later in this section, this monotonicity result also follows for all the problems in the same equivalence class as $\mathtt{DirPath}^{\leq k}_{s,t}$.

    \begin{lemma} \label{lem:path_monotonicity}
        Let $3 \leq \ell \leq k \in \N$. Then, $\mathsf{Q}(\mathtt{DirPath}^{= \ell}_{s,t}) \in O(\mathsf{Q}(\mathtt{DirPath}^{= k}_{s,t}))$.
    \end{lemma}

    \begin{proof}
        %First, we argue that for any $\ell \in [k]$, $\mathsf{Q}(\mathtt{DirPath}^{= \ell}_{s,t})) \in O(\mathsf{Q}(\mathtt{DirPath}^{= k}_{s,t}))$. 
        %The $\ell = 1$ case is trivial so suppose that $\ell > 1$. 
        Given a graph instance $G = (V,E)$ of $\mathtt{DirPath}^{= \ell}_{s,t}$, we construct a graph instance $G' =(V,E')$ of $\mathtt{DirPath}^{= k}_{s,t}$ using the color-coding technique as follows.
        Color $s$ with color $0$, color $t$ with color $k$, and for each vertex $v$ in $V \setminus \{s,t\}$, color $v$ uniformly at random from $[\ell-1]$ colors. For each $v \in V$, let $c(v)$ denote the color of $v$. Choose a color $i^* \in [\ell-1]$ and make $k-\ell$ copies of each vertex colored $i^*$. For each $j \in [k-\ell+1]$, let $v_j$ denote the $j$th copy of the vertex $v$ with $c(v) = i^*$. For all vertices $v$ with $c(v) = i^*$ and $j \in [k-\ell]$, add edges $(v_j,v_{j+1})$ in $E'$. For each edge $(u,v) \in E$ with $c(v) = c(u) + 1$, add an edge $(u,v)$ if $c(u),c(v) \neq i^*$, add an edge $(u_{k-\ell+1},v)$ in $E'$ if $c(u) = i^*$ and add an edge $(u,v_1)$ if $c(v) = i^*$. It is easy to see that if $G$ has a directed $s-t$ path of length $=\ell$, $G'$ will have a directed $s-t$ path of length $=k$ with probability at least $1/(\ell-1)^{\ell-1} \in \Omega(1)$, and if $G$ did not have a directed $s-t$ path of length $=\ell$, then $G'$ will not have a directed $s-t$ path of length $=k$. Invoking \Cref{lem:1-sided_error} proves the desired claim.
    \end{proof}

    \begin{proposition} \label{thm:dir_undir_st_paths}
        Let $3 \leq k \in \N$. Then,
        \begin{enumerate} 
            \item\label{item:pathst-pathstle} $\mathsf{Q}(\mathtt{Path}^{=k}_{s,t}) \in O(\mathsf{Q}(\mathtt{Path}^{\leq k}_{s,t})) $;
            \item\label{item:pathstle-dirpathstle} $\mathsf{Q}(\mathtt{Path}^{\leq k}_{s,t}) \in O(\mathsf{Q}(\mathtt{DirPath}^{\leq k}_{s,t})) $;
            \item\label{item:dirpathstle-dirpathst} $\mathsf{Q}(\mathtt{DirPath}^{\leq k}_{s,t})) \in O(\mathsf{Q}(\mathtt{DirPath}^{= k}_{s,t})) $;
            \item\label{item:dirpathst-pathst} $\mathsf{Q}(\mathtt{DirPath}^{= k}_{s,t})) \in O(\mathsf{Q}(\mathtt{Path}^{= k}_{s,t}))$.
        \end{enumerate}
    \end{proposition}

    \begin{proof}
        We will prove all these statements one by one via reductions.
        \begin{enumerate}
            \item \label{pt:1} 
            %We want to decide if a given graph $G = (V,E)$ has an $s-t$ path of length $=k$ given an algorithm for deciding if there is an $s-t$ path of length $\leq k$. We construct a graph $G' =(V,E')$ from $G$ using the color-coding technique as follows.
            Given a graph instance $G = (V,E)$ of $\mathtt{Path}^{=k}_{s,t}$, we construct a graph instance $G' =(V,E')$ of $\mathtt{Path}^{\leq k}_{s,t})$ using the color-coding technique as follows.
            Color $s$ with color $0$, color $t$ with color $k$, and for each vertex $v$ in $V \setminus \{s,t\}$, color $v$ uniformly at random from $[k-1]$ colors, and for each edge ${u,v} \in E$, add $\{u,v\}$ in $E'$ if $u$ and $v$ have consecutive colors. If $G$ had an $s-t$ path of length $=k$, $G'$ will have an an $s-t$ path of length $=k$ with probability at least $1/(k-1)^{k-1} \in \Omega(1)$, and if $G$ did not have an $s-t$ path of length $=k$, then $G'$ will not have it either. Moreover, any $s-t$ path in $G'$ will be of length at least $k$. 
            Invoking \Cref{lem:1-sided_error} proves the desired statement.
            %It follows that running $\mathcal{A}$ on $G'$ will output $\textsc{yes}$ with probability $\Omega(1)$ if $G$ is has an $s-t$ path of length $=k$ and $\textsc{no}$ with probability $1$ if $G$ has no such path. Therefore, repeating this process $O(1)$ times suffices.
            
            \item This follows from a reduction similar to \Cref{lem:undirected_to_directed} and invoking \Cref{lem:1-sided_error}. 
            %We want to decide if a given graph $G = (V,E)$ has an $s-t$ path of length $\leq k$ given an algorithm $\mathcal{A}$ for deciding if there is a directed $s-t$ path of length $\leq k$. We construct a graph $G' = (V,E')$ from $G$ as follows. For each edge $\{u,v\} \in E$, assign a direction uniformly at random (i.e. either have an edge directed from $u$ to $v$ or an edge directed from $v$ to $u$ with probability $1/2$ each) and add it in $E'$. If $G$ has an $s-t$ path of length $\leq k$, then $G'$ has a directed $s-t$ path of length $\leq k$ with probability at least $(1/2)^{k} = \Omega(1)$, and if $G$ did not have an $s-t$ path of length $\leq k$, then $G'$ will not have it either. It follows that running $\mathcal{A}$ on $G'$ will output $\textsc{yes}$ with probability $\Omega(1)$ if $G$ is has an $s-t$ path of length $\leq k$ and $\textsc{no}$ with probability $1$ if $G$ has no such path. Therefore, repeating this process $O(1)$ times suffices.

            \item \label{pt:3} %We want to decide if a given directed graph $G = (V,E)$ has a directed $s-t$ path of length $\leq k$ given an algorithm $\mathcal{A}$ for deciding if there is a directed $s-t$ path of length $=k$. 
            %First, we argue that $\mathcal{A}$ can be used to detect if a given graph $G$ has a directed $s-t$ path of length $=\ell$ for any $\ell \in [2,k]$. We construct a graph $G' = (V,E')$ from $G$ using the color-coding technique as follows. 
            
            %It follows that running $\mathcal{A}$ on $G'$ will output $\textsc{yes}$ with probability $\Omega(1)$ if $G$ is has an $s-t$ path of length $=\ell$ and $\textsc{no}$ with probability $1$ if $G$ has no such path. Therefore, repeating this process $O(1)$ times can boost the success probability in the $\textsc{yes}$ case to $1-O(1/k^2)$. 
            %, and only keep an edge ${u,v} \in E$ in $G'$ if $c(v) = c(u) + 1$ where $c(v)$ denotes the color of $v$.
            
            From \Cref{lem:path_monotonicity}, we have that $\mathsf{Q}(\mathtt{DirPath}^{= \ell}_{s,t})) \in O(\mathsf{Q}(\mathtt{DirPath}^{= k}_{s,t}))$ for all $3 \leq \ell \leq k$.    
            To decide if $G$ has a directed $s-t$ path of length $\leq k$, we loop over all $3 \leq \ell \leq k$ and run the above procedure to decide if $G$ has a directed $s-t$ path of length $=\ell$, and output \textsc{yes} if it outputs \textsc{yes} for any $\ell \in [k]$.

            \item \label{pt:4} 
            %We want to decide if a given directed graph $G = (V,E)$ has a directed $s-t$ path of length $= k$ given an algorithm $\mathcal{A}$ for deciding if there is an $s-t$ path of length $=k$. We construct an undirected graph $G' = (V,E')$ from $G$ using the color-coding technique similar to \cref{pt:1} (except that the undirected version of an edge in $\{u,v\} \in E$ is in $E'$ if $c(v) = c(u) + 1$). 
            Given a graph instance $G = (V,E)$ of $\mathtt{DirPath}^{= k}_{s,t}$, we construct a graph instance $G' =(V,E')$ of $\mathtt{Path}^{= k}_{s,t}$ using the color-coding technique as follows.
            Color $s$ with color $0$, color $t$ with color $k$, and for each vertex $v$ in $V \setminus \{s,t\}$, color $v$ uniformly at random from $[k-1]$ colors, and for each edge $(u,v) \in E$, add $\{u,v\}$ in $E'$ if $c(v) = c(u) + 1$.
            Now, note that if $G$ has a directed $s-t$ path of length $=k$, $G'$ will have an $s-t$ path of length $=k$ with probability at least $1/(k-1)^{k-1} \in \Omega(1)$, and if $G$ did not have a directed $s-t$ path of length $=k$, then $G'$ will not have an $s-t$ path of length $=k$. 
            Invoking \Cref{lem:1-sided_error} proves the desired statement.\qedhere
            %It follows that running $\mathcal{A}$ on $G'$ will output $\textsc{yes}$ with probability $\Omega(1)$ if $G$ is has a directed $s-t$ path of length $=k$ and $\textsc{no}$ with probability $1$ if $G$ has no such path. Therefore, repeating this process $O(1)$ times suffices.
        \end{enumerate}
    \end{proof}

    An immediate corollary of the above theorem is as follows.

    \begin{corollary}
        Let $3 \leq k \in \N$. Then, all the problems in $\mathtt{(|Dir)Path_{s,t}^{(=k|\leq k)}}$
        %$\mathtt{Path}^{=k}_{s,t}$, $\mathtt{Path}^{\leq k}_{s,t}$, $\mathtt{DirPath}^{=k}_{s,t}$, $\mathtt{DirPath}^{\leq k}_{s,t}$ 
        have the same, up to constants, quantum query complexity.
    \end{corollary}

    \begin{proposition}
        \label{thm:dir_st_blank_paths}
        Let $k \in \N$. Then,
        \begin{enumerate} 
            \item\label{item:dirpath-dirpathst} $\mathsf{Q}(\mathtt{DirPath}^{=k}) \in O(\mathsf{Q}(\mathtt{DirPath}^{= k+2}_{s,t})) $;
            \item\label{item:dirpathst-dirpath} For $k \geq 3$, $\mathsf{Q}(\mathtt{DirPath}^{=k}_{s,t}) \in O(\mathsf{Q}(\mathtt{DirPath}^{= k-2}))$.
        \end{enumerate}
    \end{proposition}

    \begin{proof}
        We will prove both these statements via reductions.
        \begin{enumerate}
            \item Given a graph instance $G = (V,E)$ of $\mathtt{DirPath}^{= k}$, we construct a graph instance $G' =(V \cup \{s,t\},E')$ of $\mathtt{DirPath}^{= k+2}_{s,t}$ using the color-coding technique as follows.
            %We want to decide if a given directed graph $G = (V,E)$ has a directed path of length $=k$ given an algorithm $\mathcal{A}$ for deciding if there is a directed $s-t$ path of length $=k+2$. We construct an undirected graph $G' = (V \cup \{s,t\},E')$ from $G$ using the color-coding technique. 
            For each vertex $v$ in $V$, color $v$ uniformly at random from $[k+1]$ colors, and for each edge $(u,v) \in E$, add it in $E'$ if $c(v) = c(u) + 1$. Also add edges $(s,v)$ for all vertices $v$ with color $1$ and edges $(v,t)$ for all vertices $v$ with color $k+1$. Now, note that if $G$ has a directed path of length $=k$, $G'$ will have a directed $s-t$ path of length $=k+2$ with probability at least $1/(k+1)^{k+1} \in \Omega(1)$, and if $G$ did not have a directed path of length $=k$, then $G'$ will not have a directed $s-t$ path of length $=k+2$. 
            Invoking \Cref{lem:1-sided_error} proves the desired statement.
            %It follows that running $\mathcal{A}$ on $G'$ will output $\textsc{yes}$ with probability $\Omega(1)$ if $G$ is has a directed path of length $=k$ and $\textsc{no}$ with probability $1$ if $G$ has no such path. Therefore, repeating this process $O(1)$ times suffices. 

            \item Given a graph instance $G = (V,E)$ of $\mathtt{DirPath}^{= k}_{s,t}$, we construct a graph instance $G' =(V \setminus \{s,t\},E')$ of $\mathtt{DirPath}^{= k-2}$ using the color-coding technique as follows.
            %We want to decide if a given directed graph $G = (V,E)$ has a directed $s-t$ path of length $=k$ given an algorithm $\mathcal{A}$ for deciding if there is a directed path of length $=k-2$. We construct an undirected graph $G' = (V \setminus \{s,t\},E')$ from $G$ using the color-coding technique.
            %Color $s$ with color $0$, color $t$ with color $k$, and 
            For each vertex $v$ in $V \setminus \{s,t\}$, color $v$ uniformly at random from $[k-1]$ colors, and for each edge $(u,v) \in E$ with $u,v \notin \{s,t\}$, add $(u,v)$ in $E'$ if $c(v) = c(u) + 1$ and either $c(u) \neq 1$ and $c(v) \neq k-1$ or $c(u) = 1$ and $(s,u) \in E$ or $c(v) = k-1$ and $(v,t) \in E$.
            %Let $S = \mathcal{N}_s$ and $T = \mathcal{N}_t$. Color all vertices in $S$ color $1$, color all vertices in $T$ color $k-1$, and for each vertex $v$ in $V \setminus (S \cup T)$, color $v$ uniformly at random from colors $[2,k-2]$. For each edge ${u,v} \in E$, add it in $E'$ if $s \neq u$, $t \neq v$ and $c(v) = c(u) + 1$. 
            Now, note that if $G$ has a directed $s-t$ path of length $=k$, $G'$ will have a directed path of length $=k-2$ with probability at least $1/(k-1)^{k-1} \in \Omega(1)$, and if $G$ did not have a directed path of length $=k$, then $G'$ will not have a directed $s-t$ path of length $=k-2$. 
            Invoking \Cref{lem:1-sided_error} proves the desired statement.\qedhere
            %It follows that running $\mathcal{A}$ on $G'$ will output $\textsc{yes}$ with probability $\Omega(1)$ if $G$ is has a directed path of length $=k$ and $\textsc{no}$ with probability $1$ if $G$ has no such path. Therefore, repeating this process $O(1)$ times suffices. 
        \end{enumerate}
    \end{proof}

    \begin{corollary}
        Let $3 \leq k \in \N$. Then, the problems in $\mathtt{DirPath}_{s,t}^{=k}$ and $\mathtt{DirPath}^{=k-2}$
        %$\mathtt{Path}^{=k}_{s,t}$, $\mathtt{Path}^{\leq k}_{s,t}$, $\mathtt{DirPath}^{=k}_{s,t}$, $\mathtt{DirPath}^{\leq k}_{s,t}$ 
        have the same, up to constants, quantum query complexity.
    \end{corollary}

    \begin{proposition} \label{thm:prom_nonprom_dir_st_paths}
        Let $3 \leq k \in \N$. Then,
        \begin{enumerate} 
            \item\label{item:promdirpathst-promdirpathstle} $\mathsf{Q}(\mathtt{PromDirPath}^{=k}_{s,t}) \in O(\mathsf{Q}(\mathtt{PromDirPath}^{\leq k}_{s,t})) $;
            \item\label{item:promdirpathstle-dirpathstle} $\mathsf{Q}(\mathtt{PromDirPath}^{\leq k}_{s,t}) \in O(\mathsf{Q}(\mathtt{DirPath}^{\leq k}_{s,t})) $;
            \item $\mathsf{Q}(\mathtt{DirPath}^{\leq k}_{s,t})) \in O(\mathsf{Q}(\mathtt{DirPath}^{= k}_{s,t})) $;
            \item\label{item:dirpathst-promdirpathst} $\mathsf{Q}(\mathtt{DirPath}^{= k}_{s,t})) \in O(\mathsf{Q}(\mathtt{PromDirPath}^{= k}_{s,t}))$.
        \end{enumerate}
    \end{proposition}

    \begin{proof}
        We will prove all these statements one by one via reductions.
        \begin{enumerate}
            \item Note that the promise of $\mathtt{PromDirPath}^{=k}_{s,t}$ is stronger than the promise of $\mathtt{PromDirPath}^{\leq k}_{s,t})$ since the \textsc{yes} and \textsc{no} instances of the former are also \textsc{yes} and \textsc{no} instances of the latter respectively. Thus, an algorithm for the latter will also be an algorithm for the former.  
            
            \item This follows directly by the definitions of $\mathtt{PromDirPath}^{\leq k}_{s,t}$ and $\mathtt{DirPath}^{\leq k}_{s,t}$.
            
            \item This follows from \cref{pt:3} of \Cref{thm:dir_undir_st_paths}.
            
            \item Given a graph instance $G = (V,E)$ of $\mathtt{DirPath}^{= k}_{s,t}$, we construct a graph instance $G' =(V,E')$ of $\mathtt{PromDirPath}^{= k}_{s,t}$ using the color-coding technique as follows.
            Color $s$ with color $0$, color $t$ with color $k$, and for each vertex $v$ in $V \setminus \{s,t\}$, color $v$ uniformly at random from $[k-1]$ colors, and for each edge $(u,v) \in E$, add $(u,v)$ in $E'$ if $c(v) = c(u) + 1$.
            Now, note that any directed $s-t$ path in $G$ of length $\neq k$ will not be preserved in $G'$. Thus, $G'$ is a valid instance of $\mathtt{PromDirPath}^{= k}_{s,t}$. Also, any directed $s-t$ path of length $=k$ in $G$ is preserved in $G'$ with probability at least $1/(k-1)^{k-1} \in \Omega(1)$. Therefore, invoking \Cref{lem:1-sided_error} proves the desired statement.\qedhere
        \end{enumerate}
    \end{proof}

    An immediate corollary of the above theorem is as follows.

    \begin{corollary}
        \label{cor:path-island}
        Let $3 \leq k \in \N$. Then, all the problems in $\mathtt{(|Prom)DirPath}_{s,t}^{(=k|\leq k)}$
        %$\mathtt{Path}^{=k}_{s,t}$, $\mathtt{Path}^{\leq k}_{s,t}$, $\mathtt{DirPath}^{=k}_{s,t}$, $\mathtt{DirPath}^{\leq k}_{s,t}$ 
        have the same, up to constants, quantum query complexity.
    \end{corollary}
    
    \subsection{Cycle-containment reductions}
    \label{sec:CycleContainmentReductions}

    We now introduce cycle-containment problems.

    \begin{definition}[Cycle-containment problems]
        Let $3 \leq k \in \N$.
        \begin{enumerate}
            \item Let $C_k$ be the undirected cycle graph with $k$ edges, and let $s$ be any vertex on the cycle. Then, $\mathtt{Cycle}^{=k} := \mathtt{Subgraph}^{C_k}$ and $\mathtt{Cycle}^{=k}_s := \mathtt{Subgraph}^{C_k}_s$.
            \item Let $C_k'$ be the directed cycle graph with $k$ arcs all pointing in the same direction, and let $s$ be any vertex on the cycle. Then, $\mathtt{DirCycle}^{=k} := \mathtt{Subgraph}^{C_k'}$ and $\mathtt{DirCycle}^{=k}_s := \mathtt{Subgraph}^{C_k'}_s$.
        \end{enumerate}
        We also consider versions of the problem where we are looking for cycles of length at most $k$, i.e., for all $\mathtt{P} \in \mathtt{(|Dir)Cycle}_{(|s)}$, we write
        \[\mathtt{P}^{\leq k} = \bigvee_{\ell=3}^k \mathtt{P}^{=k}.\]
        Finally, we also consider the finding versions of these problems. That is, for all $\mathtt{P} \in \mathtt{(|Dir)Cycle}_{(|s)}^{(=k|\leq k)}$, we consider $\mathtt{FindP}$ where the task is to output the cycle. For the $(\leq k)$-problem, outputting a cycle of any length $3 \leq \ell \leq k$ suffices.
    \end{definition}

    For simplicity, we only consider cycles of length at least $3$. Technically, one can also have a ``cycle'' of length $2$ in the directed setting, if two arcs in opposite directions are present between two vertices, but we will refrain from calling these cycles in this work.

    We also consider promise problems of these problems.

    \begin{definition}[Promise cycle-containment problems]
        Let $3 \leq k \in \N$.
        \begin{enumerate}
            \item The $\mathtt{PromCycle}^{(=k|\leq k)}$ problem is the $\mathtt{Cycle}^{(=k|\leq k)}$-problem, restricted to inputs where if there exists any cycle in the graph, there must be at least one cycle of length exactly, or at most, $k$, respectively.
            \item The $\mathtt{PromDirCycle}^{(=k|\leq k)}$ problem is the $\mathtt{DirCycle}^{(=k|\leq k)}$-problem, restricted to inputs where if there exists any directed cycle in the graph, there must be at least one directed cycle of length exactly, or at most, $k$, respectively.
            \item The $\mathtt{PromCycle}_s^{(=k|\leq k)}$ problem is the $\mathtt{Cycle}_s^{(=k|\leq k)}$-problem, restricted to inputs where if there exists any cycle in the graph that passes through $s$, there must be at least one cycle of length exactly, or at most, $k$, respectively, that passes through $s$.
            \item The $\mathtt{PromDirCycle}_s^{(=k|\leq k)}$ problem is the $\mathtt{DirCycle}_s^{(=k|\leq k)}$-problem, restricted to inputs where if there exists any directed cycle in the graph that passes through $s$, there must be at least one directed cycle of length exactly, or at most, $k$, respectively, that passes through $s$.
        \end{enumerate}
        We consider the finding versions of these problems as well, which are defined to be the finding problems under the same restrictions on the input, and denoted by $\mathtt{PromFind(|Dir)Cycle}_{(|s)}^{(=k|\leq k)}$.
    \end{definition}

    We note here that the same reductions hold between the promise and non-promise versions of the cycle-containment problems as in the path case.

    \begin{lemma}
    \label{lem:PromFindForCycle}
        Let $3 \leq k \in \N$, and $\mathtt{P} \in \mathtt{(|Dir)Cycle}_{(|s)}^{(=k|\leq k)}$.
        \begin{enumerate}
            \item We have $\mathsf{Q}(\mathtt{PromP}) \in O(\mathsf{Q}(\mathtt{P}))$.
            \item We have $\mathsf{Q}(\mathtt{FindP}) \in \Theta(\mathsf{Q}(\mathtt{PromFindP}))$.
        \end{enumerate}
    \end{lemma}

    \begin{proof}
        The proof is identical to \Cref{Thm:PromReducesNonProm}.
    \end{proof}

    Similar to the path-containment setting, the above lemma informs us that we don't need to consider the finding versions of the cycle-containment problems, since they are always equivalent to the non-promise detection versions up to randomized reductions. Thus, for all $\mathtt{P} \in \mathtt{(|Dir)Cycle}_{(|s)}^{(=k|\leq k)}$, we only consider the problems $\mathtt{P}$ and $\mathtt{PromP}$ in the remainder of this section.

    We now prove several randomized reductions between these problems, and we give a concise graphical overview of them in \Cref{fig:cycle-problems}.

    \begin{figure}[!ht]
        \centering
        \begin{tikzpicture}
            % The unrestricted cycle subgraph picture
            \begin{scope}[shift={(4,0)}]
                \node[purple] (cycle) at (6.25,2) {$\mathtt{Cycle}^{=k}$};
                \node[blue] (dircycle) at (6.25,0) {$\mathtt{DirCycle}^{=k}$};
                \node[violet] (cyclele) at (3,2) {$\mathtt{Cycle}^{\leq k}$};
                
                \node[blue] (dircyclele) at (3,0) {$\mathtt{DirCycle}^{\leq k}$};
                \node[teal] (promcycle) at (7,4) {$\mathtt{PromCycle}^{=k}$};
                \node[orange] (promdircycle) at (7,-2) {$\mathtt{PromDirCycle}^{=k}$};
                \node[teal] (promcyclele) at (3,4) {$\mathtt{PromCycle}^{\leq k}$};
                \node[orange] (promdircyclele) at (3,-2) {$\mathtt{PromDirCycle}^{\leq k}$};

                \draw[->] ([shift={(-.2,0)}]promdircycle.south) arc (-180:0:.2) node[midway, below, align = center, font = \small] {$k \mapsto k+1$ \\[-.2em] Lem \ref{thm:MonotonicityCycles}-\ref{item:PromMonotonicityDirCycle}};
                \draw[->] ([shift={(-.2,0)}]dircycle.south) arc (-180:0:.2) node[midway, below, align = center, font = \small] {$k \mapsto k+1$ \\[-.2em] Lem \ref{thm:MonotonicityCycles}-\ref{item:MonotonicityDirCycle}};
                \draw[dashed, ->] ([shift={(-.15,0)}]cyclele.north east) arc (180:-90:.15) node[midway, above, align = center, font = \small] {Lem \ref{thm:MonotonicityCycles}-\ref{item:MonotonicityCycleLeqK} \\[-.2em] $k-1 \mapsto k$};

                \draw[bend left = 5, ->] (dircycle) to node[below, align = center, font = \small] {Prop \\[-.2em] \ref{Cycle-Equivalences-Part1}-\ref{item:DirCycleEqLeqk}} (dircyclele);
                \draw[bend left = 5, ->] (dircyclele) to node[above, align = center, font = \small] {Prop \\[-.2em] \ref{Cycle-Equivalences-Part1}-\ref{item:DirCycleEqLeqk}} (dircycle);
                
                \draw[bend left = 4, ->] (promdircycle) to node[below, align = center, font = \small] {Prop \\[-.2em] \ref{Cycle-Equivalences-Part1}-\ref{item:PromDirCycleEqLeqk}} (promdircyclele);
                \draw[bend left = 4, ->] (promdircyclele) to node[above, align = center, font = \small] {Prop \\[-.2em] \ref{Cycle-Equivalences-Part1}-\ref{item:PromDirCycleEqLeqk}} (promdircycle);
                
                \draw[->] (cyclele) to node[below, align = center, font = \small] {Prop \\[-.2em] \ref{Cycle-Equivalences-Part1}-\ref{item:CycleEqLeqk}} (cycle);
                
                \draw[bend left = 4, ->] (promcycle) to node[below, align = center, font = \small] {Prop \\[-.2em] \ref{Cycle-Equivalences-Part1}-\ref{Item:PromCycleEqLeqk}} (promcyclele);
                \draw[bend left = 4, ->] (promcyclele) to node[above, align = center, font = \small] {Prop \\[-.2em] \ref{Cycle-Equivalences-Part1}-\ref{Item:PromCycleEqLeqk}} (promcycle);
                
                \draw[bend right = 10, dashed, ->] (dircycle) to node[right, align = left, font = \small] {Prop \\[-.2em] \ref{Cycle-Equivalences-Part1}-\ref{Item:DirCycleLeq}} (cycle);
                \draw[bend right = 10, ->] (cycle) to node[left, align = right, font = \small] {Prop \\[-.2em] \ref{Cycle-Equivalences-Part1}-\ref{item:ClEqKdir}} (dircycle);
                
                \draw[->] (promcyclele) to node[left, align = left, font = \small] {Prop \\[-.2em] \ref{Cycle-Equivalences-Part1}-\ref{item:promCleqkNonprom}} (cyclele);
                \draw[->] (promdircyclele) to node[left, align = left, font = \small] {Prop \\[-.2em] \ref{Cycle-Equivalences-Part1}-\ref{item:promDirCleqkNonprom}} (dircyclele);
                \draw[bend left = 5, ->] ([shift={(.35,0)}]promcycle.south) to node[right, align = left, font = \small] {Prop \\[-.2em] \ref{Cycle-Equivalences-Part1}-\ref{item:PromCEqKPromDir}} ([shift={(.35,0)}]promdircycle.north);
                \draw[->] (cyclele) to node[left, align = left, font = \small] {Prop \\[-.2em] \ref{Cycle-Equivalences-Part1}-\ref{item:CEqKDir}} (dircyclele);
            \end{scope}

            % The restricted cycle subgraph picture
            \begin{scope}[shift={(15.5,0)}]
                % s-problems
                \node[red] (cycles) at (-1,2) {$\mathtt{Cycle}^{=k}_s$};
                \node[red] (dircycles) at (-1,0) {$\mathtt{DirCycle}^{=k}_s$};
                \node[red] (cyclesle) at (2.5,2) {$\mathtt{Cycle}^{\leq k}_s$};
                \node[red] (dircyclesle) at (2.5,0) {$\mathtt{DirCycle}^{\leq k}_s$};
                \node[teal] (promcycles) at (-1,4) {$\mathtt{PromCycle}^{=k}_s$};
                \node[red] (promdircycles) at (-1,-2) {$\mathtt{PromDirCycle}^{=k}_s$};
                \node[teal] (promcyclesle) at (2.5,4) {$\mathtt{PromCycle}^{\leq k}_s$};
                \node[red] (promdircyclesle) at (2.5,-2) {$\mathtt{PromDirCycle}^{\leq k}_s$};

                \draw[bend left = 5, ->] (dircycles) to node[above, align = center, font = \small] {Prop \\[-.2em] \ref{Cycle-Equivalences-Part2}-\ref{item:DirCycleEqLeqks}} (dircyclesle);
                \draw[bend left = 5, ->] (dircyclesle) to node[below, align = center, font = \small] {Prop \\[-.2em] \ref{Cycle-Equivalences-Part2}-\ref{item:DirCycleEqLeqks}} (dircycles);

                \draw[bend left = 4, ->] (promdircycles) to node[above, align = center, font = \small] {Prop \\[-.2em] \ref{Cycle-Equivalences-Part2}-\ref{item:PromDirCycleEqLeqks}} (promdircyclesle);
                \draw[bend left = 4, ->] (promdircyclesle) to node[below, align = center, font = \small] {Prop \\[-.2em] \ref{Cycle-Equivalences-Part2}-\ref{item:PromDirCycleEqLeqks}} (promdircycles);

                \draw[bend left = 5, ->] (cyclesle) to node[below, align = center, font = \small] {Prop \\[-.2em] \ref{Cycle-Equivalences-Part2}-\ref{item:CycleEqLeqksPartA}} (cycles);
                \draw[very thick, bend left = 5, ->] (cycles) to node[above, align = center, font = \small] {Prop \\[-.2em] \ref{Cycle-Equivalences-Part2}-\ref{item:CycleEqLeqksPartB}} (cyclesle);

                \draw[bend left = 5, ->] (promcycles) to node[above, align = center, font = \small] {Prop \\[-.2em] \ref{Cycle-Equivalences-Part2}-\ref{Item:PromCycleEqLeqks}} (promcyclesle);
                \draw[bend left = 5, ->] (promcyclesle) to node[below, align = center, font = \small] {Prop \\[-.2em] \ref{Cycle-Equivalences-Part2}-\ref{Item:PromCycleEqLeqks}} (promcycles);
                
                \draw[bend left = 10, ->] (promdircycles) to node[left, align = right, font = \small] {Prop \\[-.2em] \ref{Cycle-Equivalences-Part2}-\ref{item:PromDirCycleEqksDirCycleEqks}} (dircycles);
                \draw[very thick, bend left = 10, ->] (dircycles) to node[right, align = left, font = \small] {Prop \\[-.2em] \ref{Cycle-Equivalences-Part2}-\ref{item:PromDirCycleEqksDirCycleEqks}} (promdircycles);
                
                \draw[->] (promcyclesle) to node[right, align = left, font = \small] {Prop \\[-.2em] \ref{Cycle-Equivalences-Part2}-\ref{item:promCleqkNonproms}} (cyclesle);
                \draw[->] (promdircyclesle) to node[right, align = left, font = \small] {Prop \\[-.2em] \ref{Cycle-Equivalences-Part2}-\ref{item:promDirCleqkNonproms}} (dircyclesle);

                %\draw[->] (promcycles) to node[right] {\ref{Cycle-Equivalences-Part2}-\ref{item:PromCEqKPromDirs}} (promdircycles);

                \draw[->] (cyclesle) to node[right, align = left, font = \small] {Prop \\[-.2em] \ref{Cycle-Equivalences-Part2}-\ref{item:CEqKDirs}} (dircyclesle);
                
                \draw[very thick, bend left = 10, ->] (dircycles) to node[left, align = right, font = \small] {Prop \\[-.2em] \ref{Cycle-Equivalences-Part2}-\ref{item:dirCycleEqkCycleEqks}} (cycles);
                \draw[bend left = 10, ->] (cycles) to node[right, align = left, font = \small] {Prop \\[-.2em] \ref{Cycle-Equivalences-Part2}-\ref{item:ClEqKdirs}} (dircycles);
            \end{scope}
            
            \draw[->] (promdircycles) to node[above, align = center, font = \small] {Prop \\[-.2em] \ref{thm:CrossArrow-DirCycleEqksToDirCycleEqk}} node[below] {$k+1 \mapsto k$} (promdircycle);
        \end{tikzpicture}
        \caption{Pictorial representation of the randomized reductions we prove in this section, for $k \geq 3$. The dashed connections only hold for odd values of $k$. Problems depicted in the same color have the same quantum query complexity up to constants. Note that the connections represented in bold only exist in the restricted cycle-containment version of the problem, i.e., the corresponding result cannot be obtained in the left-hand side of the picture using the same techniques.}
        \label{fig:cycle-problems}
    \end{figure}
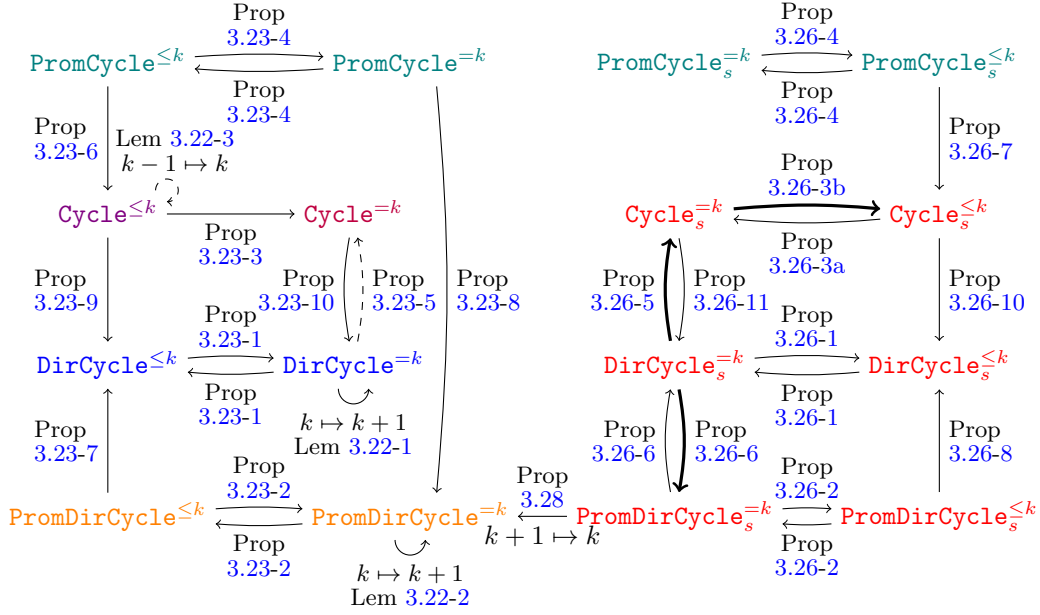

    Before we present the remaining cycle-containment problems and the reductions between them, we define some techniques similar to ones used to prove the path-containment reductions in \Cref{sec:PathContainmentReductions}, e.g., color-coding, layer insertion, but tailored towards cycle subgraph containment problems. 
 
    \begin{definition}[$\ell$-color cyclic construction for graphs ($|$through $s$)]
    \label{def:ColorCodingTechniqueForCycles} Let $\ell \in \mathbb{N}$. Let $G=(V,E)$ be a (directed or an undirected) graph. An $\ell$-color cyclic construction of $G$ ($|$through a given vertex $s$), let us denote by $G'=(V',E')$ a simple graph constructed as follows,
    \begin{enumerate} 
        \item set $V'=V$;
        \item for every $v \in V'$, assign a color from $[\ell-1]_0$ picked uniformly at random; if the $\ell$-color cyclic construction has to be through a specific vertex $s \in V'$ then color $s$ with $0$ and, for every $v \in V'\setminus \{s\}$ assign a color from $[\ell-1]$ uniformly at random;
        \item if $G$ is directed, then for every edge $(u,v) \in E$ with $c(v)=(c(u)+1) \mod \ell$, keep edge $(u,v)$ in $E'$; and, if $G$ is undirected, then for every edge $\{u,v\} \in E$ with $|c(v)-c(u)|\in \{1,\ell-1\}$ keep edge $\{u,v\}$ in $E'$. This construction additionally ensures that $G'$ is a directed graph if and only if $G$ is a directed graph.
    \end{enumerate}
    \end{definition}

    \begin{definition}[$(\ell,t)$-color cyclic construction for graphs]
    \label{def:layerInsertionCycle}
    Let $G=(V,E)$ be a (directed or an undirected) graph. Let $G'=(V',E')$ be a $\ell$-color cyclic graph constructed from $G$ ($|$through $s$) using the random construction as stated in $\Cref{def:ColorCodingTechniqueForCycles}$. The $(\ell,t)$-color cyclic construction of $G$ ($|$ through $s$), let us denote by $G''=(V'',E'')$, is a simple graph constructed (from $G'$) in the following way,
    \begin{enumerate}
        \item set $V''=V'$ and for every $v \in V''$ assign the same color that $v$ gets in $V'$;
        \item set $E''=E'$ to begin with;
        \item pick a color $i^* \in [\ell-1]$ and make $t$ copies of each vertex colored $i^*$;
        \item for each $j \in [t]$, let $v_j$ denote the $j$th copy of vertex $v$ with $c(v)=i^*$. If $G'$ is directed (undirected), then for all vertices with $c(v)=i^*$ add $(v,v_1)$ ($\{v,v_{1}\}$) in $E''$ and additionally for all $j\in [t-1]$ add edges $(v_j,v_{j+1})$ ($\{v_j,v_{j+1}\}$) in $E''$. Additionally, for all edges $(v,w) \in E''$ with $c(v)=i^*$ and $c(w)=(i^*+1) \mod \ell$, remove the edge $(v,w)$ ($\{v,w\}$) from $E''$ and instead add the edge $(v_t,w)$ ($\{v_t,w\}$) to $E''$. This construction additionally ensures that $G''$ is a directed graph if and only if $G$ (and by construction $G'$) is a directed graph.
    \end{enumerate}
    \end{definition}

    For the following problems we can show that the quantum query complexity is monotonically increasing in $k$.

    \begin{lemma}
    \label{thm:MonotonicityCycles}
    Let $3 \leq \ell \leq k \in \N$. Then, 
    \begin{enumerate}
        \item \label{item:MonotonicityDirCycle} $\Q(\mathtt{DirCycle}^{=\ell}) \in O(\Q(\DirCycleEqk))$;
        \item \label{item:PromMonotonicityDirCycle} $\Q(\mathtt{PromDirCycle}^{=\ell})\in O(\Q(\PromDirCycleEqk)$;
        \item \label{item:MonotonicityCycleLeqK} $\Q(\mathtt{Cycle}^{\leq(k-1)}) \in O(\Q(\CycleLeqk))$ whenever $k$ is an odd integer.
    \end{enumerate}
    Moreover, same monotonicity relations hold for $\mathtt{P}_{s}$ where $\mathtt{P} \in \{\mathtt{DirCycle},\mathtt{PromDirCycle},\mathtt{Cycle}\}$.
    \end{lemma}
    \begin{proof}
    We prove all these statements one by one via reductions.
    \begin{enumerate}
        \item We will prove that, for any $\ell \in [k]\setminus \{1,2\}$, $\Q(\mathtt{DirCycle}^{=\ell})\in O(\Q(\DirCycleEqk))$. Let $G = (V,E)$ be the input of $\mathtt{DirCycle}^{= \ell}$, we construct an $(\ell,k-\ell)$-color cyclic construction corresponding to $G$ as sketched in \Cref{def:layerInsertionCycle}; let us denote the resultant directed graph by $G'' =(V'',E'')$. 
        If $G$ has a directed cycle of length $=\ell$ then $G''$ will have a directed cycle of length $=k$ with probability at least $1/\ell^{\ell} \in \Omega(1)$. In any other case, $G''$ will not have any directed cycle of length $=k$. Invoking \Cref{lem:1-sided_error} proves that for any $\ell \in [k]\setminus \{1,2\}$ we have $\Q(\mathtt{DirCycle}^{=\ell}) \in O(\Q(\DirCycleEqk))$.

        \item The argument for showing $\Q(\mathtt{PromDirCycle}^{=\ell}) \in O(\Q(\PromDirCycleEqk))$ for any $\ell \in [k]\setminus \{1,2\}$ is via the same reduction that is stated in the proof of \Cref{item:MonotonicityDirCycle}. What remains to be established is that the promise for $\PromDirCycleEqk$ is not ``badly'' violated for any of the reduced instances. Let $G$ be the input instance of $\mathtt{PromDirCycle}^{=\ell}$. Let $G''$ denote a $(\ell,k-\ell)$-color cyclic construction corresponding to $G$. 
        \begin{enumerate}
            \item \textbf{$G$ is a \textsc{no} instance  of $\mathtt{PromDirCycle}^{=\ell}$.} This means there is no directed cycle of \textit{any} length in $G$ and, note that the construction of $G''$ doesn't create any new cycles. This means that $G''$ constructed from such a $G$ will also not have any cycle. Hence, a \textsc{no} instance of $\mathtt{PromDirCycle}^{=\ell}$ maps to a well-defined promise-satisfying \textsc{no} instance of $\PromDirCycleEqk$.
            \item \textbf{$G$ is a \textsc{yes} instance  of $\mathtt{PromDirCycle}^{=\ell}$.} First observe that, When $G$ is a \textsc{yes} instance of $\mathtt{PromDirCycle}^{=\ell}$ then with probability $\Omega(1)$ the constructed $G''$ will also be a promise-satisfying \textsc{yes} instance of $\PromDirCycleEqk$. For the remaining fraction, the $\ell$-length cycle from $G$ is not preserved due to the random color coding but other length $>\ell$ cycles could be present. Which would then translate to no cycle of length $=k$ to be present in $G''$ but length $>k$ cycles being present in $G''$, and this is \textit{not} in the promise of $\PromDirCycleEqk$. 
        \end{enumerate}
        In spite of the occasional violation of promise in the \textsc{yes} case, the reduction can be made to work. When $G$ is a \textsc{no} instance the constructed $G''$ fully satisfies the promise for $\PromDirCycleEqk$ while being a \textsc{no} instance itself. And, $G''$ satisfies the promise for $\PromDirCycleEqk$ while being a \textsc{yes} instance with probability $\Omega(1)$ when $G$ is a \textsc{yes} instance. Therefore, running an algorithm for $\PromDirCycleEqk$ on $G''$ for $O(1)$ many times is sufficient to deduce whether or not $G$ is a \textsc{yes} or a \textsc{no} instance with high probability. This concludes the argument.
        %\footnote{\Subha{The algorithms in the quantum setting are bounded-error. This needs to accounted for especially for the promise settings like this one. Check other sections as well to make sure that it is addressed.}\Arjan{I updated the statement of \Cref{lem:1-sided_error} a bit, so that it says the following: if there is a randomized reduction that always maps \textsc{NO} to \textsc{NO}, and with probability $\Omega(1)$ maps \textsc{YES} to \textsc{YES}, then the bounded-error algo can be used as a black box. Notably, it doesn't assume that in the \text{YES}-case, the reduction always yields a valid input to the decision problem we reduce to. I think this should address the concern. ;)}\Subha{Thanks.}}
        \item Let $G=(V,E)$ be an undirected graph that is input to $\mathtt{Cycle}^{\leq(k-1)}$. Let $\mathcal{A}$ denote an algorithm for $\CycleLeqk$. We loop over $\ell \in [k-1]$ and run $\mathcal{A}$ on $G''(\ell)$ where $G''(\ell)$ denotes a random $(\ell,k-\ell)$-cyclic construction of $G$. If $\mathcal{A}$ outputs \textsc{no} for $G''(\ell)$ for all $\ell \in [k-1]$ (more precisely, for constant many constructions of $G''(\ell)$ for each $\ell$) then we know that $G$ has cycles only of length $>k-1$ and therefore we output \textsc{no}. On the other hand, if $\mathcal{A}$ outputs \textsc{yes} on any of the above runs then that is because either there was a $\leq (k-1)$ cycle in $G$ that got encoded in one of the $G''(\ell)$ with probability $> 1/\ell^{\ell} \in \Omega(1)$ for at least one $\ell \in [k-1]$ or there was length-$k$ cycle in the original graph $G$ that survived in $G''$. Fortunately, enforcing $k$ is odd ensures that any valid length-$k$ cycle in $G''$ passes through the color-coded layers and as well as the inserted dummy layers which then corresponds to some length-$\ell$ (where $\ell < k$) cycle of $G$.  Without enforcing that $k$ is odd we cannot get rid of the original length-$k$ cycles of $G$. Finally, invoking \Cref{lem:1-sided_error} concludes the proof. 
    \end{enumerate}

    The argument for monotonicity for $\mathtt{P}_{s}$ where $\mathtt{P} \in \{\mathtt{DirCycle},\mathtt{PromDirCycle},\mathtt{Cycle}\}$, works in the same way as above except for the difference that while constructing $(\ell,k-\ell)$-color cyclic construction, the color coding construction goes through $s$ as described in \Cref{def:ColorCodingTechniqueForCycles}. Note that, this difference changes the exact probability of success for mapping the \textsc{yes} instances to \textsc{yes} instances but is still $\Omega(1)$.
    \end{proof}

    Using these definitions and lemmas we will now prove a few cycle-containment reductions as illustrated in \Cref{fig:cycle-problems}. 

    \begin{proposition}\label{Cycle-Equivalences-Part1}
    Let $3 \leq k \in \N$ be a constant. Then,
    \begin{enumerate}
        \item \label{item:DirCycleEqLeqk} $\Q(\DirCycleEqk) \in \Theta(\Q(\DirCycleLeqk))$;
        \item \label{item:PromDirCycleEqLeqk} $\Q(\PromDirCycleLeqk) \in \Theta(\Q(\PromDirCycleEqk))$;
        \item \label{item:CycleEqLeqk} $\Q(\CycleLeqk) \in O(\Q(\CycleEqk))$; 
        \item \label{Item:PromCycleEqLeqk} $\Q(\PromCycleEqk) \in \Theta(\Q(\PromCycleLeqk))$;
        \item \label{Item:DirCycleLeq} $\Q(\DirCycleEqk) \in O(\Q(\CycleEqk))$ whenever $k$ is an odd integer;
        %\item $\Q(\PromDirCycleEqk=O(\DirCycleEqk)$; \Subha{This one can be omitted from here and may be i can put it in a corollary.}
        \item \label{item:promCleqkNonprom} $\Q(\PromCycleLeqk) \in O(\Q(\CycleLeqk))$;
        \item \label{item:promDirCleqkNonprom} $\Q(\PromDirCycleLeqk) \in O(\Q(\DirCycleLeqk)))$;
        
        \item \label{item:PromCEqKPromDir} $\Q(\PromCycleEqk) \in O(\Q(\PromDirCycleEqk))$;
        \item \label{item:CEqKDir} $\Q(\CycleLeqk) \in O(\Q(\DirCycleLeqk))$;
        \item \label{item:ClEqKdir} $\Q(\CycleEqk) \in O(\Q(\DirCycleEqk))$.
        
    \end{enumerate}   
    \end{proposition}
    \begin{proof}
    We will prove all these statements one by one via reductions.
    \begin{enumerate}
        \item First, we prove $\Q(\DirCycleEqk) \in O(\Q(\DirCycleLeqk))$. We want to decide if a directed input graph $G=(V, E)$ has a directed cycle of length $=k$ using an algorithm $\mathcal{A}$ that decides whether a directed graph has a cycle of length $\leq k$. We construct a $k$-color cyclic construction of $G$ as sketched in \Cref{def:ColorCodingTechniqueForCycles}; let us denote the resultant graph (which is a directed graph) as $G'=(V',E')$. If $G$ had a directed cycle of length $k$ then $G'$ will also have a cycle of length $=k$ with probability at least $1/k^k \in \Omega(1)$, and if $G$ did not have a directed cycle of length $k$ then $G'$ will not have it either. Moreover, any cycle in $G'$ will be of length at least $k$. 
        Invoking \Cref{lem:1-sided_error} proves the desired direction.
        
        In the other direction, from \Cref{item:MonotonicityDirCycle} of \Cref{thm:MonotonicityCycles} we have a reduction that establishes $\Q(\mathtt{DirCycle}^{=\ell}) \in O(\Q(\DirCycleEqk))$ for any $\ell \in [k]\setminus \{1,2\}$. 
        %The $\ell=1$ case is not well defined so suppose that $\ell >1$. Let $G = (V,E)$ be the input of $\mathtt{DirCycle}^{= \ell}$, we construct a $(\ell,k-\ell)$-color cyclic construction corresponding to $G$ as sketched in \Cref{def:layerInsertionCycle}; let us denote the resultant directed graph by $G'' =(V'',E'')$. %Set $V'=V$. For each vertex $v \in V'$, assign a color from $[\ell-1]_0$ chosen uniformly at random. Then, choose color $i^* =\ell-1$ and make $k-\ell$ many copies of each vertex colored $i^*$. Additionally, for each $j \in [k-\ell]$, let $v_j$ denote the $j$th copy of the vertex $v$ with $c(v) = i^*$. For all vertices $v$ with $c(v) = i^*$ and $j \in [k-\ell-1]$, add edges $(v_j,v_{j+1})$ in $E'$ and also add the edge $(v,v_1)$ in $E'$. If $G$ has a directed cycle of length $=\ell$ then $G''$ will also have a directed cycle of length $=k$ with probability at least $1/\ell^{\ell}=\Omega(1)$. In any other case, $G''$ will not have any directed cycle of length $=k$. Invoking \Cref{lem:1-sided_error} proves that for any $\ell \in [k]\setminus \{1,2\}$ we have $\Q(\mathtt{DirCycle}^{=\ell})=O(\Q(\DirCycleEqk))$. 
        To decide if $G$ has a directed cycle of length $\leq k$, we loop over all $\ell \in [k]\setminus \{1,2\}$ and run the above procedure to decide if $G$ has a directed cycle of length $=\ell$, and output \textsc{yes} if it outputs \textsc{yes} for any $\ell \in [k]\setminus \{1,2\}$. Hence, proving that $\Q(\mathtt{DirCycle}^{\leq k}) \in O(\Q(\DirCycleEqk))$.
        
        \item The reductions to prove both these relations are similar to the arguments presented for \Cref{item:DirCycleEqLeqk}. What remains to be argued is that the arguments hold in the promise setting as well. To prove the relation $\Q(\PromDirCycleEqk)\in O(\Q(\PromDirCycleLeqk))$ we use the same idea of its non-promise counterparts. It is easy to see that when $G$ (the input for $\PromDirCycleEqk$) is a \textsc{no} instance then because of the promise there will be no cycle of any length in $G$ which means $G'$ (the input for $\PromDirCycleLeqk$) will also have no cycles as the construction doesn't create new cycles. When $G$ is a \textsc{yes} instance, $G'$ is also a \textsc{yes} instance with probability $\Omega(1)$. Even though the promise of $\PromDirCycleLeqk$ may not be satisfied for the rest of the times, invoking \Cref{lem:1-sided_error} concludes the argument.

        To show the other relation, i.e., $\Q(\PromDirCycleLeqk) \in O(\Q(\PromDirCycleEqk))$, we invoke \Cref{item:PromMonotonicityDirCycle} of \Cref{thm:MonotonicityCycles} and then use the same reduction used for its non-promise counterparts as stated in \Cref{item:DirCycleEqLeqk}.
        
        \item We will show that $\Q(\CycleLeqk) \in O(\Q(\CycleEqk))$. Let $G$ be the input to $\CycleLeqk$. The idea is to loop over $\ell \in [3,k]$ and construct an $(\ell,k-\ell)$-cyclic construction of $G$ (as sketched in \Cref{def:layerInsertionCycle}) for every value of $\ell \in [3,k]$; let the constructed instance be denoted by $G''(\ell)$. The idea is to now run an algorithm for $\CycleEqk$, let us denote by $\mathcal{A}$,  on every $G''(\ell)$ and output \textsc{yes} if $\mathcal{A}$ outputs \textsc{yes} on $G''(\ell)$ for any value of $\ell \in [3,k]$, otherwise output \textsc{no}. This reduction works because of the following reasons. Suppose that $G$ did not contain any cycle of length $\leq k$, i.e., $G$ is a \textsc{no} instance of $\CycleLeqk$. Then either $G$ contains cycles of length $>k$ or no cycle at all. In either of the cases none of the $G''(\ell)$s will not contain a cycle of length $=k$ because the construction enforces that no new cycles are created and the length of the cycles (if present) in $G$ doesn't decrease in length. This means a \textsc{no} instance of $\CycleLeqk$ maps to a \textsc{no} instances of $\CycleEqk$. On the other hand, if $G$ was a \textsc{yes} instance of $\CycleLeqk$ then $G''(\ell)$ will be a \textsc{yes} instance of $\CycleEqk$ with probability $>1/\ell^{\ell} \in \Omega(1)$ for at least one $\ell \in [3,k]$. Finally, invoking \Cref{lem:1-sided_error} concludes the proof.

        (It is beneficial to note that the reduction in the other direction doesn't hold; at least we don't know how to prove it yet. Mainly, because we don't know how to get rid of smaller cycles.)
        
        \item The relation $\Q(\PromCycleLeqk) \in O(\Q(\PromCycleEqk))$ follows from the definitions of the respective promise problems, along with the reduction we use for \Cref{item:CycleEqLeqk} and by invoking \Cref{lem:1-sided_error}.
        
        For the other direction, $\Q(\PromCycleEqk) \in O(\Q(\PromCycleLeqk))$, we run the $\PromCycleLeqk$ algorithm on the $k$-color cyclic construction of $G=(V,E)$, let us denote by $G'=(V',E')$ where $G$ is the input to the $\PromCycleEqk$ problem. The promise on our input $G$ ensures that there won't be any cycle in the \textsc{no} case so running algorithm for $\PromCycleLeqk$ on $G'$ will give \textsc{no}. And, in the \textsc{yes} case the construction $G'$ will have a cycle of length $k$ with probability $> 1/k^k \in \Omega(1)$. Invoking \Cref{lem:1-sided_error} gives us the desired result.
        \item We will now argue $\Q(\DirCycleEqk) \in O(\Q(\CycleEqk)$ holds for odd values of $k$. Let $G=(V,E)$ be an input to $\DirCycleEqk$ problem. We first construct a $k$-color cyclic construction of $G$, let us denote it by $G'=(V',E')$. Note that $G'$ will also be a directed graph and will contain a directed cycle of length $=k$ with probability at least $1/k^k \in \Omega(1)$ if $G$ contains one. Moreover, if at all $G'$ contains directed cycles then it will be of length at least $k$ and furthermore ``cycles'' with inconsistent edge directions won't exist in $G'$ because of the color-coding. We now construct an undirected simple graph from $G'$, let us denote by $G''=(V'',E'')$. We do so by forgetting the directions of all the edges in $E'$. We then run an algorithm for $\CycleEqk$ on $G''$. If $k$ is even, then forgetting the edge directions in $G'$ could induce undirected cycles of length $=k$ that were not present in the directed $G'$; for example - a cycle could be just a multipartite matching between a few adjacent colors.  But this is not a problem in the case when $k$ is odd because any multipartite matching would only induce cycles of even length. Invoking \Cref{lem:1-sided_error} gives us the desired relation.
    \end{enumerate}
    
    For \Cref{item:promCleqkNonprom,item:promDirCleqkNonprom}, the relations follow from a straightforward application of \Cref{lem:PromFindForCycle}. And, for \Cref{item:CEqKDir,item:ClEqKdir,item:PromCEqKPromDir}, the relations follow from \Cref{lem:undirected_to_directed}.
    \end{proof}

    \begin{theorem}[Implicit in {\cite[Theorem 8]{belovs2012span}}]
    \label{Thm:BR-TriangleVersusForest}
    Let $3\leq k \in \N$ be a constant. Then, there exists a $O(n)$-query quantum algorithm that, given a simple graph $G$ with $n$ vertices, accepts if $G$ contains a cycle of length $k$ and rejects if $G$ is a forest, except with error probability at most $1/3$.  
    \end{theorem}
    
    Combining \Cref{Item:PromCycleEqLeqk} of \Cref{Cycle-Equivalences-Part1} and \Cref{Thm:BR-TriangleVersusForest} we conclude the following.

    \begin{corollary}
    Let $3 \leq k \in \mathbb{N}$ be a constant. Then, $\Q(\mathtt{PromCycle}^{=k|\leq k}) \in \Theta(n)$. 
    \end{corollary}

    We will now present a few other cycle-containment reductions where the cycles are required to pass through a specific vertex, let us denote by $s$, given as input. For these reductions we will use color-coding and layer insertion techniques as done earlier, but we will use the color-coding version that passes through $s$. It is interesting to note that for some of these problems we are able to prove stronger relation in contrast to its counterparts that don't have this additional requirement of the subgraph containing $s$. For example, \Cref{item:CycleEqLeqksPartB,item:dirCycleEqkCycleEqks} in \Cref{Cycle-Equivalences-Part2} are to name a few. See \Cref{fig:cycle-problems} for an illustrative and comparative summary of the results.

    \begin{proposition}\label{Cycle-Equivalences-Part2}
    Let $3 \leq k \in \N$ be a constant. Then,
    \begin{enumerate} 
        \item \label{item:DirCycleEqLeqks} $\Q(\DirCycleEqks) \in \Theta(\Q(\DirCycleLeqks))$;
        \item \label{item:PromDirCycleEqLeqks} $\Q(\PromDirCycleLeqks) \in \Theta(\Q(\PromDirCycleEqks))$;
        \item \label{item:CycleEqLeqks} 
        \begin{enumerate}
            \item \label{item:CycleEqLeqksPartA} $\Q(\CycleLeqks) \in O(\Q(\CycleEqks))$;
            \item \label{item:CycleEqLeqksPartB} $\Q(\CycleEqks) \in O(\Q(\CycleLeqks))$;
        \end{enumerate} 
        \item \label{Item:PromCycleEqLeqks} $\Q(\PromCycleEqks) \in \Theta(\Q(\PromCycleLeqks))$;
        \item \label{item:dirCycleEqkCycleEqks} $\Q(\DirCycleEqks) \in O(\Q(\CycleEqks))$;
        \item \label{item:PromDirCycleEqksDirCycleEqks} $\Q(\PromDirCycleEqks) \in  \Theta(\Q(\DirCycleEqks))$; 
        \item \label{item:promCleqkNonproms} $\Q(\PromCycleLeqks) \in O(\Q(\CycleLeqks))$;
        \item \label{item:promDirCleqkNonproms} $\Q(\PromDirCycleLeqks) \in O(\Q(\DirCycleLeqks)))$;
    
        \item \label{item:PromCEqKPromDirs} $\Q(\PromCycleEqks) \in O(\Q(\PromDirCycleEqks))$;
        \item \label{item:CEqKDirs} $\Q(\CycleLeqks) \in O(\Q(\DirCycleLeqks))$;
        \item \label{item:ClEqKdirs} $\Q(\CycleEqks) \in O(\Q(\DirCycleEqks))$.  
    \end{enumerate}   
    \end{proposition}
    \begin{proof}
    The reductions for \Cref{item:DirCycleEqLeqks,item:PromDirCycleEqLeqks,item:CycleEqLeqksPartA,Item:PromCycleEqLeqks} follow (almost) exactly via the reductions used for their problem counterparts in \Cref{Cycle-Equivalences-Part1} where the subgraph doesn't have to go through vertex $s$. The only change needed here is that we use the color coding construction that goes through $s$ as described in \Cref{def:ColorCodingTechniqueForCycles}; this difference changes the exact probability of success for \textsc{yes} instances but is still $\Omega(1)$ as required. 

    For the item in \Cref{item:CycleEqLeqksPartB} we want to argue that $\Q(\CycleEqks) \in O(\Q(\CycleLeqks))$. Let $G=(V,E)$ and $s \in V$ be the input to $\CycleEqks$. The main idea for this reduction is to try and get rid of all the small (i.e., length $<k$) cycles in $G$. Towards establishing that, we first learn the neighborhood of $s$ with $n$ queries to $G$; let us denote its neighborhood by $N(s)$. As $\Q(\CycleEqks) \in \Omega(n)$ from \Cref{lem:subgraph-characterization} we can now treat our input $G,s$ as $G,s$ equipped with query-free access to $N(s)$. We will now construct an instance $G'=(V',E')$ from $G=(V,E)$ in the following way. For every $w \in N(s)$, we assign $w$ to one of the three groups, let us denote by $N^{\texttt{red}}(s),N^{\texttt{yellow}}(s),N^{\texttt{blue}}(s)$, uniformly at random. Consequently, $N(s)=N^{\texttt{red}}(s) \sqcup N^{\texttt{yellow}}(s) \sqcup N^{\texttt{blue}}(s)$. Set $V'=V\setminus (N^{\texttt{red}}(s) \sqcup N^{\texttt{blue}}(s))$. For every $\{u,v\} \in E$ add it to $E'$ only when $u,v \notin N(s)\cup\{s\}$. We now add vertices $r,b$ in $V'$; here $r,b$ corresponds to a vertex after ``contracting'' the vertices in $N^{\texttt{red}}(s),N^{\texttt{blue}}(s)$, respectively. Add $\{s,r\}, \{s,b\}$ to $E'$ and additionally for every $\{u,v\} \in E$ such that $u \in N^{\texttt{red}}(s)$ and $v \notin N^{\texttt{red}}(s) \sqcup N^{\texttt{blue}}(s) \sqcup s$ add $\{v,r\}$ to $E'$. Similarly, for every $\{u,v\} \in E$ such that $u \in N^{\texttt{blue}}(s)$ and $v \notin N^{\texttt{red}}(s) \sqcup N^{\texttt{blue}}(s) \sqcup s$ add $\{v,b\}$ to $E'$.  We now construct $G''=(V'',E'')$ by color-coding $G'=(V',E')$ by assigning colors $1,2,k$ to $s,r,b$, respectively and color coding the rest of the vertices into the rest of the $k-3$ colors as done in other reductions. We can now run an algorithm for $\CycleLeqks$ on $G'',s$. First note that any cycle in $G''$ going through $s$ (if present) has to pass through both $r,b$ and therefore must is of length $\geq k$. If there is a cycle of length $=k$ in $G$ (\textsc{yes} instance) then with probability $\Omega(1)$ there will be a cycle of length $=k$ through $s$ in $G''$ and the algorithm will output \textsc{yes}. If there was no cycle through $s$ in $G$ then $G''$ also has no cycle through $s$ because the construction doesn't create new cycles. Moreover, any smaller cycles in $G$ passing through $s$ don't survive in $G''$. And, cycles of length $>k$ in $G$ map to \textsc{no} instance of $G'',s$ either by not surviving or by being too long. Therefore, the reduction maps the \textsc{no} instance $G,s$ to a \textsc{no} instance $G'',s$. Invoking \Cref{lem:1-sided_error} concludes the argument.

    For the relation in \Cref{item:dirCycleEqkCycleEqks}, i.e., towards proving $\Q(\DirCycleEqks) \in O(\Q(\CycleEqks))$, we use a similar trick as done in \Cref{item:CycleEqLeqksPartB}. Let $G=(V,E)$ be a directed graph that is input to $\DirCycleEqks$. We first learn the neighborhood of $s$ using $n$ queries and that is fine because $\Q(\DirCycleEqks) \in \Omega(n)$; see \Cref{lem:subgraph-characterization}. Now we construct $G''=(V'',E'')$ exactly as we did in the previous argument except that the constructed graph is directed. We now forget the directions of the edges in the constructed graph and then running an algorithm for $\CycleEqks$ on this graph suffices. If $G$ was a \textsc{no} instance then the final constructed graph undirected-$G''$ is also a \textsc{no} instance because undirected-$G''$ will have cycles of length $>k$. And in the case when $G$ is a \textsc{yes} instance then the final constructed graph undirected-$G''$ is also a \textsc{yes} instance with probability $\Omega(1)$. Invoking \Cref{lem:1-sided_error} concludes the proof.
   
    For the relation in \Cref{item:PromDirCycleEqksDirCycleEqks}, one direction is straightforward application of \Cref{lem:PromFindForCycle}. We now argue the other direction, i.e., $\Q(\DirCycleEqks) \in O(\Q(\PromDirCycleEqks))$. Let the directed graph $G=(V,E)$ and vertex $s$ be the input to $\DirCycleEqks$ problem. We first construct a $k$-color cyclic construction of $G$ through $s$ using \Cref{def:ColorCodingTechniqueForCycles}. Let us denote this new  directed graph by $G'=(V',E')$. Note that, the construction ensures if $G'$ has any cycle then it will only have cycles of length $=k$ and not any other integral multiples of $k$ because there is only one single vertex of color $0$ in $G'$ (which is $s$) and to have any bigger cycle in $G'$ one needs to have another vertex $\neq s$ of color $0$, which is not possible because of the construction. Which means the graph $G'$ either has a directed cycle of length $=k$ through $s$, which happens with probability $\Omega(1)$ when the original graph $G$ has a length $=k$ cycle, or $G'$ has no directed cycle passing through $s$ when $G$ is a \textsc{no} instance. These requirements perfectly align with the promise of $\PromDirCycleEqks$. Therefore, it suffices to run an algorithm for $\PromDirCycleEqks$ on $G'$ to solve $\DirCycleEqks$ on $G$. And, invoking \Cref{lem:1-sided_error} gives us the desired relation.
    
    For \Cref{item:promCleqkNonproms,item:promDirCleqkNonproms}, the relations follow from a straightforward application of \Cref{lem:PromFindForCycle}. And, for \Cref{item:PromCEqKPromDirs,item:CEqKDirs,item:ClEqKdirs}, the relations follow from \cref{lem:undirected_to_directed}. 
    \end{proof}
    
    Consequently, we get the following.
    \begin{corollary}
    \label{cor:CycleSIsland}
    Let $3 \leq k \in \N$. Then, all the problems in $\mathtt{(|Prom)DirCycle}_{s}^{(=k|\leq k)}$ have the same, up to constants, quantum query complexity.
    \end{corollary}

    Finally, we prove a connection between the restricted and unrestricted version of the cycle-containment problems.

    \begin{proposition}\label{thm:CrossArrow-DirCycleEqksToDirCycleEqk}
        Let $3 \leq k \in \N$ be a constant. Then, we have $\Q(\mathtt{PromDirCycle}_s^{=(k+1)}) \in % O(\Q(\PromDirCycleEqk))$.
        O(\Q(\PromDirCycleEqk))$. %\Subha{Amin and I think we can make this work without the promise as well because if T is empty then there was no cycle through s to begin with.}
    \end{proposition}
    
    \begin{proof}
        Let the directed graph $G=(V,E)$ and the vertex $s$ be the input to problem $\mathtt{PromDirCycle}_s^{=(k+1)}$. We spend $n$ queries learning the neighborhood of $s$, and let $S$ and $T$ be the sets of vertices that are connected to $s$ by arcs coming out from and going in to $s$, respectively. Next, we assign colors $1$ and $k$ to the vertices from $S$ and $T$, respectively, and we assign colors $\{2, \dots, k-1\}$ uniformly at random to all the other vertices in $V \setminus (S \cup T \cup \{s\})$. We let $G' = (V',E')$ be a graph with all vertices except $s$ (i.e., $V' = V \setminus \{s\}$), where we only keep the arcs of the old graph $G$ if in the new graph $G'$ whenever the arcs are outgoing from color $i$ to $i+1$ for $i \in \{1, \ldots, k-1\}$, and we put an outgoing arc from every vertex in $T$ to every vertex in $S$. If $G$ is a \textsc{NO} instance to the $\mathtt{PromDirCycle}_s^{=(k+1)}$ problem, then $G'$ remains a \textsc{NO} instance to the $\PromDirCycleEqk$ problem. On the other hand, if $G$ contains a cycle of length $k+1$ going through $s$, then $G'$ now contains a cycle of length $k$ with probability at least $1/(k-2)^{k-1} \in \Omega(1)$. Thus, we employ \Cref{lem:1-sided_error} to solve the $\mathtt{PromDirCycle}_s^{=(k+1)}$ problem using an algorithm for the $\PromDirCycleEqk$ problem. We conclude the proof by observing from \Cref{lem:subgraph-characterization} that the query complexity of $\Q(\PromDirCycleEqk) \in \Omega(n)$ for $k \geq 3$.\qedhere

        % We construct a $k$-color cyclic construction of $G$ through $s$ as described in \Cref{def:ColorCodingTechniqueForCycles}; let us denote the resultant graph as $G'=(V',E')$. The construction ensures that the only directed cycles of length $=k$ that go through $s$ has a chance to survive in $G'$. Hence, running an algorithm $\PromDirCycleEqk$ on $G'$ will detect such a cycle (if it existed in $G$) with probability $\Omega(1)$ and if there were no cycles of length $=k$ in $G$ through $s$ then $G'$ will also not have any cycle of length $k$. Invoking \Cref{lem:1-sided_error} concludes the proof. Thus, $\Q(\DirCycleEqks) \in O(\Q(\DirCycleEqk))$.
        
        % Observe that, the reduction described earlier can also be used to strengthen the result in their respective promise versions. Consequently, implying $\Q(\PromDirCycleEqks) \in O(\Q(\PromDirCycleEqk))$.
    \end{proof}

    % \begin{proposition}\label{thm:CrossArrow-DirCycleEqksToDirCycleEqk} Let $3 \leq k \in \N$ be a constant. Then, $\Q(\DirCycleEqks) \in O(\Q(\DirCycleEqk))$ and $\Q(\PromDirCycleEqks) \in O(\Q(\PromDirCycleEqk))$. \todo{Improve to $k-1$!}
    % \end{proposition}
    % \begin{proof}
    % Let the graph $G=(V,E)$ and vertex $s$ be the input to the $\DirCycleEqks$ problem. We construct a $k$-color cyclic construction of $G$ through $s$ as described in \Cref{def:ColorCodingTechniqueForCycles}; let us denote the resultant graph as $G'=(V',E')$. The construction ensures that the only directed cycles of length $=k$ that go through $s$ has a chance to survive in $G'$. Hence, running an algorithm $\DirCycleEqk$ on $G'$ will detect such a cycle (if it existed in $G$) with probability $\Omega(1)$ and if there were no cycles of length $=k$ in $G$ through $s$ then $G'$ will also not have any cycle of length $k$. Invoking \Cref{lem:1-sided_error} concludes the proof. Thus, $\Q(\DirCycleEqks) \in O(\Q(\DirCycleEqk))$.

    % Observe that, the reduction described earlier can also be used to strengthen the result in their respective promise versions. Consequently, implying $\Q(\PromDirCycleEqks) \in O(\Q(\PromDirCycleEqk))$.
    % \end{proof}

    \subsection{Path-cycle-containment reductions}
    \label{sec:PathCycleContainmentReductions}

    In this subsection, we present reductions between path and cycle problems (in both directions) to further refine the classification established in \Cref{sec:PathContainmentReductions,sec:CycleContainmentReductions} for path and cycle problems, respectively.

    \begin{proposition}\label{thm:CrossArrow-DirCycleEqksToDirPathstEqk} Let $3 \leq k \in \N$ be a constant. Then, $\Q(\DirCycleEqks) \in O(\Q(\mathtt{DirPath}_{s,t}^{=k}))$.
    \end{proposition}
    \begin{proof}
    Let the graph $G=(V,E)$ and vertex $s$ be the input to $\DirCycleEqks$. Construct a new directed $G'=(V',E')$ in the following way:
    \begin{enumerate} 
        \item set $V'=V \sqcup \{t\}$ where $t$ denotes a new vertex that is not already present in $V$;
        \item for every $(u,v) \in E$ with $v\neq s$, keep $(u,v)$ in $E'$ and for all $(u,s) \in E$ create edge $(u,t)$ in $E'$. Basically, we are directing all the incoming edges of $s$ to $t$ in our new graph $G'$.
    \end{enumerate}
    Observe that, every length-$\ell$ directed cycle in $G$ passing through $s$ maps to a length-$\ell$ $s \rightarrow t$ path in $G'$. Moreover, no other $s \rightarrow t$ paths are created in $G'$ in this process. Therefore, with this construction, we can solve $\DirCycleEqks$ on $G,s$ by solving $\mathtt{DirPath}_{s,t}^{=k}$ on graph $G',s,t$.
    \end{proof}

    \begin{proposition}\label{thm:crossArrow-promDirPathEqkstToPromDirCycleEqks}
    Let $3 \leq k \in \N$ be a constant. Then, $\Q(\mathtt{PromDirPath}_{s,t}^{=k}) \in O(\Q(\PromDirCycleEqks))$.
    \end{proposition}
    \begin{proof}
    Let the graph $G=(V,E)$ along with the vertices $s,t$ be the input to $\mathtt{PromDirPath}_{s,t}^{=k}$. We will construct a directed graph $G'=(V',E')$ by first setting $V'=V$. Color $s$ with $0$ and color $t$ with $k$. Additionally, for all $v \in V'\setminus \{s,t\}$ assign a color from $[k-1]$ chosen uniformly at random. And, for each edge $(u,v) \in E$, add $(u,v)$ to $E'$ if $c(v)=c(u)+1$. We now construct $G''=(V'',E'')$ from $G'$ by merging the vertices $s,t$ into say $s''$ - basically means, for all edges $(s,u) \in E'$ add edge $(s'',u)$ to $E''$ and for all $(u,t) \in E'$ add edge $(u,s'')$ to $E''$. We will run the algorithm for $\PromDirCycleEqks$ on graph $G''$ and vertex $s''$.

    Let us now see why the reduction holds. The \textsc{yes} instance of $\mathtt{PromDirPath}_{s,t}^{=k}$ constructs a \textsc{yes} instance $G'$ to $\mathtt{PromDirPath}_{s,t}^{=k}$ with probability $1/(k-1)^{k-1} \in \Omega(1)$. And, a \textsc{yes} instance $G'$ maps to a \textsc{yes} instance $G''$ of $\PromDirCycleEqks$ with probability $1$. Moreover, the construction ensures that there are no directed cycles that go through only $s$ or only $t$ in graph $G'$ because $s$ only contains outgoing edges and $t$ only contains incoming edges in $G'$ - this means if algorithm for $\PromDirCycleEqks$ on $G'',s''$ outputs \textsc{yes} then there is corresponding $s \rightarrow t$ path of length $k$ in $G'$ (and hence in $G$). Additionally, the promise of $\mathtt{PromDirPath}^{=k}_{s,t}$ ensures that if there is no $s \rightarrow t$ path in $G$ then there is no path at all from $s \rightarrow t$ in $G$. This promise is satisfied even in the color coded version $G'$ and satisfies the promise for $\PromDirCycleEqks$ as well. Therefore, invoking \Cref{lem:1-sided_error} gives us the desired relation.
    \end{proof}

    Using \Cref{cor:path-island} from \Cref{sec:PathContainmentReductions}, \Cref{cor:CycleSIsland} from \Cref{sec:CycleContainmentReductions} and \Cref{thm:CrossArrow-DirCycleEqksToDirPathstEqk,thm:crossArrow-promDirPathEqkstToPromDirCycleEqks} we conclude the following.
    \begin{corollary}\label{ConnectingPathIslandsToCycleSisland}
    Let $3 \leq k \in \N$ be a constant. Then, the path problems $\mathtt{(|Prom)DirPath}_{s,t}^{(=k|\leq k)}$ and the cycle problems $\mathtt{(|Prom)DirCycle}_{s}^{(=k|\leq k)}$ have the same, up to constants, quantum query complexity. 
    \end{corollary}

    \begin{proposition}\label{thm:crossArrow-promCycleEqksToPromPathEqkst}
    Let $3 \leq k \in \N$ be a constant. Then, $\Q(\PromCycleEqks) \in O(\Q(\mathtt{PromPath}_{s,t}^{=k}))$. 
    \end{proposition}
    \begin{proof}
    To prove $\Q(\PromCycleEqks) \in O(\Q(\mathtt{PromPath}_{s,t}^{=k}))$, we use the following argument. Let the graph $G=(V,E)$ and vertex $s$ be the input of $\PromCycleEqks$. We construct graph $G'=(V',E')$ from $G=(V,E)$ in the following way:
    \begin{enumerate} 
        \item set $V'=V \sqcup \{t\}$ where $t$ is a new vertex that was not present in $V$;
        \item for every $\{u,v\} \in E$ such that $u\neq s$ and $v \neq s$, add the edge $\{u,v\}$ in $E'$, and, for every $\{u,s\} \in E$, choose uniformly at random to either add $\{u,s\}$ to $E'$ or (exclusive or) to add $\{u,t\}$ to $E'$.
    \end{enumerate}
    We will run the algorithm for $\mathtt{PromPath}_{s,t}^{=k}$ on graph $G'=(V',E')$ and vertices $s,t$. Observe that, if there was a length-$k$ cycle through $s$ in graph $G$ then with probability $1/2 \in \Omega(1)$ there will be length $k$ path between $s,t$ in $G'$. In the case where the algorithm for $\mathtt{PromPath}^{=k}_{s,t}$ outputs \textsc{yes} then there are two distinct neighbors of $s$ in the original graph $G$ with a path of length $k-2$ which means there is a cycle through $s$ in $G$. Invoking \Cref{lem:1-sided_error} concludes the reduction.
    \end{proof}

    \begin{proposition}
    \label{thm:crossArrow-PromPathstToPromCycles}
    Let $3 \leq k \in \N$ be a constant. Then, $\Q(\mathtt{PromPath}_{s,t}^{=k}) \in O(\Q(\mathtt{PromCycle}_{s}^{=k}))$. 
    %\Arjan{Perhaps we can change this back to the version where $k$ remains the same? I think we will need this technique of learning the neighborhood of $s$ for the connection between $\mathtt{Cycle}^{=k}$ and $\mathtt{DirCycle}^{=k}$ as well anyway.}\Subha{yes, it works.}
    \end{proposition}
    \begin{proof}
    %Let graph $G=(V,E)$ and vertices $s,t$ be the input to $\mathtt{PromPath}_{s,t}^{=k}$. We construct a new graph $G'=(V',E')$ by setting $V'=V \sqcup \{v\}$ where $v$ is a vertex that was not originally present in $V$. For all edges $\{u,v\} \in E$ we add them in $E'$ and additionally we add edges $\{v,s\},\{v,t\}$ to $E'$; basically $v$ has only two neighbors, $s$ and $t$. We run the algorithm for $\mathtt{PromCycle}_{s}^{=(k+2)}$ on graph $G'$ and vertex $v$. Observe that if $G,s,t$ was a \textsc{yes} instance then there is length-$k$ path between $s,t$ and as essentially we insert $v$ between $s,t$ in our construction, it ensures that there is a cycle of length-$k+2$ passing through $v$ in $G'$. Hence, a \textsc{yes} instance of $G$ maps to a \textsc{yes} instance of $G'$. Moreover, given that $v$ has only two neighbors $s,t$, if $G',v$ is a \textsc{yes} instance then the cycle of length $k+2$ through $v$ is possible only with a length-$k$ path between $s,t$. This concludes our proof.  

    Let graph $G=(V,E)$ and vertices $s,t$ be the input to $\mathtt{PromPath}_{s,t}^{=k}$. Using $2n$ many queries to $G$, we learn the neighborhood of $s,t$ which we denote by $N(s), N(t)$, respectively. It is fine to do so because $\mathtt{PromPath}_{s,t}^{=k} \in \Omega(n)$. We will now construct an instance $G'=(V',E')$ from $G=(V,E)$ in the following way. For every $u \in N(s)$, we assign $u$ to one of the two groups, let us denote by $N^{\textit{red}},N^{\textit{yellow}}$, uniformly at random. Similarly, for every $v \in N(t)\setminus N(s)$, we assign $v$ to either $N^{\textit{yellow}}$  or to a new group, let us denote by $N^{\textit{blue}}$, uniformly at random. Similarly to the idea in \Cref{item:CycleEqLeqksPartB} of \Cref{Cycle-Equivalences-Part2} we construct a new graph $G'=(V',E')$ by contracting all the vertices in $N^{\textit{red}},N^{\textit{blue}}$ and denoting them as $r,b$, respectively and preserving all the edges from $E$ to $E'$ that are incident on $N^{\textit{red}},N^{\textit{blue}}$ (except the ones to $s,t$) as edges to/from $r,b$, respectively. We also add all the edges from $E$ in $E'$ that don't involve any of the vertices in $\{s,t\}\cup N^{\textit{red}} \cup N^{\textit{blue}}$ as either of its end points. Additionally, we add the edges $\{r,s'\}, \{b,s'\}$ to $E'$. To solve $\mathtt{PromPath}_{s,t}^{=k}$ on $G$ we now run an algorithm for $\PromCycleEqks$ on $G',s'$. If $G$ has a length $k$ path between $s,t$ then with probability $\Omega(1)$ there will be a length $k$ cycle through $s'$ in graph $G'$. If $G$ is a \textsc{no} instance then because of the promise there is no path between $s,t$ in $G$ and the construction ensures that there is no cycle (of any length) going through $s'$ in graph $G'$. Invoking \Cref{lem:1-sided_error} concludes the proof. \qedhere
    \end{proof}

    As $\mathtt{PromPath}_{s,t}^{=k}$ is known to have quantum query complexity of $\Theta(n)$ (see \Cref{fig:path-problems}), using \Cref{Item:PromCycleEqLeqks} of \Cref{Cycle-Equivalences-Part2} and \Cref{thm:crossArrow-promCycleEqksToPromPathEqkst,thm:crossArrow-PromPathstToPromCycles} we can conclude the following:

    \begin{corollary}\label{thm:ComplexityPromCycles}
    Let $3 \leq k \in \N$ be a constant. Then, $\Q(\mathtt{PromCycle}_{s}^{(=k|\leq k)}) \in \Theta(n)$.   
    %\Subha{I can perhaps add more problems in this statement.}
    \end{corollary}

    %\Arjan{We can still add monotonicity for $\mathtt{PromDirCycle}^{=k}$, $\mathtt{DirCycle}^{=k}$, and a result that for odd values of $k$, $\mathsf{Q}(\mathtt{Cycle}^{\leq(k-1)}) \in O(\mathsf{Q}(\mathtt{Cycle}^{\leq k}))$.}\Subha{Added it in the earlier subsection also added monotonicity for the going through s versions.}

    \section{Length-$k$ directed path-detection algorithm}
    \label{sec:dir-path-algo}

    In this section, we consider the subgraph containment problem $\mathtt{DirPath}^{=k}$, and we develop a novel algorithm for this problem. This implies a query upper bound for all of the problems displayed in red in \Cref{fig:islands}.

    We start by employing the color-coding technique, introduced by Alon, Yuster and Zwick~\cite{alon1995color}, to partition the graph into $k+1$ subsets.

    \begin{lemma}
        \label{lem:dir-path-color-coding}
        Let $k,n \in \N$, and let $G = (V,A)$ be a directed graph with $|V| = n$ and $k \in O(1)$. Let $V_0 \sqcup \cdots \sqcup V_k \subseteq V$ be a random partition of $V$ such that $|V_j| \in \Theta(n/k)$. We define $G' = (V,A')$, where
        \[A' = A_1 \sqcup \cdots \sqcup A_k, \qquad \text{with} \qquad A_j = \{(v,w) \in A : v \in V_{j-1}, w \in V_j\}.\]
        Now, if $\mathtt{DirPath}^{=k}(G) = 0$, then $\mathtt{DirPath}^{=k}(G') = 0$. On the other hand, if $\mathtt{DirPath}^{=k}(G) = 1$, then
        \[\underset{G'}{\P}\left[\mathtt{DirPath}^{=k}(G') = 1\right] \in \Omega(1),\]
        where the probability distribution is over the randomization in the procedure that constructs $G'$.
    \end{lemma}

    \begin{proof}
        For the negative case, we observe that the new graph $G'$ only contains a subset of the arcs in $G$. Therefore, if the original graph $G$ does not contain a directed path of length $k$, then neither does $G'$.

        For the positive case, we observe that every vertex obtains any given color with probability $1/(k+1)$. Thus, the probability that every vertex in the directed path receives the correct color is $(k+1)^{-(k+1)} \in \Omega(1)$.
    \end{proof}

    We can combine the previous theorem with \Cref{lem:1-sided_error}, so without loss of generality we can consider the color-coded version of the problem, i.e., where the graph's vertices are already partitioned into $k+1$ disjoint sets and where we're looking for a path that traverses all the way from one end to the other, as displayed in \Cref{fig:layered-path}. To optimize the parameters in the recursion, we consider a version this problem where the first and last layer have a smaller size than all the intermediate layers, and setting $r \in \Theta(n)$ recovers the problem considered in \Cref{lem:dir-path-color-coding}.

    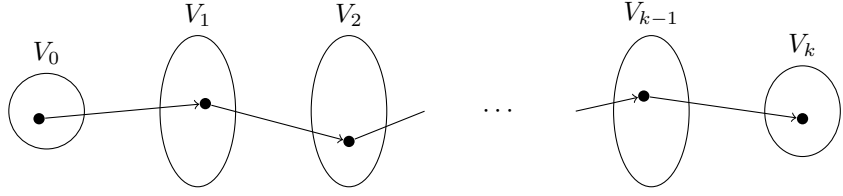
\begin{figure}[!ht]
        \centering
        \begin{tikzpicture}
            \draw (0,0) ellipse (.5 and .5);
            \draw (2,0) ellipse (.5 and 1);
            \draw (4,0) ellipse (.5 and 1);
            \node at (6,0) {$\cdots$};
            \draw (8,0) ellipse (.5 and 1);
            \draw (10,0) ellipse (.5 and .6);
            \node[above] at (0,.5) {$V_0$};
            \node[above] at (2,1) {$V_1$};
            \node[above] at (4,1) {$V_2$};
            \node[above] at (8,1) {$V_{k-1}$};
            \node[above] at (10,.6) {$V_k$};

            \node[fill,circle,minimum width=1.5,inner sep=1.5] (0) at (-.1,-.1) {};
            \node[fill,circle,minimum width=1.5,inner sep=1.5] (1) at (2.1,.1) {};
            \node[fill,circle,minimum width=1.5,inner sep=1.5] (2) at (4,-.4) {};
            \node[fill,circle,minimum width=1.5,inner sep=1.5] (4) at (7.9,.2) {};
            \node[fill,circle,minimum width=1.5,inner sep=1.5] (5) at (10,-.1) {};
            \draw[->] (0) to (1);
            \draw[->] (1) to (2);
            \draw (2) to (5,0);
            \draw[->] (7,0) to (4);
            \draw[->] (4) to (5);
        \end{tikzpicture}
        \caption{The layered-path problem. Every set $V_j$ for $j \in [k-1]$ represents a set of $\Theta(n)$ vertices, and the first and last layer might have a different size, say $|V_0| = |V_k| = r$. We are looking for a directed path that traverses all the layers from $V_0$ to $V_k$.}
        \label{fig:layered-path}
    \end{figure}

    % We define a slight generalization of this problem, which we will refer to as the ``layered path-finding problem''.

    % \begin{definition}[The layered path-finding problem]
    %     Let $k,n,r \in \N$, with $k \in O(1)$ and $r \leq n$. Let $V = V_0 \sqcup \cdots \sqcup V_k$, such that $|V_1|, \dots, |V_{k-1}| \in \Theta(n)$, and $|V_0|,|V_k| \in \Theta(r)$. Let $G = (V,A)$ be a directed graph such that
    %     \[A = A_1 \sqcup \cdots \sqcup A_k, \qquad \text{where} \qquad A_j \in V_{j-1} \times V_j.\]
    %     Then, we say that $G$ is an instance of the layered path-finding problem, with length $k$ and first and last layer size $r$.
    % \end{definition}

    We now provide a recursive nested walk construction that finds such a layered path, if it exists.

    \begin{lemma}
        \label{lem:dir-path-recurrence}
        Let $k,n,r \in \N$, with $3 \leq k \in O(1)$ and $r \leq n$. Let $n/2 \geq s \in \N$, and suppose that we can solve the layered path-finding problem with length $k-2$ and first and last layer size $s$ with $C_{n,s}^{(k-2)}$ queries. Then, we can solve the layered path-finding problem with complexity $C_{n,r}^{(k)}$, satisfying
        \[C^{(k)}_{n,r} \in \widetilde{O}\left(s\sqrt{r} + \left(\frac{n}{s}\right)\left(\sqrt{sr} + C^{(k-2)}_{n,s}\right)\right).\]
    \end{lemma}

    \begin{proof}
        We use a quantum walk. To that end, we first construct the graph that we will be walking over.

        Let $J_1$ be the Johnson graph of all subsets of size $s$ from $V_1$. Similarly, let $J_{k-1}$ be the Johnson graph of all subsets of size $s$ from $V_{k-1}$. Their spectral gaps are exactly equal to $n/(s(n-s))$~\cite[Section~12.3.2]{brouwer2011spectra}. We will be walking over the Cartesian product of these graphs, i.e., $G = J_1 \Box J_{k-1}$, and so we compute the spectral gap of $G$ as
        \[\delta(G) = \min\{\delta(J_1), \delta(J_{k-1})\} = \frac{n}{s(n-s)} \in \Theta\left(\frac{1}{s}\right).\]

        Now, we describe the marked set. We say that a vertex of $G$, denoted by $(W_1, W_{k-1})$ where $W_1 \subseteq V_1$ and $W_{k-1} \subseteq V_{k-1}$ is marked whenever there exits a directed path of length $k$ from $V_0$ to $V_k$ that passes through $W_1$ and $W_{k-1}$. If there exists one such path, then we easily compute that the fraction of marked vertices, i.e., $\varepsilon$, satisfies
        \[\varepsilon \geq \frac{\binom{|V_1|-1}{s-1}}{\binom{|V_1|}{s}} \cdot \frac{\binom{|V_{k-1}|-1}{s-1}}{\binom{|V_{k-1}|}{s}} = \frac{s}{|V_1|} \cdot \frac{s}{|V_{k-1}|} \in \Theta\left(\left(\frac{s}{n}\right)^2\right).\]

        Next, we desribe the data structure that we store along the way. At every vertex $(W_1,W_{k-1})$, we store for every vertex $v \in W_1$ whether it can be reached from at least one vertex in $V_0$, and similarly for every vertex $v \in W_{k-1}$ whether there exists an edge going from $v$ to a vertex in $V_k$. Since these are simple search problems, we observe that the setup cost becomes $\mathsf{S} \in O(s\sqrt{r})$, and since in the update routine we always only exchange a constant number of vertices, we have $\mathsf{U} \in O(\sqrt{r})$. Finally, the checking routine, is the recursive procedure that only considers the vertices in $W_1$ and $W_{k-1}$ that are marked in the database, which results in a checking cost $\mathsf{C} = C_{n,s}^{(k-2)}$.

        Now, plugging everything into the MNRS-framework, i.e., \Cref{thm:mnrs}, yields the claimed result.
    \end{proof}

    Now that we have described the recursive algorithm, we should work out what the resulting number of queries becomes, for a given value of $k \in O(1)$. This is the objective of the following lemma.

    \begin{lemma}
        \label{lem:dir-path-recurrence-solution}
        Let $\alpha_r \in [0,1]$. For all $k \in \N$, let $\alpha_{k,r}$ be defined by the recurrence relation
        \[\alpha_{k+2,r} := \min_{\alpha_s \in [0,1]} \max\left\{\alpha_s + \frac12\alpha_r, 1 - \frac12\alpha_s + \frac12\alpha_r, 1 - \alpha_s + \alpha_{k,s}\right\},\]
        with $\alpha_{1,r} = \alpha_r$, and $\alpha_{2,r} = \frac12 + \frac12\alpha_r$. Next, let $n \in \N$, and $r \in \Theta(n^{\alpha_r})$. Then, we can solve the layered-path problem with first and last layer size $r$ with a number of queries that satisfies $C_{n,r}^{(k)} \in \widetilde{O}(n^{\alpha_{k,r}})$. Moreover, the solution to the recurrence relation satisfies $\alpha_{k,n} \in 3/2 - \Theta(c^{-k})$ with $c \approx 1.33$.
    \end{lemma}

    \begin{proof}
        We first observe that if $k = 1$, then we simply need to search over $\binom{r}{2}$ possible edges, which requires $O(r) = O(n^{\alpha_r}) = O(n^{\alpha_{1,r}})$ queries.
        
        Next, for $k = 2$, then we first search for the middle vertex, and then check whether there is a connection from the middle vertex to both the first and last layer. This requires a number of queries that satisfies $O(\sqrt{nr}) = O(n^{\frac12 + \frac12\alpha_r}) = O(n^{\alpha_{2,r}})$.

        Then, for every $k$, we observe by induction and \Cref{lem:dir-path-recurrence} that the number of queries required to solve the layered-path problem with length $k+2$ and first and last layer size $r$ satisfies
        \begin{align*}
            C_{n,r}^{(k+2)} &\in \widetilde{O}\left(\min_{\substack{s \in \N \\ 1 \leq s \leq n}} s\sqrt{r} + \left(\frac{n}{s}\right)\left(\sqrt{sr} + C_{n,s}^{(k)}\right)\right) \\
            &\subseteq O\left(\min_{\alpha_s \in [0,1]} n^{\alpha_s + \frac12\alpha_r} + n^{1-\frac12\alpha_s+\frac12\alpha_r} + n^{1-\alpha_s+\alpha_{k,s}}\right) = O(n^{\alpha_{k+2,r}}).
        \end{align*}
        
        Thus, it remains to solve the recurrence relation. To that end, we prove the following claim:
        
        \begin{claim}
            We consider the recurrence relations
            \[x_{k+2} = \frac{1+x_k}{2-y_k}, \quad \text{and} \quad y_{k+2} = \frac{1-y_k}{2(2-y_k)}, \quad \text{with} \quad x_2 = \frac12, x_3 = \frac23, y_2 = \frac12, y_3 = \frac12.\]
            These recurrence relations are well-defined, $(x_k + y_k)_{k=2}^{\infty}$ is increasing, and for all $k \geq 2$, we have $0 \leq 3/2 - (x_k + y_k) \in \Theta(c^{-k})$ where $c = \sqrt{3+\sqrt{17}}/2 \approx 1.33$. Finally, we have $\alpha_{k+2,r} = \max\{x_k + y_k, x_{k+2} + y_{k+2}\alpha_r\}$ for all $k \geq 2$.
        \end{claim}

        For future reference, we compute the first few values for $x_k$ and $y_k$.

        \begin{table}[!ht]
            \centering
            \begin{tabular}{r|ccccccccccc}
                $k$ & $2$ & $3$ & $4$ & $5$ & $6$ & $7$ & $8$ & $9$ & $10$ & $11$ & $12$ \\\hline
                &&&& \\[-1em]
                $y_k$ & $\frac12$ & $\frac12$ & $\frac16$ & $\frac16$ & $\frac{5}{22}$ & $\frac{5}{22}$ & $\frac{17}{78}$ & $\frac{17}{78}$ & $\frac{61}{278}$ & $\frac{61}{278}$ & $\frac{217}{990}$ \\[.2em]
                $x_k$ & $\frac12$ & $\frac23$ & $1$ & $\frac{10}{9}$ & $\frac{12}{11}$ & $\frac{38}{33}$ & $\frac{46}{39}$ & $\frac{142}{117}$ & $\frac{170}{139}$ & $\frac{518}{417}$ & $\frac{206}{165}$ \\[.2em]\hline
                &&&& \\[-1em]
                $x_k + y_k$ & $1$ & $\frac76$ & $\frac76$ & $\frac{23}{18}$ & $\frac{29}{22}$ & $\frac{91}{66}$ & $\frac{109}{78}$ & $\frac{335}{234}$ & $\frac{401}{278}$ & $\frac{1219}{834}$ & $\frac{1453}{990}$
            \end{tabular}
        \end{table}

        \begin{proof}
            We first assert well-definedness. To that end, observe that $|y_k| \leq 1/2$ implies that $|y_{k+2}| \leq (1+|y_k|)/(2(2-|y_k|)) \leq 1/2$, and so $y_k \neq 2$ for all $k \geq 2$. Consequently, $1 - y_k \geq 1/2$, and so $x_k \leq 2$ implies that $x_{k+2} \leq 3/(2-1/2) = 2$, hence by induction $x_k \leq 2$ for all $k \geq 2$.
            
            Next, we can easily verify that for all $\ell \geq 1$,
            \[y_{2\ell} = y_{2\ell+1} = \frac{\sqrt{17}}{2(-(-\chi)^\ell + 1)} + \varphi, \qquad \text{where} \qquad \chi = \frac{13 + 3\sqrt{17}}{4}, \quad \text{and} \quad \varphi = \frac{5 - \sqrt{17}}{4}.\]

            Now, we assert that $(x_k + y_k)_{k=2}^{\infty}$ is increasing. To that end, we write $\alpha = 1/(1-\varphi)$, and we bound for all $k \geq 2$ that
            \[\alpha - x_{k+2} = \frac{\varphi - y_k}{\alpha(2 - y_k)} + \frac{\alpha - x_k}{2 - y_k}, \qquad \text{which implies} \qquad \left|\alpha - x_{k+2} - \frac{\alpha - x_k}{2 - \varphi}\right| < \frac{3}{\chi^{\lfloor k/2\rfloor}}.\]
            As $(2 - \varphi)/\chi < 1/3$, we obtain that for all $k \geq 2$ and $\ell \in \N$,
            \begin{align*}
                \left|\alpha - x_{k+2\ell} - \frac{\alpha - x_k}{(2 - \varphi)^\ell}\right| &< \frac{3}{\chi^{\lfloor k/2\rfloor}} \sum_{m = 0}^{\ell-1} \frac{1}{(2 - \varphi)^m\chi^{\ell-1-m}} \\
                &= \frac{3}{\chi^{\lfloor k/2\rfloor}(2 - \varphi)^{\ell-1}} \sum_{m=0}^{\ell-1} \left(\frac{2 - \varphi}{\chi}\right)^m < \frac{9}{2\chi^{\lfloor k/2\rfloor}(2-\varphi)^{\ell-1}}.
            \end{align*}
            Now, we let $k \in \{10,11\}$, and we check that $(x_{k+1} - x_k)/(2 - \varphi) > 1/1000 > 10/\chi^5$. Thus, we observe that for all $\ell \in \N$ that
            \begin{align*}
                &(x_{k+2\ell+1} - y_{k+2\ell+1}) - (y_{k+2\ell} + x_{k+2\ell}) \\
                &\qquad = -(\alpha - x_{k+2\ell+1}) + (\alpha - x_{k+2\ell}) + (y_{k+2\ell+1} - \varphi) - (y_{k+2\ell} - \varphi) \\
                &\qquad > \frac{\alpha - x_{k+1}}{(2 - \varphi)^{\ell}} - \frac{\alpha - x_k}{(2 - \varphi)^{\ell}} - \frac{9}{\chi^{\lfloor k/2\rfloor}(2 - \varphi)^{\ell-1}} - \frac{5}{\chi^{\lfloor k/2\rfloor + \ell}} \\
                &\qquad > \frac{1}{(2 - \varphi)^{\ell-1}}\left[\frac{x_{k+1} - x_k}{2 - \varphi} - \frac{10}{\chi^{\lfloor k/2\rfloor}}\right] > 0.
            \end{align*}
            Hence, $(x_k+y_k)_{k=10}^{\infty}$ is indeed increasing, and we manually check it is for the first few values of $k$ as well.

            From the above analysis we immediately observe that $y_k \to \varphi$ and $x_k \to \alpha$, as $k \to \infty$. Moreover, we find
            \[\frac32 = \alpha + \varphi > x_k + y_k \in \frac32 - \Theta\left(c^{-k}\right),\]
            where $c = \min\{\sqrt{\chi}, \sqrt{2-\varphi}\} = \sqrt{3+\sqrt{17}}/2$.

            Finally, we prove that for all $k \geq 2$, $\alpha_{k+2,r} = \max\{x_k + y_k, x_{k+2} + y_{k+2}\alpha_r\}$. To that end, we work out the first couple of steps to obtain that
            \[\alpha_{3,r} = \max\left\{1, \frac23 + \frac{\alpha_r}{6}\right\}, \qquad \alpha_{4,r} = 1 + \frac{\alpha_r}{6}, \qquad \text{and} \qquad \alpha_{5,r} = \max\left\{\frac76, \frac{10}{9} + \frac{\alpha_r}{6}\right\},\]
            and so the basis for induction is established. For the induction step, we suppose that the relation holds for a given $k \geq 2$, and then observe that
            \[\alpha_{k+4,r} = \min_{\alpha_s \in [0,1]} \max\left\{\alpha_s + \frac{\alpha_r}{2}, 1 - \frac{\alpha_s}{2} + \frac{\alpha_r}{2}, 1 - \alpha_s + x_k + y_k, 1 - \alpha_s + x_{k+2} + y_{k+2}\alpha_s\right\}.\]
            Minimizing the first term in combination with the second, third and fourth terms yields $\alpha_s = 2/3$, $\alpha_s = \min\{1,(1+x_k+y_k)/2 - \alpha_r/4\}$ and $\alpha_s = \min\{1,(1+x_{k+2})/(2-y_{k+2}) - \alpha_r/(2(2-y_{k+2}))\}$, respectively. Thus,
            \[\alpha_{k+4,r} = \max\left\{\frac23 + \frac{\alpha_r}{2}, \frac{1+x_k+y_k}{2} + \frac{\alpha_r}{4}, x_k + y_k, x_{k+4} + y_{k+4}\alpha_r, x_{k+2} + y_{k+2}\right\}.\]
            We obtain $x_{k+2}+y_{k+2} \geq 7/6$ and $(x_k + y_k)_{k=2}^{\infty}$ is increasing, so the first and third term are always dominated by the fifth. For the second term, we compare to the fourth. We let $k \in \{8,9\}$, and we observe that $\alpha - x_k > 1/20 > 54/\chi^4$. Let $\ell \in \N_0$, and using $y_{k+2\ell+4} < 1/4$, we observe that
            \begin{align*}
                &x_{k+2\ell+4} + y_{k+2\ell+4}\alpha_r - \frac{1+x_{k+2\ell}+y_{k+2\ell}}{2} - \frac{\alpha_r}{4} \\
                &\qquad \geq x_{k+2\ell+4} + y_{k+2\ell+4} - \frac{3 + 2(x_{k+2\ell}+y_{k+2\ell})}{4} \\
                &\qquad = \frac12\left(\frac32 - (x_{k+2\ell} + y_{k+2\ell})\right) - \left(\frac32 - (x_{k+2\ell+4} + y_{k+2\ell+4})\right) \\
                &\qquad > \left[\frac12 - \frac{1}{(2-\varphi)^2}\right]\left(\frac32 - (x_{k+2\ell}+y_{k+2\ell})\right) - \frac{8}{\chi^{\lfloor k/2\rfloor+\ell}} > \frac{\alpha - x_k}{6(2-\varphi)^{\ell}} - \frac{9}{\chi^{\lfloor k/2\rfloor+\ell}} \\
                &\qquad > \frac{1}{6(2-\varphi)^{\ell}}\left[\alpha - x_k - \frac{54}{\chi^{\lfloor k/2\rfloor}}\right] > 0.
            \end{align*}
            Thus, for all $k \geq 8$, the second term is dominated by the fourth, and we manually check this is the case for smaller values of $k$ as well.
        \end{proof}
        
        Finally, we plug in $r = n$, and hence $\alpha_r = 1$, to conclude that $\alpha_{k,n} = x_k + y_k \in 3/2 - \Theta(c^{-k})$, where $c \approx 1.33$.
    \end{proof}

    Putting everything together, we now obtain the following result.

    \begin{corollary}
        \label{cor:dir-path-detection}
        $\mathsf{Q}(\mathtt{DirPath}^{=k}) \in \widetilde{O}(n^{\frac32 - \alpha_k})$, where $\alpha_k \in \Omega(c^{-k})$ as $k \to \infty$, with $c = \sqrt{3+\sqrt{17}}/2 \approx 1.33$.
    \end{corollary}

    \begin{proof}
        We combine \Cref{lem:dir-path-color-coding} with \Cref{lem:1-sided_error} to obtain that solving the layered path-finding problem with length $k$ and first and last layer size $n$ is sufficient for finding a directed path of length $k$ in a directed graph. The result then follows from plugging in $r \in \Theta(n)$ in \Cref{lem:dir-path-recurrence-solution}.
    \end{proof}

    We compute the resulting query complexity for the first few values of $k$ for good measure.

    \begin{table}[!ht]
        \centering
        \begin{tabular}{r|ccccccc}
            $k$ & $1$ & $2$ & $3$ & $4$ & $5$ & $6$ & $7$ \\\hline\\[-1em]
            $\mathsf{Q}(\mathtt{DirPath}^{=k})$ & $\widetilde{O}(n)$ & $\widetilde{O}(n)$ & $\widetilde{O}(n^{7/6})$ & $\widetilde{O}(n^{7/6})$ & $\widetilde{O}(n^{23/18})$ & $\widetilde{O}(n^{29/22})$ & $\widetilde{O}(n^{91/66})$
        \end{tabular}
        \caption{The query complexity of $\mathsf{Q}(\mathtt{DirPath}^{=k})$ for the first few values of $k$.}
        \label{tbl:dir-path-detection}
    \end{table}

    \section{Length-$\leq k$ cycle-detection algorithm}
    \label{sec:cycle-leq-k-algo}

    In this section, we modify the approach developed in \cite{childs2012quantum}, to obtain a novel quantum algorithm for the $\mathtt{Cycle}^{\leq k}$-problem, whenever $k$ is odd. The core result that we are using is the following theorem.

    \begin{theorem}[{\cite[Theorem~4.7]{childs2012quantum}}] \label{thm:vertex_cover}
        Let $\mathcal{P}$ be the property that an $n$-vertex graph either has more than $\overline{m}$ edges, where $\overline{m} \in \Omega(n)$, or contains a given subgraph $H$. Let $H'$
        be the graph obtained by deleting all degree-$1$ vertices of $H$ that are not part of an isolated edge. Then $\Q(\mathcal{P}) \in \widetilde{O}(\sqrt{\overline{m}}n^{1-\frac{1}{\text{vc}(H')+1}})$, where $\text{vc}(H')$ denotes the vertex cover of $H'$, i.e., the minimum number of vertices in $H'$ required to cover all edges.
    \end{theorem}

    Applying this to $k$-cycles, we get the following corollary.

    \begin{corollary} \label{cor:sparse_cycle_alg}
        Let $k \geq 3$. Let $\mathcal{P}$ be the property that an $n$-vertex graph either has more than $\overline{m}$ edges, where $\overline{m} \in \Omega(n)$, or contains a $k$-cycle. Then, there is a quantum algorithm that computes $\mathcal{P}$ using $\widetilde{O}(\sqrt{\overline{m}}n^{1-\frac{1}{\ceil{k/2}+1}})$ queries.
    \end{corollary}

    We modify the algorithm attaining the result in \Cref{thm:vertex_cover} (and \Cref{cor:sparse_cycle_alg}) to get a slight improvement for the case when $k$ is odd (making a stronger assumption on $\overline{m}$ which will also be satisfied for our applications). 

    \begin{theorem} \label{thm:sparse_cycle}
        Let $k \geq 3$. Let $\mathcal{P}$ be the property that an $n$-vertex graph either has more than $\overline{m}$ edges, where $\overline{m} \in \Omega(n^{1+1/(\ell+1)})$ and $\ell = \floor{k/2}$, or contains a $k$-cycle. Then, there is a quantum algorithm \textsc{SparseCycle} that computes $\mathcal{P}$ using $\widetilde{O}(\sqrt{\overline{m}}n^{1-\frac{1}{\ell+1}})$ queries.
    \end{theorem}
    
    \begin{proof}
        Our algorithm is very similar to that of \Cref{thm:vertex_cover} by \cite{childs2012quantum}. The checking cost of their quantum walk procedure is $0$ since their walk database contains a $k$-cycle whenever a vertex is marked. In a nutshell, for the $k$-cycle problem, we improve the complexity of the setup and update steps in their walk at the cost of increasing the checking complexity (which will still be bounded by the complexity corresponding to the setup and update steps). We provide the algorithm and the complexity analysis for completeness.  
        
        We use \Cref{cor:approx_counting} to decide if the number of edges $m$ in the input graph is at least $\overline{m}$ or at most $\overline{m}/3$ edges with probability at least $1-1/n$ using $\widetilde{O}(\sqrt{n^2/\overline{m}}) \subseteq \widetilde{O}(\sqrt{n})$ queries. If this procedure accepts, we accept; otherwise, we suppose that $m \in O(\overline{m})$ and proceed as follows. 
        %The algorithm we use to prove this theorem is quite similar to the algorithm used by \cite{childs2012quantum} to prove \Cref{thm:vertex_cover}. We will explicitly state this algorithm for completeness.

        We will assume that $k$ is odd since for the case when $k$ is even, we get the desired result directly from \Cref{cor:sparse_cycle_alg}.

        We will partition the vertex set $V$ into $\ceil{\log n}$ many parts according to its degree. We will not explicitly compute this partition; instead, we will maintain the relevant information about this partition in our quantum walk (as we explain below).  
        
        For $i \in [\ell]$, fix $p_i \in [\ceil{\log n}]$ and $q_i = 2^{p_i}$. We will describe an algorithm for determining if there are vertices $v_i$ for $i \in [\ell]$ with degree close to $q_i$ (i.e. $[q_i/2, 3q_i/2]$) such that $v_1, \dots, v_\ell$ are non-adjacent vertices of a $k$-cycle. We can run this procedure for all $\ceil{\log n}^\ell$ tuples of length $\ell$ comprising of elements from $[\ceil{\log n}]$, incurring a polylogarithmic overhead.

        We will use a quantum walk procedure on a Johnson graph invoking \Cref{thm:mnrs}. In the walk database, for each $i \in [\ell]$, we will store $r_i$ vertices whose degree is close to $q_i$ and their neighborhoods. Note that the number of vertices whose degree is in $[q_i/2, 3q_i/2]$ is $t_i \in O(\overline{m}/q_i)$.
        %, and for each stored vertex $v$ among these $r_i$ vertices, we will store the whole neighborhood of $v$ if $\deg(v) \in [2^{p_i-1}, 2^{p_i}]$. 
        
        To determine if a vertex $v$ has degree at least $q_i/2$ or at most $q_i/6$ with probability at least $1-1/n^3$, we can use the approximate degree counting algorithm in  \Cref{cor:approx_degree_counting} using $\tilde{O}(\sqrt{n/q_i})$ queries. Similarly, to determine if a vertex $v$ has degree at least $9q_i/2$ or at most $3q_i/2$ with probability at least $1-1/n^3$, we can use the approximate degree counting algorithm in \Cref{cor:approx_degree_counting} using $\tilde{O}(\sqrt{n/q_i})$ queries. Thus, for a vertex $v$, if $\deg(v) \in [q_i/2, 3q_i/2]$, then it is accepted by the first check and rejected by the second check with probability at least $1-2/n^3$.
        Since there are $t_i$ vertices with degree in $[q_i/2, 3q_i/2]$, the probability that a randomly sampled vertex has such degree is $t_i/n$. Thus, via Grover's algorithm, the total cost of searching for a vertex with degree close to $q_i$ is $\widetilde{O}(\sqrt{n/t_i} \cdot \sqrt{n/q_i}) \subseteq \widetilde{O}(n/\sqrt{t_i})$. 
        %To determine if a vertex $v$ has degree close to $2^{p_i}$, we can use the approximate degree counting algorithm in  \Cref{cor:approx_degree_counting} with $\epsilon = 1/2$ and $\delta = 1/n^3$ and if the estimate $\tilde{t} \in [2^{p_i-2}, 9\cdot 2^{p_i-2}]$, we accept and otherwise, we reject. Note that $\tilde{t} \in [2^{p_i-2}, 9 \cdot 2^{p_i-2}]$ implies that $\deg(v) \in [2^{p_i-1}/3, 9 \cdot 2^{p_i-1}]$ with probability at least $1-1/n^3$ which is fine since $\deg(v) = \Theta(2^{p_i})$ and if $\deg(v) \in [2^{p_i-1}, 3 \cdot 2^{p_i-1}]$ (i.e. close to $2^{p_i}$), then $\tilde{t} \in [2^{p_i-2}, 9 \cdot 2^{p_i-2}]$ with probability at least $1-1/n^3$. Since there are $t_i$ vertices with degree close to $2^{p_i}$, the probability that a randomly sampled vertex have such degree is $n/t_i$. Thus, via Grover's algorithm \Amin{cite this?}, the total cost of searching for a vertex with degree close to $2^{p_i}$ is $\tilde{O}(\sqrt{n/t_i} \cdot \sqrt{n}) = \tilde{O}(n/\sqrt{t_i})$.   
        Computing the neighborhood for any $v$ with its degree close to $q_i$ will cost $\widetilde{O}(\sqrt{n \deg(v)}) = \widetilde{O}(\sqrt{n \cdot q_i})$ by \Cref{lem:find_all_marked_elements}.
        
        For each update step of our walk, for each $i \in [\ell]$, we randomly remove $s_i$ of the $r_i$ vertices (and their neighbors) and add $s_i$ other vertices with degree close to $q_i$ and their neighborhoods. We say that a vertex of our walk is marked if the database associated with this vertex contains $\ell$ non-adjacent vertices of a $k$-cycle and their neighborhood in $\ell$ different blocks. 

        We are now ready to compute the cost of our quantum walk. The setup (and update) cost involves, for each $i \in [\ell]$, searching for $r_i$ (respectively $s_i$) vertices with degree close to $q_i$ and computing their neighborhoods. We earlier noticed that for each vertex, the former step costs $\widetilde{O}(n/\sqrt{t_i})$ and the latter step costs $\widetilde{O}(\sqrt{n \cdot q_i})$ queries. Thus, the total cost for each vertex is $\widetilde{O}(n/\sqrt{t_i} + \sqrt{n \cdot q_i}) \subseteq \widetilde{O}(\sqrt{n \overline{m}/t_i})$ since $t_i \cdot q_i \in O(\overline{m})$ and $\overline{m} \in \Omega(n)$. It follows that the setup and update costs are $S \in \widetilde{O}(\sum_{i \in [\ell]} r_i \cdot \sqrt{n \overline{m}/t_i})$ and $U \in \widetilde{O}(\sum_{i \in [\ell]} s_i \cdot \sqrt{n \overline{m}/t_i})$ respectively.
        %The setup (and update) cost involves checking if for each $i \in [\ell]$, each of the $r_i$ (respectively $s_i$) vertices have their degrees falling in the appropriate region and computing their neighborhood if it does. For each vertex, the former step costs $\tilde{O}(\sqrt{n/2^{p_i}}) = \tilde{O}(\sqrt{n})$, and the latter step costs    
        % Let $v_1, \dots, v_k$ be vertices of a $k$-cycle.
        The checking procedure involves determining if the database contains $\ell$ non-adjacent vertices of a $k$-cycle in $\ell$ different blocks. Since we also store the neighborhoods of all the vertices, we need to decide if there are $\ell$ vertices in $\ell$ different blocks forming a length-$2\ell$ path that can be extended to a cycle of length $k = 2\ell+1$. This can be done by searching for an edge that connects a pair of endpoints of any such path in $O(\sqrt{n^2}) = O(n)$ queries using Grover's algorithm. The spectral gap $\delta$ of our walk is $\Omega(\min_{i \in [\ell]}s_i/r_i)$ and the probability $\epsilon$ of marked vertices of our walk is $\Omega(\prod_{j \in [\ell]} r_j/t_j)$. Therefore, the total cost of our walk is
        \begin{equation*}
            S + \frac{1}{\sqrt{\epsilon}} \left( \frac{1}{\sqrt{\delta}}U + C\right) \in \widetilde{O}\left(\sum_{i \in [\ell]} r_i \cdot \sqrt{\frac{n \overline{m}}{t_i}} + \prod_{j \in [\ell]} \sqrt{\frac{t_j}{r_j}} \left(\max_{i \in [\ell]} \sqrt{\frac{r_i}{s_i}} \left(\sum_{i \in [\ell]} s_i \cdot \sqrt{\frac{n \overline{m}}{t_i}}\right) + n \right)\right)
        \end{equation*}

        We will let 
        \begin{equation*}
            \frac{r_i}{r_j} = \frac{s_i}{s_j} = \sqrt{\frac{t_i}{t_j}}, \qquad \qquad s_1 = 1, \qquad \qquad r_1 = \sqrt{t_1} n^{\frac{1}{2}-\frac{1}{\ell+1}},
        \end{equation*}
        and argue that all the three terms in the above expression are bounded by $\widetilde{O}(\sqrt{\overline{m}}n^{1-\frac{1}{\ell+1}})$. First, notice that

        \begin{align*}
            \sum_{i \in [\ell]} r_i \cdot \sqrt{\frac{n \overline{m}}{t_i}} = \ell \cdot r_1 \cdot \sqrt{\frac{n \overline{m}}{t_1}} \in O (\sqrt{\overline{m}} n^{1-\frac{1}{\ell+1}}).
        \end{align*}

        For the second term, we first simplify
        \begin{align*}
            \prod_{j \in [\ell]} \frac{t_j}{r_j} = \left(\frac{\sqrt{t_1}}{r_1}\right)^\ell \cdot \prod_{j \in [\ell]}\sqrt{t_j} = \frac{t_1}{r_1} \cdot \left(\frac{\sqrt{t_1}}{r_1}\right)^{\ell-1} \prod_{j \in [2,\ell]}\sqrt{t_j} \leq \frac{t_1}{r_1} \cdot \left(\frac{\sqrt{t_1 n}}{r_1}\right)^{\ell-1} = \frac{t_1}{r_1} n^{\frac{\ell-1}{\ell+1}}
            %= n^{\frac{1}{2}-\frac{1}{\ell+1}} \cdot \prod_{j \in [\ell]}\sqrt{t_j} \leq \sqrt{t_1} \cdot n^{\frac{1}{2}-\frac{1}{\ell+1}} \cdot 
        \end{align*}
        so
        \begin{align*}
            \prod_{j \in [\ell]} \sqrt{\frac{t_j}{r_j}} = \sqrt{\frac{t_1}{r_1}} n^{\frac{\ell-1}{2(\ell+1)}}.
        \end{align*}
        
        We also simplify
        \begin{align*}
            \max_{i \in [\ell]} \sqrt{\frac{r_i}{s_i}} \left(\sum_{i \in [\ell]} s_i \cdot \sqrt{\frac{n \overline{m}}{t_i}}\right) =  \ell \cdot \sqrt{r_1 s_1} \cdot \sqrt{\frac{n \overline{m}}{t_1}} \in O\left(\sqrt{r_1 s_1} \cdot \sqrt{\frac{n \overline{m}}{t_1}} \right) = O\left(\sqrt{\frac{r_1} {t_1}} \cdot \sqrt{n \overline{m}} \right).
        \end{align*}
        
        Taking their products gives us
        \begin{align*}
            \prod_{j \in [\ell]} \sqrt{\frac{t_j}{r_j}} \left(\max_{i \in [\ell]} \sqrt{\frac{r_i}{s_i}} \left(\sum_{i \in [\ell]} s_i \cdot \sqrt{\frac{n \overline{m}}{t_i}}\right) \right) \in O\left(\sqrt{n \overline{m}} \cdot n^{\frac{\ell-1}{2(\ell+1)}}\right) = O (\sqrt{\overline{m}} n^{1-\frac{1}{\ell+1}}).
        \end{align*}

        It remains to argue that the term corresponding to the checking cost is bounded by $O (\sqrt{\overline{m}} n^{1-\frac{1}{\ell+1}})$. To show this, it is sufficient to argue that $C \in O(1/\sqrt{\delta} \cdot  U)$, which we do as follows.
        \begin{align*}
            \frac{1}{\sqrt{\delta}} U \geq \sqrt{\frac{r_1} {t_1}} \cdot \sqrt{n \overline{m}} = \sqrt{\frac{n^{\frac{1}{2}-\frac{1}{\ell+1}}}{\sqrt{t_1}}} \cdot \sqrt{n \overline{m}} \geq n^{\frac{1}{2}-\frac{1}{2(\ell+1)}} \cdot \sqrt{\overline{m}} \in \Omega(n) \subseteq \Omega(C).
        \end{align*}
        where the last step follows since $\overline{m} \in \Omega(n^{1+1/(\ell+1)})$.
    \end{proof}

    Now, we can combine this construction with a combinatorial result by Bondy and Simonovits~\cite{bondy1974cycles}, that connects the non-existence of an even-length cycle of a given length to a sparsity condition on the graph.

    \begin{theorem}[\cite{bondy1974cycles}] \label{thm:dense_cycles}
        Let $G$ be a graph on $n$ vertices. For any $\ell \geq 2$, if $|E(G)| > 100 \ell n^{1+1/\ell}$ then $G$ contains $\mathcal{C}_{2\ell}$ as a subgraph.
    \end{theorem}

    Combining the previous two results now gives us our algorithm that solves the $\mathtt{Cycle}^{\leq k}$-problem.

    \begin{theorem} \label{thm:cycles<=k}
        Let $k \geq 4$. Then, $\mathsf{Q}(\mathtt{Cycle^{\leq k}}) \in \widetilde{O}(n^{\frac{3}{2}-\gamma(k)})$, where
        \begin{equation*}
            \gamma(k) = 
            \begin{cases}
                \frac{k-2}{k(k+2)}, & \text{if $k$ is even}, \\
                \frac{k-3}{(k-1)(k+1)}, & \text{if $k$ is odd}.
            \end{cases}
        \end{equation*}
    \end{theorem}

    \begin{proof}
        Let $\ell = \floor{k/2}$. The general strategy of our algorithm is to first check if $G$ is dense enough to already have a $\mathcal{C}_{2\ell}$ using \Cref{thm:dense_cycles}. If it is not, then we know that $E(G) \in O(n^{1+1/\ell})$, which we can use to efficiently compute if there is a $\mathcal{C}_j$ for any $j \in [3,k]$ using \Cref{thm:sparse_cycle}. See \Cref{alg:cycle<=k} for a formal description of this algorithm. 
    
        The correctness is easy to check. It remains to argue the desired complexity. For the density reduction stage, it follows from \Cref{cor:approx_counting} that $\widetilde{O}(\sqrt{n^2/\overline{m}}) = \widetilde{O}(\sqrt{n^{1-1/\ell}})$ queries would be sufficient. If the algorithm does not abort in the density reduction stage, then we can suppose that $E(G) \in O(\overline{m}) = O(n^{1+1/\ell})$. The complexity expression in \Cref{thm:sparse_cycle} is non-decreasing in $k$ so the complexity contribution for the iteration where $j=k$ will dominate the cost of the for loop in the sparse case stage. In fact, this step will dominate the cost of \Cref{alg:cycle<=k} since the cost of the density reduction stage is sublinear. From \Cref{cor:sparse_cycle_alg}, we deduce that the cost will be
        \begin{equation*}
            \widetilde{O}(\sqrt{\overline{m}} n^{1-\frac{1}{\ell+1}}) = \widetilde{O}(n^{\frac{3}{2} - \frac{\ell-1}{2\ell (\ell+1)}}) =
            \begin{cases}
                 \widetilde{O}(n^{\frac{3}{2} - \frac{k-2}{k(k+2)}}), & \text{if $k$ is even}, \\
                 \widetilde{O}(n^{\frac{3}{2} - \frac{k-3}{(k-1)(k+1)}}), & \text{if $k$ is odd}.
            \end{cases}\qedhere
        \end{equation*}
        
        \begin{algorithm}[!ht]
            \caption{Algorithm for $\mathtt{Cycle}^{\leq k}$ for $k \geq 4$}
            \begin{algorithmic}[1] 
            \label{alg:cycle<=k}
                \State $\ell \gets \lfloor k/2 \rfloor$
                \State $\overline m \gets 100 \ell \, n^{1+1/\ell}$
                \medskip
                \State \textbf{(Density reduction)}
                \State Run \textsc{ApproxEdgeCount} from \Cref{cor:approx_counting} to distinguish if
                $|E(G)| \ge \tfrac32 \overline m$ or $|E(G)| \le \tfrac12 \overline m$ with prob $1/n$.
                \If{$|E(G)| \ge \tfrac32 \overline m$}
                    \State \Return \textsc{yes}
                    %\Comment{By Theorem~4.10, $G$ contains $C_{2\ell}$, hence a cycle of length at most $k$}
                \EndIf
    
                \medskip
                \State \textbf{(Sparse case)}
                %\Comment{Now proceed under the promise $|E(G)| = O(\overline m)$}
                \For{$j \gets 3$ \textbf{to} $k$}
                \State Run \textsc{SparseCycle} from \Cref{thm:sparse_cycle}.
                \If{$\mathcal{A}$ outputs \textsc{yes}}
                    \State \Return \textsc{yes}
                \EndIf
                \EndFor
                \State \Return \textsc{no}
            \end{algorithmic}
        \end{algorithm}
    \end{proof}

    Finally, we summarize the best-known upper bounds on the query complexities for $\mathtt{Cycle}^{=k}$ and $\mathtt{Cycle}^{\leq k}$ in \Cref{fig:ub-cycle}. We observe that for all $k$, the $\mathtt{Cycle}^{\leq k}$-problem can be solved in $o(n^{3/2})$ queries, whereas for the $\mathtt{Cycle}^{=k}$-problem, we can only make that statement for even values of $k$. It would be a very interesting direction of future research to establish whether the $\mathtt{Cycle}^{=k}$-problem for odd values of $k$ can be solved in $\widetilde{O}(n^{3/2})$ queries as well.

    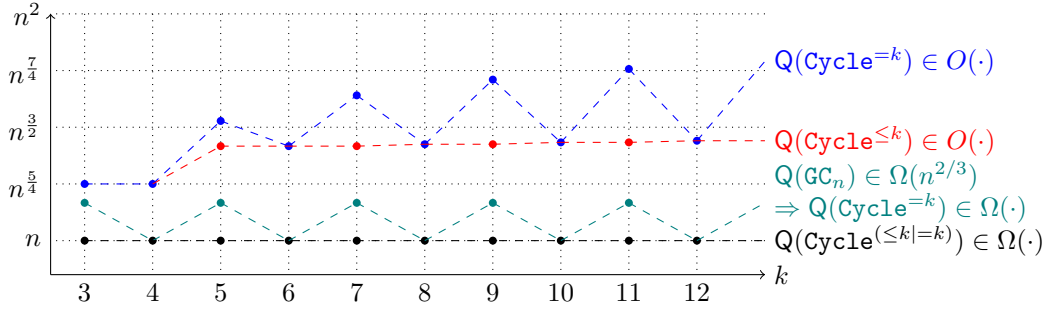
\begin{figure}[!ht]
        \centering
        \begin{tikzpicture}[vertex/.style={fill, inner sep = .15em, rounded corners = .15em}, yscale=3, xscale=.9]
            \draw[->] (2.5,.85) to (13,.85) node[right] {$k$};
            \draw[->] (2.5,.85) to (2.5,2);
            \draw[dotted] (13,1) to (2.5,1) node[left] {$n$};
            \draw[dotted] (13,{5/4}) to (2.5,{5/4}) node[left] {$n^{\frac54}$};
            \draw[dotted] (13,{3/2}) to (2.5,{3/2}) node[left] {$n^{\frac32}$};
            \draw[dotted] (13,{7/4}) to (2.5,{7/4}) node[left] {$n^{\frac74}$};
            \draw[dotted] (13,2) to (2.5,2) node[left] {$n^2$};

            \foreach \k in {3,...,12} {
                \draw[dotted] (\k,2) to (\k,.85) node[below] {$\k$};
                \node[vertex] at (\k,1) {};
            }
            
            \node[vertex,blue] at (3,{5/4}) {};
            \node[vertex,blue] at (4,{5/4}) {};
            \draw[dashed,blue] (3,{5/4}) to (4,{5/4});
            \draw[dashed,blue] (4,{5/4}) to (5,{55/36});
            \draw[dashed,red] (4,{5/4}) to (5,{17/12});
            \draw[dashed,teal] (3,{7/6}) to (4,1);
            \draw[dashed,teal] (4,1) to (5,{7/6});

            \node[vertex,teal] at (3,{7/6}) {};
            \draw[dashed] (3,1) to (13,1);

            \foreach \k in {6,8,10,12} {
                \node[vertex,blue] at (\k,{3/2-(\k-2)/(\k*(\k+2))}) {};
                \draw[dashed,blue] (\k,{3/2-(\k-2)/(\k*(\k+2))}) to ({\k+1},{2-3/(\k+2)+1/(\k+2)^2});
                \draw[dashed,red] (\k,{3/2-(\k-2)/(\k*(\k+2))}) to ({\k+1},{3/2-(\k-2)/(\k*(\k+2))});
                \draw[dashed,teal] (\k,1) to ({\k+1},{7/6});
            }
            \foreach \k in {5,7,9,11} {
                \node[vertex,blue] at (\k,{2-3/(\k+1)+1/(\k+1)^2}) {};
                \draw[dashed,blue] (\k,{2-3/(\k+1)+1/(\k+1)^2}) to ({\k+1},{3/2-(\k-1)/((\k+1)*(\k+3))});
                
                \node[vertex, red] at (\k,{3/2-(\k-3)/((\k-1)*(\k+1))}) {};
                \draw[dashed,red] (\k,{3/2-(\k-3)/((\k-1)*(\k+1))}) to ({\k+1},{3/2-(\k-1)/((\k+1)*(\k+3))});
                
                \node[vertex,teal] at (\k,{7/6}) {};
                \draw[dashed,teal] (\k,{7/6}) to ({\k+1},1);
            }

            \node[right,blue] at (13,{2-3/(12+2)+1/(12+2)^2}) {$\mathsf{Q}(\mathtt{Cycle}^{=k}) \in O(\cdot)$};
            \node[right,red] at (13,{3/2-(12-2)/(12*(12+2))}) {$\mathsf{Q}(\mathtt{Cycle}^{\leq k}) \in O(\cdot)$};
            \node[right,black] at (13,1) {$\mathsf{Q}(\mathtt{Cycle}^{(\leq k|=k)}) \in \Omega(\cdot)$};
            \node[right,teal,align=left] at (13,{7/6+.05}) {$\mathsf{Q}(\mathtt{GC}_n) \in \Omega(n^{2/3})$ \\
            $\Rightarrow \mathsf{Q}(\mathtt{Cycle}^{=k}) \in \Omega(\cdot)$};
        \end{tikzpicture}
        \caption{The best-known upper bounds on the query complexity of the $\mathtt{Cycle}^{=k}$-problem (shown in blue) and the $\mathtt{Cycle}^{\leq k}$-problem (shown in red). The black vertices represent the best-known unconditional lower bounds on both problems, and the green dots represent the best-known conditional lower bounds on the $\mathtt{Cycle}^{=k}$-problem, conditional on $\mathsf{Q}(\mathtt{GC}_n) \in \Omega(n^{2/3})$.}
        \label{fig:ub-cycle}
    \end{figure}
    
    We also observe that the expressions we obtain for general values of $k$ do not recover the best-known upper bounds for the cases where $k = 3$ and $k = 4$. This suggests that for higher values of $k$, the upper bounds we know for these cycle-containment problems are likely not tight. Therefore, we think it's a very nice future direction of research to improve the best-known algorithms for these problems.
    
    \section{Conditional lower bounds based on graph collision}
    \label{sec:lower-bounds}

    In this final section, we prove a novel super-linear conditional lower bound on the query complexities of several problems in \Cref{fig:islands}. Our lower bounds will be based on the graph-collision problem, introduced in \cite[Section~4.2]{magniez2007quantum}.

    \begin{definition}[{Graph collision problem~\cite[Section~4.2]{magniez2007quantum}}]
        Let $G = (V,E)$ be an undirected graph. We consider the boolean function $\mathtt{GC}_G = \mathtt{GraphCollision}_G : \{0,1\}^V \to \{0,1\}$, which evaluates to $1$ on input $x \in \{0,1\}^V$ if and only if there exists an edge $vw \in E$, such that $x_v = x_w = 1$. This is the \textit{graph collision problem} on $G$. For all $n \in \N$,
        \[\mathsf{Q}(\mathtt{GC}_n) := \max_{\substack{G = (V,E) \\ |V| = n}} \mathsf{Q}(\mathtt{GC}_G).\]
    \end{definition}

    We find by a straightforward reduction to search that $\mathsf{Q}(\mathtt{GC}_n) \in \Omega(\sqrt{n})$. Conversely, we have that $\mathsf{Q}(\mathtt{GC}_n) \in O(n^{2/3})$, see for instance~\cite[Theorem~3]{magniez2007quantum}.
    
    Balodis and Iraids~\cite{balodis2016quantum} showed that the quantum query complexity of triangle finding is related to the graph collision problem, which gives us a conditional lower bound on several of the equivalence classes displayed in \Cref{fig:islands}. We modify Balodis's and Iraids's approach to give the same lower bound to the cycle-finding problem through a vertex $s$.

    \begin{proposition}
        \label{prop:cycle-s-lb}
        $\mathsf{Q}(\mathtt{Cycle}_s^{=5}) \in \Omega(\sqrt{n} \cdot \mathsf{Q}(\mathtt{GC}_n))$.
        
        % Let $k \geq 5$ be an odd integer. Then, $\mathsf{Q}(\mathtt{Cycle}_s^{\leq k}) \in \Omega(\sqrt{n} \cdot \mathsf{Q}(\mathtt{GC}_n))$. \Arjan{Just proving it for $k = 5$ should be enough with the new reductions.}
    \end{proposition}

    \begin{proof}
        %First, we suppose that $k=5$. 
        % Recall that the inputs to $\mathtt{PromOR}_n$ are promised to have hamming weight at most $1$, and \textsc{yes} instances comprise of strings with hamming weight $1$ and \textsc{no} instances comprise of strings with hamming weight $0$.
        
        Let $G = (V,E)$ be any graph on $n$ vertices, where we write $V = \{v_1, \dots, v_n\}$. We will reduce $\mathtt{GC}_G \circ \mathtt{OR}_n$ to $\mathtt{Cycle}_s^{\leq k}$. That is, given query access to individual bits of $n$ strings $x^{(i)} \in \{0,1\}^n$ with $x^{(i)}$ associated with $v_i$, which forms an instance of $\mathtt{GC}_G \circ \mathtt{OR}_n$, we construct a graph instance $G' =(V',E')$ of $\mathtt{Cycle}_s^{=5}$ as follows. Let $V' = \{v_j^{(\ell)} : j \in [n], \ell \in [4]\} \cup \{s\}$. We define
        \[\overline{E} = \{\{s,v_j^{(1)}\},\{s,v_j^{(4)}\} : j \in [n]\} \cup \{\{v_j^{(2)}, v_k^{(3)}\}, \{v_k^{(2)}, v_j^{(3)}\} : \{v_j,v_k\} \in E\},\]
        and
        \[E_x = \{\{v_j^{(1)},v_k^{(2)}\},\{v_k^{(3)}, v_j^{(4)}\} : x_j^{(k)} = 1\}.\]
        
        We now argue that $G' = (V',\overline{E} \cup E_x)$ is a positive input to $\mathtt{Cycle}^{=5}$ if and only if $x$ is a positive input to $\mathtt{GC}_G \circ \mathtt{OR}_n$. To that end, suppose that $x$ is a positive input for $\mathtt{GC}_G \circ \mathtt{OR}_n$. Then, there exist $j,k \in [n]$ such that $\{v_j,v_k\} \in E$, and that there are $\ell_j, \ell_k \in [n]$ such that $x_{\ell_j}^{(j)} = x_{\ell_k}^{(k)} = 1$, and so we have a cycle $s$ --- $v_{\ell_j}^{(1)}$ --- $v_j^{(2)}$ --- $v_k^{(3)}$ --- $v_{\ell_k}^{(4)}$ --- $s$ of length $5$ in $G'$.

        For the reverse direction, suppose that $G'$ has a cycle of length $5$ passing through $s$. As $G'$ is a color-coded graph with an odd number of layers, any cycle of length $5$ must traverse through all the layers, which means its path can be written as $s$ --- $v_{\ell_j}^{(1)}$ --- $v_j^{(2)}$ --- $v_k^{(3)}$ --- $v_{\ell_k}^{(4)}$ --- $s$, for suitable choices of $\ell_j, j, k, \ell_k \in [n]$. This then implies that $\{v_j,v_k\} \in E$, and $x_{\ell_j}^{(j)} = x_{\ell_k}^{(k)} = 1$, which means that $x$ is a positive instance for $\mathtt{GC}_G \circ \mathtt{PromOR}_n$.

        Thus, we find that $\mathsf{Q}(\mathtt{Cycle}_s^{=5}) \geq \mathsf{Q}(\mathtt{GC}_G \circ \mathtt{OR}_n) \in \Theta(\sqrt{n} \cdot \mathsf{Q}(\mathtt{GC}_G))$, for all graphs $G$ with $n$ vertices, using that the quantum query complexity multiplies under composition~\cite{reichardt2011reflections}. Finally, we can take the maximum over all such $G$ on the right-hand side of the equation to obtain the claimed result.
    \end{proof}

    \section*{Acknowledgments}
    
    We would like to thank Andrew Childs and Matthew Coudron for many useful discussions. A.C.\ acknowledges support by a Simons-CIQC postdoctoral fellowship through NSF QLCI Grant No.\ 2016245. A.S.G.\ acknowledges support from the U.S. Department of Energy, Office of Science, Accelerated Research in Quantum Computing, Fundamental Algorithmic Research toward Quantum Utility (FAR-Qu) and the National Institute of Standards and Technology (NIST). S.P.\ acknowledges the support from the Dutch Ministry of Education, Culture, and Science through Gravitation project ``Challenges in Cyber Security - 024.006.037'' for this work.
    
    \bibliographystyle{alphaurl}

    \newpage
    
    \bibliography{references}

@article{childs2012quantum,
  title={Quantum query complexity of minor-closed graph properties},
  author={Childs, Andrew M and Kothari, Robin},
  journal={SIAM Journal on Computing},
  volume={41},
  number={6},
  pages={1426--1450},
  year={2012},
  publisher={SIAM}
}

@article{lee2012learning,
  title={Learning graph based quantum query algorithms for finding constant-size subgraphs},
  author={Lee, Troy and Magniez, Fr{\'e}d{\'e}ric and Santha, Miklos},
  journal={Chicago Journal of Theoretical Computer Science},
  volume={10},
  pages={1--21},
  year={2012}
}

@inproceedings{le2014improved,
  title={Improved quantum algorithm for triangle finding via combinatorial arguments},
  author={Le Gall, Fran{\c{c}}ois},
  booktitle={2014 IEEE 55th Annual Symposium on Foundations of Computer Science},
  pages={216--225},
  year={2014},
  organization={IEEE}
}

@article{carette2020extended,
  title={Extended learning graphs for triangle finding},
  author={Carette, Titouan and Lauri{\`e}re, Mathieu and Magniez, Fr{\'e}d{\'e}ric},
  journal={Algorithmica},
  volume={82},
  number={4},
  pages={980--1005},
  year={2020},
  publisher={Springer}
}

@inproceedings{jeffery2013nested,
  title={Nested quantum walks with quantum data structures},
  author={Jeffery, Stacey and Kothari, Robin and Magniez, Fr{\'e}d{\'e}ric},
  booktitle={Proceedings of the twenty-fourth annual ACM-SIAM symposium on Discrete algorithms},
  pages={1474--1485},
  year={2013},
  organization={SIAM}
}

@article{jeffery2022quantum,
  title={Quantum subroutine composition},
  author={Jeffery, Stacey},
  journal={arXiv preprint arXiv:2209.14146},
  year={2022}
}

@article{cornelissen2025quantum,
  title={Quantum walks through generalized graph composition},
  author={Cornelissen, Arjan},
  journal={arXiv preprint arXiv:2510.04973},
  year={2025}
}

@article{bondy1974cycles,
title = {Cycles of even length in graphs},
journal = {Journal of Combinatorial Theory, Series B},
volume = {16},
number = {2},
pages = {97-105},
year = {1974},
author = {Bondy, J. Adrian and Simonovits, Mikl\'{o}s},
}

@article{balodis2016quantum,
  title={Quantum Lower Bound for Graph Collision Implies Lower Bound for Triangle Detection.},
  author={Balodis, Kaspars and Iraids, J{\=a}nis},
  journal={Baltic Journal of Modern Computing},
  volume={4},
  number={4},
  year={2016}
}

@article{magniez2007quantum,
  title={Quantum algorithms for the triangle problem},
  author={Magniez, Fr{\'e}d{\'e}ric and Santha, Miklos and Szegedy, Mario},
  journal={SIAM Journal on Computing},
  volume={37},
  number={2},
  pages={413--424},
  year={2007},
  publisher={SIAM}
}

@article{brassard00,
author = {Brassard, Gilles and Hoyer, Peter and Mosca, Michele and Tapp, Alain},
year = {2000},
title = {Quantum Amplitude Amplification and Estimation},
volume = {305},
journal = {AMS Contemporary Mathematics Series}
}

@article{zhu2011quantum,
  title={Quantum query complexity of subgraph containment with constant-sized certificates},
  author={Zhu, Yechao},
  journal={arXiv preprint arXiv:1109.4165},
  year={2011}
}

@inproceedings{belovs2012span,
  title={Span programs and quantum algorithms for st-connectivity and claw detection},
  author={Belovs, Aleksandrs and Reichardt, Ben W},
  booktitle={European Symposium on Algorithms},
  pages={193--204},
  year={2012},
  organization={Springer}
}

@article{ambainis2007quantum,
  title={Quantum walk algorithm for element distinctness},
  author={Ambainis, Andris},
  journal={SIAM Journal on Computing},
  volume={37},
  number={1},
  pages={210--239},
  year={2007},
  publisher={SIAM}
}

@article{zhang2005power,
  title={On the power of Ambainis lower bounds},
  author={Zhang, Shengyu},
  journal={Theoretical Computer Science},
  volume={339},
  number={2-3},
  pages={241--256},
  year={2005},
  publisher={Elsevier}
}

@article{vspalek2006all,
  title={All Quantum Adversary Methods are Equivalent},
  author={{\v{S}}palek, Robert and Szegedy, Mario},
  journal={Theory of Computing},
  volume={2},
  number={1},
  pages={1--18},
  year={2006},
  publisher={Theory of Computing Exchange}
}

@article{cade2018time,
  title={Time and space efficient quantum algorithms for detecting cycles and testing bipartiteness},
  author={Cade, Chris and Montanaro, Ashley and Belovs, Aleksandrs},
  journal={Quantum Information and Computation},
  volume={18},
  number={1-2},
  pages={18--50},
  year={2018},
  publisher={Sciendo}
}

@inproceedings{magniez2007search,
  title={Search via quantum walk},
  author={Magniez, Fr{\'e}d{\'e}ric and Nayak, Ashwin and Roland, J{\'e}r{\'e}mie and Santha, Miklos},
  booktitle={Proceedings of the thirty-ninth annual ACM symposium on Theory of computing},
  pages={575--584},
  year={2007}
}

@article{bentley1980general,
  title={A general method for solving divide-and-conquer recurrences},
  author={Bentley, Jon Louis and Haken, Dorothea and Saxe, James B},
  journal={ACM SIGACT News},
  volume={12},
  number={3},
  pages={36--44},
  year={1980},
  publisher={ACM New York, NY, USA}
}

@article{durr2006graph,
author = {D\"{u}rr, Christoph and Heiligman, Mark and Hoyer, Peter and Mhalla, Mehdi},
title = {Quantum Query Complexity of Some Graph Problems},
journal = {SIAM Journal on Computing},
volume = {35},
number = {6},
pages = {1310-1328},
year = {2006}
}

@inproceedings{buhrman01elementdistinctness,
author = {Buhrman, Harry and de Wolf, Ronald and D\"{u}rr, Christoph and Heiligman, Mark and H"yer, Peter and Magniez, Fr\'{e}d\'{e}ric and Santha, Miklos},
title = {Quantum Algorithms for Element Distinctness},
booktitle = {Proceedings of the 16th Annual Conference on Computational Complexity},
year = {2001}
}

@inproceedings{hoyer07adversary,
author = {Hoyer, Peter and Lee, Troy and Spalek, Robert},
title = {Negative weights make adversaries stronger},
year = {2007},
booktitle = {Proceedings of the Thirty-Ninth Annual ACM Symposium on Theory of Computing},
pages = {526–535}
}

@inproceedings{beals98polynomial,
author = {Beals, Robert and Buhrman, Harry and Cleve, Richard and Mosca, Michele and de Wolf, Ronald},
title = {Quantum Lower Bounds by Polynomials},
year = {1998},
booktitle = {Proceedings of the 39th Annual Symposium on Foundations of Computer Science}
}

@inproceedings{childs03weldedtree,
author = {Childs, Andrew M. and Cleve, Richard and Deotto, Enrico and Farhi, Edward and Gutmann, Sam and Spielman, Daniel A.},
title = {Exponential algorithmic speedup by a quantum walk},
year = {2003},
booktitle = {Proceedings of the Thirty-Fifth Annual ACM Symposium on Theory of Computing}
}

@InProceedings{childs23weldedtree,
  author =	{Childs, Andrew M. and Coudron, Matthew and Gilani, Amin Shiraz},
  title =	{{Quantum Algorithms and the Power of Forgetting}},
  booktitle =	{14th Innovations in Theoretical Computer Science Conference (ITCS 2023)},
  year =	{2023}
}

@book{nielsen2010quantum,
  title={Quantum computation and quantum information},
  author={Nielsen, Michael A and Chuang, Isaac L},
  year={2010},
  publisher={Cambridge university press}
}

@inproceedings{reichardt2011reflections,
  title={Reflections for quantum query algorithms},
  author={Reichardt, Ben W},
  booktitle={Proceedings of the Twenty-Second Annual ACM-SIAM Symposium on Discrete Algorithms},
  pages={560--569},
  year={2011},
  organization={SIAM}
}

@article{alon1995color,
  title={Color-coding},
  author={Alon, Noga and Yuster, Raphael and Zwick, Uri},
  journal={Journal of the ACM (JACM)},
  volume={42},
  number={4},
  pages={844--856},
  year={1995},
  publisher={ACM New York, NY, USA}
}

@article{belovs2014power,
  title={On the power of non-adaptive learning graphs},
  author={Belovs, Aleksandrs and Rosmanis, Ansis},
  journal={computational complexity},
  volume={23},
  number={2},
  pages={323--354},
  year={2014},
  publisher={Springer}
}

@article{beigi2020quantum,
  title={Quantum speedup based on classical decision trees},
  author={Beigi, Salman and Taghavi, Leila},
  journal={Quantum},
  volume={4},
  pages={241},
  year={2020},
  publisher={Verein zur F{\"o}rderung des Open Access Publizierens in den Quantenwissenschaften}
}

@inproceedings{lin2015upper,
  title={Upper bounds on quantum query complexity inspired by the elitzur-vaidman bomb tester},
  author={Lin, Cedric Yen-Yu and Lin, Han-Hsuan},
  booktitle={Proceedings of the 30th Conference on Computational Complexity},
  pages={537--566},
  year={2015}
}

@book{brouwer2011spectra,
  title={Spectra of graphs},
  author={Brouwer, Andries E and Haemers, Willem H},
  year={2011},
  publisher={Springer Science \& Business Media}
}
\end{document}